\documentclass[acmsmall,screen]{acmart}

\setcopyright{rightsretained}
\acmDOI{10.1145/3632896}
\acmYear{2024}
\copyrightyear{2024}
\acmSubmissionID{popl24main-p261-p}
\acmJournal{PACMPL}
\acmVolume{8}
\acmNumber{POPL}
\acmArticle{54}
\acmMonth{1}
\received{2023-07-11}
\received[accepted]{2023-11-07}

\usepackage{booktabs}   %
\usepackage{subcaption} %

\usepackage[utf8]{inputenc}            %
\usepackage[T1]{fontenc}               %
\usepackage{amsmath,mathtools,amsthm}  %
\usepackage{graphicx}
\usepackage{thmtools,thm-restate}      %
\usepackage{mathpartir}                %
\usepackage{stmaryrd}                  %
\usepackage{xspace}                    %

\usepackage{url}
\usepackage[hyphenbreaks]{breakurl}

\usepackage{xcolor}                    %
\usepackage{tikz}                      %
\usetikzlibrary{trees}
\usepackage[ruled, noend]{algorithm2e}
\usepackage[normalem]{ulem}
\usepackage{listings}
\usepackage{enumerate}                 %
\usepackage{enumitem}
\usepackage{cleveref}                  %
\usepackage{placeins}         %
\usepackage{threeparttable}
\usepackage{fontawesome}               %
\usepackage{pifont}                    %
\usepackage[scaled=0.85]{beramono}

\DeclareMathAlphabet{\mathbsf}{\encodingdefault}{\sfdefault}{bx}{n}

\definecolor{codegreen}{rgb}{0,0.6,0}
\definecolor{codeblue}{rgb}{0,0,0.8}
\definecolor{codegray}{rgb}{0.5,0.5,0.5}
\definecolor{codepurple}{rgb}{0.58,0,0.82}
\definecolor{backcolour}{rgb}{0.95,0.95,0.92}

\lstdefinestyle{mystyle}{
    commentstyle=\color{codegreen},
    keywordstyle=\color{codeblue},
    numberstyle=\tiny\color{codegray},
    stringstyle=\color{codepurple},
    basicstyle=\ttfamily\small,
    breakatwhitespace=false,
    breaklines=true,
    captionpos=b,
    keepspaces=true,
    numbers=left,
    numbersep=5pt,
    showspaces=false,
    showstringspaces=false,
    showtabs=false,
    tabsize=2,
    columns=fullflexible,
    morekeywords={
      Unl, Lin, Any, linfun, fun, let, in, do, return,
      linlet, lindo},
    literate={
    {->}{{$\to$}}{2}
    {<-}{{$\revto$}}{2}
    {-@}{{$\multimap$}}{2}
    {<|}{{$\langle$}}{1}
    {|>}{{$\rangle$}}{1}
    {[|}{{$\llbracket$}}{1}
    {|]}{{$\rrbracket$}}{1}
  }
}

\def\Multimap{\hbox{$=$}\kern -0.5pt\hbox{$\circ$}}

\lstdefinestyle{links}{
    commentstyle=\textit,
    numberstyle=\tiny,
    stringstyle=\ttfamily\small,
    basicstyle=\ttfamily\small,
    breakatwhitespace=false,
    breaklines=true,
    captionpos=b,
    keepspaces=true,
    numbers=none,
    numbersep=5pt,
    showspaces=false,
    showstringspaces=false,
    showtabs=false,
    tabsize=2,
    columns=fullflexible,
    keywords=[1]{
      linfun, let, linlet, in, lindo, xlin,
      fun, do, handle, sig, switch, case, var, return, forall, if, else
    },
    keywordstyle=[1]\bfseries,
    keywords=[2]{
      Unl, Lin, Any
    },
    keywordstyle=[2],
    morecomment = [l]{\#},
}

\newcommand{\Calc}{\ensuremath{\mathrm{F}^\Lin_{\mathrm{eff}}}\xspace}
\newcommand{\CalcS}{\ensuremath{\mathrm{Q}^{\Lin}_{\mathrm{eff}}}\xspace}

\newcommand{\CalcQ}{\ensuremath{\mathrm{Q}^\Lin_{\mathrm{eff}}}\xspace} %

\newcommand{\Links}{\textsc{Links}\xspace}

\newcommand{\Quill}{\textsc{Quill}\xspace}
\newcommand{\Rose}{\textsc{Rose}\xspace}
\newcommand{\Eff}{\textsc{Eff}\xspace}
\newcommand{\Koka}{\textsc{Koka}\xspace}
\newcommand{\Effekt}{\textsc{Effekt}\xspace}
\newcommand{\Frank}{\textsc{Frank}\xspace}
\newcommand{\Helium}{\textsc{Helium}\xspace}

\newcommand{\Fpop}{\ensuremath{\mathrm{F}^\Lin}\xspace}

\newcommand{\Rplus}{\protect\hspace{-.05em}\protect\raisebox{.15ex}{\texttt{+}}}
\newcommand{\Cpp}{\textsc{C}\Rplus\Rplus\xspace}
\newcommand{\Go}{\textsc{Go}}
\newcommand{\Java}{\textsc{Java}}
\newcommand{\Scheme}{\textsc{Scheme}}
\newcommand{\OCaml}{\textsc{OCaml}}
\newcommand{\Wasm}{\textsc{WebAssembly}}

\colorlet{hlcolor}{lightgray!50}
\newcommand{\hlbox}[1]{{\colorbox{hlcolor}{$#1$}}}

\newcommand{\highlightwithstyle}[2]{
  {\setlength{\fboxsep}{1pt} \colorbox{hlcolor}{\pgfsetfillopacity{1}$#1#2$}}
}
\newcommand{\hl}[1]{\mathpalette\highlightwithstyle{#1}}

\newcommand{\kind}[2]{{#1} ^ {#2}}
\newcommand{\sto}{\twoheadrightarrow}

\newcommand{\Belong}[2]{#2 \ni #1}

\makeatletter
\newcommand{\ogeneric}[2][0.7]{%
  \vphantom{{\oplus}}\mathpalette\o@generic{{#1}{#2}}%
}
\newcommand{\o@generic}[2]{\o@@generic#1#2}
\newcommand{\o@@generic}[3]{%
  \begingroup
  \sbox\z@{$\m@th#1\oplus$}%
  \dimen@=\dimexpr\ht\z@+\dp\z@\relax
  \savebox\tw@[\totalheight]{$\m@th#1\bigcirc$}%
  \makebox[\wd\z@]{%
    \ooalign{%
      $#1\vcenter{\hbox{\resizebox{\dimen@}{!}{\usebox\tw@}}}$\cr
      \hidewidth
      $#1\vcenter{\hbox{\resizebox{#2\dimen@}{!}{$#1\vphantom{\oplus}{#3}$}}}$%
      \hidewidth
      \cr
    }%
  }%
  \endgroup
}
\makeatother

\newcommand{\meta}{\mathsf}

\newcommand{\lsig}{:}
\newcommand{\lsep}{{\,;\,}}
\newcommand{\linle}{\preceq}

\newcommand{\rle}{{\,\leqslant\,}}

\newcommand{\rlack}{{\,\bot\,}}

\newcommand{\un}{\linle \Unl}
\newcommand{\subtype}{\leqslant}
\newcommand{\cR}{K}

\newcommand{\qto}{\Rightarrow}
\newcommand{\fresh}{\;\mathit{fresh}}
\newcommand{\ol}[1]{\overline{#1}}
\newcommand{\setcomplement}[1]{{#1^\mathtt{c}}}
\newcommand{\satsubst}[1]{\llbracket{#1}\rrbracket_{sat}}

\newcommand{\munify}{\meta{unify}}

\newcommand{\solve}{\meta{solve}}

\newcommand{\solverow}{\meta{srow}}
\newcommand{\unifyrow}{\meta{urow}}
\newcommand{\unifylab}{\meta{ulab}}
\newcommand{\unifylin}{\meta{ulin}}
\newcommand{\solvelin}{\meta{trlin}}

\newcommand{\sqgen}{\sqsubseteq}

\newcommand{\mL}{\mathcal{L}}

\newcommand{\mS}{\mathcal{S}}
\newcommand{\mT}{\mathcal{T}}

\newcommand{\rowset}[1]{\mathsf{set}(#1)}
\newcommand{\factor}{\mathsf{factorise}}

\newcommand{\mlet}{\text{let}}
\newcommand{\mreturn}{\text{return}}
\newcommand{\massert}{\text{assert}}
\newcommand{\mfresh}{\text{assume fresh}}
\newcommand{\mfail}{\text{fail}}

\newcommand{\Lin}{\circ}
\newcommand{\Unl}{\bullet}

\newcommand{\islin}{\meta{lin}}

\newcommand{\lintag}[1]{{#1}^\Lin}
\newcommand{\taglin}{\meta{tag}}

\newcommand{\ftv}[1]{\meta{ftv}(#1)}

\newcommand{\dom}[1]{\meta{dom}(#1)}
\newcommand{\BL}[1]{\mathsf{bl}(#1)}

\newcommand{\gen}{\mathsf{gen}}

\newcommand{\makeunl}[1]{\mathsf{un}(#1)} %
\newcommand{\makeleq}[1]{\mathsf{leq}(#1)}
\newcommand{\makesub}[1]{\mathsf{sub}(#1)}
\newcommand{\makelack}[1]{\mathsf{lack}(#1)}

\renewcommand{\infer}[6]{#1 \vdash #2 : #3 \dashv #4, #5, #6}
\renewcommand{\infer}[6]{#1 \vdash #2 : #3 \dashv #4, #5, #6}
\newcommand{\unify}[3]{#1 \sim #2 : #3}

\newcommand{\umul}{\cup}

\newcommand{\refa}[1]{{\color{red}    {#1}\ifthenelse{\equal{#1}{}}{}{\,}(1)}}
\newcommand{\refb}[1]{{\color{blue}   {#1}\ifthenelse{\equal{#1}{}}{}{\,}(2)}}
\newcommand{\refc}[1]{{\color{violet} {#1}\ifthenelse{\equal{#1}{}}{}{\,}(3)}}
\newcommand{\refd}[1]{{\color{purple} {#1}\ifthenelse{\equal{#1}{}}{}{\,}(4)}}
\newcommand{\refe}[1]{{\color{cyan}   {#1}\ifthenelse{\equal{#1}{}}{}{\,}(5)}}
\newcommand{\reff}[1]{{\color{magenta}{#1}\ifthenelse{\equal{#1}{}}{}{\,}(6)}}
\newcommand{\refg}[1]{{\color{brown} {#1}\ifthenelse{\equal{#1}{}}{}{\,}(7)}}
\newcommand{\refh}[1]{{\color{orange}  {#1}\ifthenelse{\equal{#1}{}}{}{\,}(8)}}

\makeatletter
\DeclareRobustCommand{\Circle}{%
  \mathbin{\mathpalette\on@ntimes\relax}%
}
\newcommand{\on@ntimes}[2]{%
  \vcenter{\hbox{%
    \sbox0{\m@th$#1\otimes$}%
    \setlength\unitlength{\wd0}%
    \begin{picture}(1,1)
    \linethickness{0.35pt}
    \put(.5,.5){\circle{.8}}
    \end{picture}%
  }}%
}
\makeatother

\newcommand{\xmark}{\text{\ding{55}}} %
\newcommand{\tmark}{\text{\ding{51}}} %

\newcommand{\notsmall}{}

\newcommand{\slab}[1]{\textrm{#1}}
\newcommand{\semlab}[1]{\text{\scshape{E-#1}}}
\newcommand{\lsemlab}[1]{\text{\scshape{L-#1}}}

\newcommand{\tylab}[1]{\text{\scshape{T-#1}}}
\newcommand{\qtylab}[1]{\text{\scshape{Q-#1}}}

\newcommand{\tilab}[1]{\text{\scshape{Q-#1}$^\text{W}$}}

\newcommand{\sklab}[1]{\text{\scshape{S-#1}}}
\newcommand{\klab}[1]{\text{\scshape{K-#1}}}

\newcommand{\plab}[1]{\text{\scshape{P-#1}}}
\newcommand{\ulab}[1]{\text{\scshape{U-#1}}}

\newcommand{\revto}{\ensuremath{\leftarrow}}

\newcommand{\dec}[1]{\mathsf{#1}}

\newcommand{\keyw}[1]{{{\mathbsf{#1}}}}
\newcommand{\Handle}{\keyw{handle}\;}
\newcommand{\With}{\;\keyw{with}\;}
\newcommand{\Let}{\keyw{let}\;}
\newcommand{\Lety}{\keyw{let}}
\newcommand{\Seq}{\,;}
\newcommand{\Rec}{\keyw{rec}\;}
\newcommand{\In}{\;\keyw{in}\;}
\newcommand{\Do}{\keyw{do}\;}

\newcommand{\Ret}{\keyw{return}\;}

\renewcommand{\If}{\keyw{if}\;}
\newcommand{\Then}{\;\keyw{then}\;}
\renewcommand{\Else}{\;\keyw{else}\;}

\newcommand{\File}{\mathsf{File}}
\newcommand{\OpenFile}{\mathsf{open}}
\newcommand{\CloseFile}{\mathsf{close}}

\newcommand{\WriteFile}{\mathsf{write}}

\newcommand{\Abs}{\mathsf{Abs}}
\newcommand{\Presence}{\mathsf{Presence}}
\newcommand{\Row}{\mathsf{Row}}
\newcommand{\Type}{\mathsf{Type}}

\newcommand{\Comp}{\mathsf{Comp}}
\newcommand{\Effect}{\mathsf{Effect}}
\newcommand{\Handler}{\mathsf{Handler}}

\newcommand{\Unit}{()}
\newcommand{\TUnit}{\Unit}

\newcommand{\Int}{\mathsf{Int}}
\newcommand{\Bool}{\mathsf{Bool}}
\newcommand{\String}{\mathsf{String}}
\newcommand{\makestring}[1]{\text{\textquotedbl #1\textquotedbl}}

\newcommand{\True}{\mathsf{true}}
\newcommand{\False}{\mathsf{false}}

\newcommand{\Choose}{\dec{Choose}}

\newcommand{\Get}{\mathsf{Get}}
\newcommand{\Print}{\mathsf{Print}}

\newcommand{\Close}{\mathsf{Close}}

\newcommand{\eff}{\, !\,}
\newcommand{\typ}[2]{#1 \vdash #2}
\newcommand{\typc}[3]{\typ{#1}{#2 \eff #3}}

\newcommand{\hto}{\rightrightarrows}

\newcommand{\reducesto}{\leadsto}
\newcommand{\lreducesto}[2]{\; \prescript{#1}{#2}{\leadsto}\; }

\newcommand{\TL}[1]{\mathscr{L}(#1)}

\newcommand{\ba}{\begin{array}}
\newcommand{\ea}{\end{array}}

\newcommand{\bl}{\ba[t]{@{}l@{}}}
\newcommand{\el}{\ea}

\renewenvironment{displaymath}{\notsmall\[}{\]\normalsize\ignorespacesafterend}

\newenvironment{syntax}{\begin{displaymath}\ba{@{}l@{\quad}r@{~}c@{~}l@{}}}{\ea\end{displaymath}\ignorespacesafterend}
\newenvironment{bigreductions}{\begin{displaymath}\ba{@{}l@{\quad}@{}r@{~}r@{~}l@{}}}{\ea\end{displaymath}\ignorespacesafterend}
\newenvironment{reductions}{\begin{displaymath}\ba{@{}l@{\quad}@{}r@{~}c@{~}l@{}}}{\ea\end{displaymath}\ignorespacesafterend}

\newcommand{\EC}{\mathcal{E}}
\newcommand{\FC}{\mathcal{F}}

\begin{document}

\title{Soundly Handling Linearity}

\author{Wenhao Tang}
\orcid{0009-0000-6589-3821}
\affiliation{%
  \institution{The University of Edinburgh}
  \country{United Kingdom}}
\email{wenhao.tang@ed.ac.uk}

\author{Daniel Hillerstr{\"o}m}
\orcid{0000-0003-4730-9315}
\affiliation{%
  \institution{Huawei Zurich Research Center}
  \country{Switzerland}}
\email{daniel.hillerstrom@ed.ac.uk}

\author{Sam Lindley}
\orcid{0000-0002-1360-4714}
\affiliation{%
  \institution{The University of Edinburgh}
  \country{United Kingdom}}
\email{sam.lindley@ed.ac.uk}

\author{J.\ Garrett Morris}
\orcid{0000-0002-3992-1080}
\affiliation{%
  \institution{University of Iowa}
  \country{USA}
}
\email{garrett-morris@uiowa.edu}

\begin{abstract}

We propose a novel approach to soundly combining linear types with
multi-shot effect handlers.
Linear type systems statically ensure that resources such as file
handles and communication channels are used exactly once.
Effect handlers provide a rich modular programming abstraction for
implementing features ranging from exceptions to concurrency to
backtracking.
Whereas conventional linear type systems bake in the assumption that
continuations are invoked exactly once, effect handlers allow
continuations to be discarded (e.g. for exceptions) or invoked more
than once (e.g. for backtracking).
This mismatch leads to soundness bugs in existing systems such as the
programming language \Links, which combines linearity (for session
types) with effect handlers.
We introduce control-flow linearity as a means to ensure that
continuations are used in accordance with the linearity of any
resources they capture, ruling out such soundness bugs.

We formalise the notion of control-flow linearity in a System F-style
core calculus $\Calc$ equipped with linear types, an effect type
system, and effect handlers. We define a linearity-aware semantics in
order to formally prove that $\Calc$ preserves the integrity of linear
values in the sense that no linear value is discarded or duplicated.
In order to show that control-flow linearity can be made practical, we
adapt \Links based on the design of \Calc, in doing so fixing a
long-standing soundness bug.

Finally, to better expose the potential of control-flow linearity, we
define an ML-style core calculus $\CalcS$, based on qualified types,
which requires no programmer provided annotations, and instead relies
entirely on type inference to infer control-flow linearity.
Both linearity and effects are captured by qualified types.
$\CalcS$ overcomes a number of practical limitations of $\Calc$,
supporting abstraction over linearity, linearity dependencies between
type variables, and a much more fine-grained notion of control-flow
linearity.

\end{abstract}



\maketitle

\section{Introduction}
\label{sec:introduction}

Many programming languages support linear resources such as file
handles, communication channels, network connections, and so forth.
Special care must be taken to preserve the integrity of linear
resources in the presence of first-class continuations that may be
invoked multiple times~\cite{FriedmanH85}, as a linear resource may be
inadvertently be accessed more than once.
\Java{}~\citep{Pressler18} and \OCaml{}~\citep{multicore} have each
recently been retrofitted with facilities for programming with
first-class continuations that must be invoked exactly once, partly in
order to avoid such pitfalls.
Nonetheless, multi-shot continuations are a compelling feature,
supporting applications such as backtracking
search~\cite{FriedmanHK84} and probabilistic
programming~\citep{KiselyovS09}. In this paper we explore how to
soundly handle linearity in the presence of multi-shot effect
handlers~\citep{PlotkinP13}.

We first illustrate the issues with combining linearity with
multi-shot effect handlers by exhibiting a soundness bug in the
programming language \Links~\citep{CooperLWY06}, which is equipped
with linear session-typed channels~\citep{lindley2017lfst} and effect
handlers with multi-shot continuations~\citep{HillerstromLA20}.
We begin by defining a function \lstinline[style=links]$outch$ that
forks a child process and returns an output channel for communicating
with it.
The idea is that we will use a combination of exceptions and
multi-shot continuations to send two integers, rather than an integer
followed by a string, along the endpoint (with session type
\lstinline[style=links]$!Int.!String.End$) returned by
the function \lstinline[style=links]$outch$.
\begin{lstlisting}[style=links]
sig outch : () ~> !Int.!String.End
fun outch() {
  fork(fun(ic) {
    var (i, ic) = receive(ic);     # receive the integer
    var (s, ic) = receive(ic);     # receive the string
    println(intToString(i) ^^ s);  # convert, concat, and print
    close(ic)                      # close the input channel
  })
}
\end{lstlisting}
The primitive \lstinline[style=links]$fork$ creates a child process
and two endpoints of a session-typed channel. One endpoint is passed
to the child process and the other endpoint is returned to the caller.
Here the function returns an output endpoint of type
\lstinline[style=links]$!Int.!String.End$ and the child process is
supplied with an input endpoint of type
\lstinline[style=links]$?Int.?String.End$.
The child receives an integer and a string on the input endpoint, then
prints them out before closing the endpoint.

Now we invoke \lstinline[style=links]$outch$ in a context in which we
exploit the power of multi-shot continuations to return twice and the power of
exceptions to abort the current computation.
\begin{lstlisting}[style=links]
handle({
  var oc = outch();
  var msg = if (do Choose) 42 else 84; # choose an integer message to send
  var oc = send(msg, oc);
  do Fail;                             # this is our exception
  var oc = send("well-typed", oc);
  close(oc)
}) {
  case <Fail> -> ()
  case <Choose => resume> -> resume(true); resume(false)
}
\end{lstlisting}
We handle a computation that performs two operations: 1)
\lstinline[style=links]$Choose : () => Bool$; and 2)
\lstinline[style=links]$Fail : forall a. () => a$.
The handled computation invokes \lstinline[style=links]$outch$,
forking a child process and binding the output endpoint of the
resulting channel to \lstinline[style=links]$oc$. Next, it invokes the
operation \lstinline[style=links]$Choose$ to select between two
possible integer messages, which is sent on the channel. Then, it
performs the \lstinline[style=links]$Fail$ operation, before sending a
string along the channel and closing it.
This is all very well and satisfies the type-checker; however, the
described control flow is not actually what happens, because in fact
the continuation of \lstinline[style=links]$Choose$ is invoked twice and
the continuation of \lstinline[style=links]$Fail$ is never invoked.
The behaviours of \lstinline[style=links]$Fail$ and
\lstinline[style=links]$Choose$ are defined by the corresponding
operation clauses of the handler. For \lstinline[style=links]$Fail$
the captured continuation is discarded (it must be: it is never bound);
for \lstinline[style=links]$Choose$ the continuation is bound to
\lstinline[style=links]$resume$ and invoked twice: first with
\lstinline[style=links]$true$ and then with
\lstinline[style=links]$false$.

Running the program causes a segmentation fault when printing the
received values, as it erroneously attempts to concatenate a string
with an integer. To see why, follow the control flow of the parent
process. It performs \lstinline[style=links]$Choose$, which initially
selects \lstinline[style=links]$42$ and sends it over the channel.
The child process receives this integer and subsequently expects to
receive a string.
Back on the parent process execution is aborted via
\lstinline[style=links]$Fail$, which causes the initial invocation of
\lstinline[style=links]$resume$ to return, leading to the second
invocation of \lstinline[style=links]$resume$, which restores the
aborted context at the point of selecting an integer. Now
\lstinline[style=links]$Choose$ selects \lstinline[style=links]$84$
and sends it over the channel.
The child process receives this second integer, mistakenly treating it
as a string.

In this paper we rule out such soundness bugs by tracking
\emph{control-flow linearity}: a means to statically assure how often
a continuation may be invoked, mediating between linear resources and
effectful operations to ensure that effect handlers cannot violate
linearity constraints on resources.

The main contributions of this paper are:
\begin{itemize}
\item
  We give high-level overview of the main ideas of the paper through a
  series of worked examples that illustrate the difficulties of
  combining effect handlers with linearity, how they can be resolved
  by tracking control-flow linearity, and how the approach can be
  refined using qualified types~\citep{jones94} (Section~\ref{sec:overview}).
\item
  We introduce \Calc (pronounced ``F-eff-pop''), a System F-style core
  calculus equipped with linear types, an effect type system, and
  effect handlers (Section~\ref{sec:feffpop}). We prove syntactic type
  soundness and a semantic linear safety property.
\item
  Inspired by \Calc we implement control-flow linearity in \Links,
  fixing a long-standing type-soundness bug (Section~\ref{sec:links}).
\item
  Motivated by expressiveness limitations of \Calc we introduce \CalcS
  (pronounced ``Q-eff-pop''), an ML-style core calculus inspired by
  \Quill~\citep{quill} and \Rose~\citep{rose}, based on qualified types
  (Section~\ref{sec:qeffpopsub}). We prove soundness and completeness
  of type inference for \CalcS. Along the way, we identify a semantic
  soundness bug in \Quill and conjecture a fix.
\end{itemize}
Section~\ref{sec:extensions} outlines how control-flow linearity
applies to shallow handlers~\citep{HillerstromL18}.
Section~\ref{sec:related-work} discusses related work and
Section~\ref{sec:conclusion} conclude and discusses future work.

\section{Overview}
\label{sec:overview}

In this section, we give a high-level overview of the main ideas of
the paper by way of a series of examples.
We first compare standard value linearity with non-standard
control-flow linearity, illustrating how the latter may be tracked in
an explicit calculus \Calc (Section~\ref{sec:feffpop}).
For readability we omit uninteresting syntactic artifacts from our
examples.
We show how control-flow linearity allows linear resources and
multi-shot continuations to coexist peacefully.
We then highlight two limitations of \Calc: linear types require
syntactic overhead which harms modularity, and row-polymorphism based
effect types lead to coarse tracking of control-flow linearity.
We exploit qualified types to relax both limitations in an ML-style
calculus \CalcS (Section~\ref{sec:qeffpopsub}).

\subsection{Value Linearity}
\label{sec:value-linearity}

Value linearity classifies the \emph{use} of values: linear
values must be used exactly once whereas unlimited values can be used
zero, one, or multiple times (linear types differ from uniqueness
types, which instead track the number of references to a value).
Equivalently, value linearity characterises whether values contain
linear resources: linear values can contain linear resources whereas
unlimited values cannot.
Conventional linear type systems track value linearity.
\Calc adapts the subkinding-based linear type system of
\Fpop~\citep{fpop}.
The linearity $Y$ of a value type is part of its kind $\kind\Type Y$
and can be either linear $\Lin$ or unlimited $\Unl$.
For example, file handles are linear resources
($\File:\kind\Type\Lin$) and integers are unlimited resources
($\Int:\kind\Type\Unl$).

A linearity annotation on a $\lambda$-abstraction defines the
linearity of the function itself.
Consider the following function $\dec{faithfulWrite}$ which takes a
file handle $f$ and returns another function that takes a string $s$,
faithfully writes $s$ to $f$, and then closes the file handle.
\[\ba{l@{\;}c@{\;}l}
  \dec{faithfulWrite} &:& \File \to^\Unl (\String \to^\Lin \TUnit) \\
  \dec{faithfulWrite} &=& \lambda^\Unl f . (
    \lambda^\Lin s . \Let\, f' \revto \WriteFile\, (s, f) \In \CloseFile\, f')
\ea\]
The outer unlimited function ($\to^\Unl$) %
yields a linear function ($\to^\Lin$) expecting a string.
The linear type system dictates that the inner function is linear as
it captures the linear file handle $f$.

One important property of value linearity is that unlimited value
types can be treated as linear value types, as it is always safe to
use unlimited values (which contain no linear resources) just once.
This property is embodied by the subkinding relation
$\vdash\kind\Type\Unl\le\kind\Type\Lin$ in \Calc.
For instance, consider the polymorphic identity function.
\[\ba{l@{\;}c@{\;}l}
  \dec{id} &:& \forall \mu^{\Row}\,\alpha^{\kind\Type\Lin} .\, \alpha \to^\Unl \alpha \eff \{\mu\} \\
  \dec{id} &=& \Lambda \mu^{\Row}\,\alpha^{\kind\Type\Lin} .\, \lambda^\Unl x .\, x
\ea\]
The return type of the function is a computation type $\alpha\eff\{\mu\}$
where $\alpha$ is the linear type of values returned ($x$ is used
exactly once) and $\mu$ is the row of effects performed by the
function.
(We chose to omit the corresponding effect annotations in the
signature of $\dec{faithfulWrite}$ because they are empty, but
henceforth we will write them explicitly.)
Subkinding allows the identity function to be applied to both linear
and unlimited values. It is always sound to use an unlimited value
exactly once. Thus, we have both $\vdash\Int:\kind\Type\Lin$ and
$\vdash\File:\kind\Type\Lin$, and if $R$ is an effect row type:
\[\ba{l@{\;}c@{\;}l}
\dec{id}\;R\;\File &:& \File \to^\Unl \File \eff \{R\} \\
\dec{id}\;R\;\Int &:& \Int \to^\Unl \Int \eff \{R\}
\ea\]

\subsection{Control-Flow Linearity}
\label{sec:control-flow-linearity}

Control-flow linearity tracks how many times control may enter a local
context: a control-flow-linear context must be entered exactly once; a
control-flow-unlimited context may be entered zero, one, or multiple
times.
Equivalently, control-flow linearity characterises whether a
local context captures linear resources: a control-flow-linear context
can capture linear resources; a control-flow-unlimited context cannot.

To better explain control-flow linearity, we first reprise the
soundness problem due to the interaction of linear resources and
multi-shot continuations of \Cref{sec:introduction} via a simpler
example in \Calc.
Consider the following function $\dec{dubiousWrite}_\xmark$, which takes a
file handle and non-deterministically writes $\makestring{A}$ or
$\makestring{B}$ to it depending on the result of $\Choose$.
We ignore control-flow linearity for now.
\[\ba{l@{\;}c@{\;}l@{\;}l}
  \dec{dubiousWrite}_\xmark &:& \File \to^\Unl \TUnit\eff\{\Choose:\TUnit\sto\Bool\} \span\\
  \dec{dubiousWrite}_\xmark &=& \lambda^\Unl f . \span\\
    \span\span\quad&\Let b \revto (\Do\Choose\,\Unit)^{\{\Choose:\TUnit\sto\Bool\}} \In \\
    \span\span&
    \left.\ba{@{}l@{}}
    \Let s \revto \If b \Then \makestring{A} \Else\makestring{B} \In \\
    \Let f' \revto \WriteFile\,(s,f)\In \CloseFile\, f'
    \ea\quad\right\} \text{continuation of $\Choose$}
\ea\]
The $\Do \Choose\,\Unit$ expression invokes operation $\Choose$ with a
unit argument.
\Calc adapts an effect system based on \citeauthor{remy1994type}-style
row polymorphism \citep{LindleyC12, linksrow}.
Effect types in \Calc are rows containing operation labels with their
signatures and ended with potential row variables.
The effect type $\{\Choose : \TUnit\sto\Bool\}$ denotes that
$\dec{dubiousWrite}_\xmark$ may invoke the operation $\Choose$, which
takes a unit and returns a boolean value as indicated by its signature
$\TUnit\sto\Bool$.
The problem arises when we handle $\Choose$ using multi-shot
continuations.
\[
\Let f \revto \OpenFile\,\makestring{C.txt} \In
\Handle (\dec{dubiousWrite}_\xmark\;f)
\With \{\Choose\;\_\;r\mapsto r\;\True\Seq r\;\False\}
\]
The file $\makestring{C.txt}$ is opened and the file handle is bound
to $f$ before $\dec{dubiousWrite}_\xmark\;f$ is handled by an effect
handler that handles the $\Choose$ operation.
In the handler clause, $r$ binds the continuation of $\Choose$, which
expects a parameter of type $\Bool$.
As $r$ is invoked twice (first with $\True$ and then with $\False$),
the file handle $f$ is written and closed twice, which leads to a
runtime error because it is closed before the second write.
The essential problem is that the continuation of $\Choose$ should be
used linearly as it captures the linear file handle $f$, but it is
invoked twice by the effect handler.
Conventional linear type systems cannot detect this kind of error as
they only track value linearity.

Motivated by the observation that only a local context, reified as the
continuation of an operation, may be captured by a multi-shot handler,
we track control-flow linearity at the granularity of operations.
We use the control-flow linearity of an operation to represent the
control-flow linearity of the continuation of the operation.
Control-flow-linear operations can be used in contexts which may
contain linear resources, whereas control-flow-unlimited operations
cannot.
An operation signature $A \sto^Y B$ is annotated with a linearity $Y$
to denote its control-flow linearity.
The $\dec{dubiousWrite}_\xmark$ function can now be rewritten to
correctly track control-flow linearity as follows.
\[\ba{l@{\;}c@{\;}l@{\;}l}
  \dec{dubiousWrite}_\tmark &:& \File \to^\Unl \TUnit\eff\{\Choose:\TUnit\sto^\Lin\Bool\} \span\\
  \dec{dubiousWrite}_\tmark &=& \lambda^\Unl f . \span\\
    \span\span\quad&\Lety^\Lin b \revto (\Do\Choose\,\Unit)^{\{\Choose:\TUnit\sto^\Lin\Bool\}} \In \\
    \span\span&
    \left.\ba{@{}l@{}}
    \Lety^\Lin s \revto \If b \Then \makestring{A} \Else \makestring{B} \In \\
    \Lety^\Lin f' \revto \WriteFile\,(s,f)\In \CloseFile\, f'
    \ea\quad\right\} \text{continuation of $\Choose$}
\ea\]
Now, the type of $\dec{dubiousWrite}_\tmark$ specifies that the
operation $\Choose:\TUnit\sto^\Lin\Bool$ is control-flow linear
(i.e. the continuation of $\Choose$ is linear).
We also annotate let-bindings with linearity information.
In $\Lety^Y x \revto M \In N$, the term $N$ has control-flow linearity
$Y$, and in particular the $\Lin$ annotations on the let-bindings in
$\dec{dubiousWrite}_\tmark$ permit the use of the linear file handle
throughout.

The linear type system of \Calc uses the control-flow linearity of
operations to restrict the use of continuations in handlers, which
ensures that control-flow-linear contexts are entered only once.
For instance, consider the handling of $\dec{dubiousWrite}_\tmark$ with
the same multi-shot handler.
\[
\Let f \revto \OpenFile\,\makestring{C.txt} \In
\Handle (\dec{dubiousWrite}_\tmark\;f)
\With \{\Choose\;\_\;r\mapsto r\;\True\Seq r\;\False\}
\]
This is ill-typed due to the fact that $\Choose$ is control-flow
linear, which means the resumption $r$ has a linear function type,
meaning it must be applied exactly once.

We lift the control-flow linearity of operations to effect row types
and reflect it in their kinds $\kind\Row Y$.
Similar to value linearity, we also have a subkinding relation for
control-flow linearity.
Recall that the control-flow linearity of (the operations in) effect
row types is actually the control-flow linearity of their contexts,
not themselves.
This induces a duality between value linearity and control-flow
linearity paralleling the duality between positive values and negative
continuations.
As a consequence, the subkinding relation for control-flow linearity
is $\vdash\kind\Row\Lin\le\kind\Row\Unl$, the reverse of that for
value linearity.
Intuitively, this says that control-flow-linear operations can be
treated as control-flow-unlimited operations, because it is safe to
use control-flow-linear operations in unlimited contexts.
For example, consider the following function $\dec{tossCoin}$ which
takes a function that returns a boolean and tosses a coin using this
function.
\[\ba{l@{\;}c@{\;}l@{\;}l}
  \dec{tossCoin} &:& \forall \mu^{\kind\Row\Unl} .
    (\TUnit\to^\Unl\Bool\eff\{\mu\}) \to^\Unl \String\eff\{\mu\} \span\\
  \dec{tossCoin} &=& \Lambda \mu^{\kind\Row\Unl} .
    \lambda^\Unl g .
    &\Lety^\Unl\, b \revto g\,\Unit \In \If b \Then \makestring{heads} \Else \makestring{tails}
\ea\]
As no linear resource is used, the effect type of $\dec{tossCoin}$ and
its parameter is given by a control-flow-unlimited row variable
$\mu:\kind\Row\Unl$.
Via subkinding, we can instantiate $\mu$ with operations with either
control-flow linearity.
For instance, suppose we have $\vdash R_1:\kind\Row\Unl$ and $\vdash
R_2:\kind\Row\Lin$ for $R_1 = \Choose:()\sto^\Unl\Bool$ and $R_2 =
\Choose:()\sto^\Lin\Bool$, then:
\[\ba{l}%
\dec{tossCoin}\; R_1\; (\lambda^\Unl \Unit .  (\Do\,\Choose\,\Unit)^{\{R_1\}})
: \String\eff\{R_1\} \\
\dec{tossCoin}\; R_2\; (\lambda^\Unl \Unit .  (\Do\,\Choose\,\Unit)^{\{R_2\}})
: \String\eff\{R_2\} \\
\ea\]

The subkinding relation of control-flow linearity only
influences how operations are used, not how they are handled.
We can \emph{use} control-flow-linear operations as
control-flow-unlimited operations (i.e., use them in unlimited
contexts), but this does not imply that we can \emph{handle}
control-flow-linear operations as control-flow-unlimited operations
(i.e., handle them by resuming any number of times).
Our linear type system does not allow control-flow-linear operations
to be handled by multi-shot handlers despite the subkinding relation
$\kind\Row\Lin\le\kind\Row\Unl$.
This is because when handling, we directly look at the control-flow
linearity on operation signatures instead of their kinds, where no
$\sto^\Lin$ can be upcast to $\sto^\Unl$.
This can be seen more clearly from the typing rules in
\Cref{sec:feffpop-typing}.
We formally state the soundness of \Calc in
\Cref{sec:metatheory,sec:feffpop-linear-semantics}.

\subsection{Qualified Linear Types}
\label{sec:qualified-linear-types}

As we have seen from the examples so far, \Calc requires linearity
annotations on $\lambda$-abstractions and let-bindings.
Though this can suffice for an explicit calculus, it can prove
cumbersome for practical programming languages and curtail the
modularity of programs.
Unfortunately, we cannot entirely overcome these limitations by
introducing subsumption relations between types, or using
Hindley-Milner type inference to infer them.
The reason is that there are inner dependencies on the linearity.
For instance, consider the following function $\dec{verboseId}$ which
is almost the same as the function $\dec{id}$ in
\Cref{sec:value-linearity} but outputs the log message $\makestring{id
  is called}$ using the operation $\Print:\String\sto\TUnit$ before
returning.
\[\ba{l@{\;}c@{\;}l}
  \dec{verboseId} &:& \forall \mu^{\kind\Row {Y_1}}\,\alpha^{\kind\Type{Y_2}} .\,
    \alpha \to^{Y_0} \alpha \eff \{\Print : \String\sto^{Y_3}\TUnit \lsep\mu\} \\
  \dec{verboseId} &=& \Lambda \mu^{\kind\Row {Y_1}}\,\alpha^{\kind\Type{Y_2}} .\,
    \lambda^{Y_0} x .\, \Lety^{Y_4}\, \Unit \revto \Do \Print\; \makestring{id is called} \In x
\ea\]

Depending on different choices of $Y_0$, $Y_1$, $Y_2$, $Y_3$, and
$Y_4$, we can give ten well typed variations of $\dec{verboseId}$.
Their types are shown as follows, omitting primary kinds and
signatures for readability.
\begin{equation*}\small
  \begin{split}
    \forall \mu^{\Unl} \, \alpha^{\Unl} . \alpha \to^\Unl \alpha \eff \{\Print : \Unl\lsep \mu\} \\
    \forall \mu^{\Unl} \, \alpha^{\Unl} . \alpha \to^\Unl \alpha \eff \{\Print : \Lin\lsep \mu\} \\
    \forall \mu^{\Lin} \, \alpha^{\Unl} . \alpha \to^\Unl \alpha \eff \{\Print : \Unl\lsep \mu\} \\
    \forall \mu^{\Lin} \, \alpha^{\Unl} . \alpha \to^\Unl \alpha \eff \{\Print : \Lin\lsep \mu\} \\
    \forall \mu^{\Lin} \, \alpha^{\Lin} . \alpha \to^\Unl \alpha \eff \{\Print : \Lin\lsep \mu\} \\
  \end{split}
  \quad \quad
  \begin{split}
    \forall \mu^{\Unl} \, \alpha^{\Unl} . \alpha \to^\Lin \alpha \eff \{\Print : \Unl\lsep \mu\} \\
    \forall \mu^{\Unl} \, \alpha^{\Unl} . \alpha \to^\Lin \alpha \eff \{\Print : \Lin\lsep \mu\} \\
    \forall \mu^{\Lin} \, \alpha^{\Unl} . \alpha \to^\Lin \alpha \eff \{\Print : \Unl\lsep \mu\} \\
    \forall \mu^{\Lin} \, \alpha^{\Unl} . \alpha \to^\Lin \alpha \eff \{\Print : \Lin\lsep \mu\} \\
    \forall \mu^{\Lin} \, \alpha^{\Lin} . \alpha \to^\Lin \alpha \eff \{\Print : \Lin\lsep \mu\} \\
  \end{split}
\end{equation*}
The key observation is that the control-flow linearity of the
operation $\Print$ (as well as the row variable $\mu$) depends on the
value linearity of the parameter type $\alpha$, because the parameter
$x$ is used in the continuation of $\Print$.
To express this kind of dependency, we use a linear type system based
on qualified types inspired by \Quill~\citep{quill}.
In the ML-style calculus \CalcS with qualified linear types,
$\dec{verboseId}$ can be written and ascribed a principal type as
follows.
\[\ba{l@{\;}c@{\;}l}
  \dec{verboseId} &:& \forall \alpha\, \mu\, \phi\, \phi' .\,
    (\alpha\linle\phi) \qto
    \alpha \to^{\phi'} \alpha \eff \{\Print : \phi ; \mu\} \\
  \dec{verboseId} &=& \lambda x .\, \Do \Print\; \makestring{42}\, ;\, x
\ea\]
The linearity variables $\phi$ and $\phi'$ quantify over $\Lin$ and $\Unl$.
We do not use kinds to represent linearity of type variables; instead,
all linearity information is represented using predicates of the form
$\tau\linle\tau'$, where $\tau$ is a value type, row type or linearity
type ($\Lin,\Unl$ or a linearity variable).
The type scheme of $\dec{verboseId}$ is extended with the predicate
$\alpha\linle\phi$, meaning that the value linearity of $\alpha$ is
less than that of $\phi$, which is the control-flow linearity of
$\Print$.
This type scheme succinctly expresses all ten possibilities listed
above.
The type inference algorithm of \CalcS
(Section~\ref{sec:type-inference}) infers all such linearity
dependency constraints without the need for any type, effect, or
linearity annotations.

\subsection{Qualified Effect Types}
\label{sec:qualified-effect-types}

In addition to the syntactic overhead of linear types, the row-based
effect system of \Calc is also not entirely satisfying when tracking
control-flow linearity.
Row-based effect systems
have demonstrated their practicality in research languages such as
\Links~\citep{linksrow}, \Koka~\citep{koka}, and \Frank~\citep{frank}.
In such effect systems, sequenced computations must have the same
effect type, which can be smoothly realised by unification in systems
based on Hindley-Milner type inference.
However, though fixing effect types between sequenced computations is
often acceptable, it does introduce some imprecision, and this can
become more pronounced when control-flow linearity is brought into the
mix.

To see the problem concretely in \Calc, consider the following
function $\dec{verboseClose}$ which takes a file handle, reads a
string using the operation $\Get:\TUnit\sto\String$, closes the file
handle, and outputs the string using the operation
$\Print:\String\sto\TUnit$.
\[\ba{l@{\;}c@{\;}l@{\;}l}
  \dec{verboseClose} &:& \File \to^\Unl \TUnit\eff\{R\} \span\\
  \dec{verboseClose} &=& \lambda^\Unl f .
  &\Lety^\Lin s\revto (\Do \Get\, \TUnit)^{\{R_1\}}
  \In \Lety^\Unl \Unit\revto \CloseFile\, f
  \In (\Do \Print\, s)^{\{R_2\}}
\ea\]

Note that the second $\Lety$-binding does not need to be annotated as
linear, because the linear resource $f$ does not appear after it.
The linear resource $f$ also does not appear in the continuation of
$\Print$.
Since $R_1$, $R_2$, and $R$ should be equal in the row-based effect
system of \Calc, omitting the full operation signatures for
simplicity, we could write $R = R_1 = R_2 = \{\Get : \Lin, \Print :
\Unl\}$ in the ideal case.
However, this is actually ill-typed because all operations in $R_1$
should be control-flow linear, as the linear resource $f$ is used in
their continuations.

An intuitive way to relax this limitation of \Calc is to introduce a
trivial subtyping relation on concrete effect row types.
We say $R_1$ is a subtype of $R_2$, if all operation labels in $R_1$
are also in $R_2$ with the same signatures, and when $R_1$ ends with a
row variable, $R_2$ must end with the same row variable.
Then, in the $\dec{verboseClose}$ example, we can write $R_1 = \{\Get
: \Lin\}$, $R_2 = \{\Print : \Unl\}$, and $R =
\{\Get:\Lin,\Print:\Unl\}$, which are safe given that $R_1$ and $R_2$
are both subtypes of $R$.

We call the subtyping relation trivial because it does not allow
subtyping between row variables; an open row $R_1$ is a subtype of
$R_2$ only if $R_2$ contains the same row variable as $R_1$.
For the above $\dec{verboseClose}$ example this works, but for other
functions which make greater use of polymorphism, it can still seem
overly-restrictive.
For instance, consider the following function $\dec{sandwichClose}$
which takes two functions and a file handle, and makes a sandwich
using them.
\[\ba{l@{\;}c@{\;}l@{\;}l}
  \dec{sandwichClose} &:& (\TUnit\to^\Unl\TUnit\eff\{R_1\}, \File, \TUnit\to^\Unl\TUnit\eff\{R_2\})
    \to^\Unl \TUnit\eff\{R\} \span\\
  \dec{sandwichClose} &=& \lambda^\Unl (g, f, h) .
  &\Lety^\Lin \Unit\revto g\,\Unit \In \Lety^\Unl \Unit\revto \CloseFile\, f
  \In h\,\Unit
\ea\]
Using our trivial-subtyping workaround, we require both $R_1$ and
$R_2$ to be subtypes of $R$.
The problem appears when we try to be polymorphic over $R_1$ and
$R_2$.
Because they are subtypes of the same row type $R$, their row
variables must be the same, i.e., we can only write $R_1=R_2=\mu$ in
\Calc.

To support non-trivial subtyping relations between row variables, we
may again use qualified types, this time to express row subtyping
constraints.
In addition to qualified linear types, \CalcS also supports qualified
effect types inspired by \Rose~\citep{rose}.
In \CalcS, the function $\dec{sandwichClose}$ can be given the
following type. Note that here we still choose to fix functions to be
unlimited for readability.

\[\ba{l@{\;}c@{\;}l@{\;}l}
  \dec{sandwichClose} &:& \forall \mu_1\,\mu_2\,\mu.
  (\mu_1\rle\mu, \mu_2\rle\mu, \File\linle\mu_1) \span\\
  &\qto&
  (\TUnit\to^\Unl\TUnit\eff\{\mu_1\}, \File, \TUnit\to^\Unl\TUnit\eff\{\mu_2\})
    \to^\Unl \TUnit\eff\{\mu\} \span\\
  \dec{sandwichClose} &=& \lambda^\Unl (g, f, h) .
  &\Let \Unit\revto g\,\Unit \In \Let \Unit\revto \CloseFile\, f
  \In h\,\Unit
\ea\]
The constraints $\mu_1\rle\mu$ and $\mu_2\rle\mu$ express that
rows $\mu_1$ and $\mu_2$ are contained in $\mu$, and the constraint
$\File\linle\mu_1$ expresses that the value linearity of $\File$ is
less than the control-flow linearity of $\mu_1$, which essentially
means that $\mu_1$ is control-flow linear.
As in \Cref{sec:qualified-linear-types}, the type inference algorithm
of \CalcS infers these row subtyping constraints without the need for
any annotation.
The qualified linear types and qualified effect types of \CalcS are
decidable.
We give a constraint solving algorithm which checks the satisfiability
of both linearity constraints and row constraints in
\Cref{sec:constraint-solving}.

\section{An explicit handler calculus with linear types}
\label{sec:feffpop}

In this section, we present the syntax, type-and-effect system,
operational semantics and metatheory of \Calc, a System F-style
fine-grain call-by-value calculus with linear types and effect
handlers.
\Calc is based on the core language of \Links which adapts the
subkinding-based linear type system of \Fpop~\citep{fpop} and a
row-based effect system~\citep{linksrow,LindleyC12}.
The linear type system and effect system of \Calc are extended to
track control-flow linearity, which addresses the soundness problem
arising from the interference of linear resources and multi-shot
continuations.
We show that \Calc is truly linearity safe by defining a
linearity-aware semantics and proving that no linear resource is
discarded or duplicated during evaluation in the presence of
multi-shot effect handlers.

\subsection{Syntax and Kinding Rules}
\label{sec:syntax-and-kinding}

\Cref{fig:feffpop-syntax} shows the syntax of types, kinds, contexts,
values, and computations of \Calc.
We introduce a syntactic category $Y$ for linearity consisting of
$\Unl$ and $\Lin$, which intuitively means unlimited and linear,
respectively.
The meaning of linearity varies for values and effects; value types
track value linearity, and effect types track control-flow linearity.
Everything relevant to linearity is highlighted in the figure.
The remaining part is a relatively standard fine-grain call-by-value
calculus with effect handlers and row-based effect system
\citep{HillerstromLA20}.

\begin{figure}[htb]
  \begin{syntax}
  \slab{Value types}    &A,B  &::= & \alpha \mid A \to^{\hl{Y}} C
                               \mid  \forall^{\hl{Y}} \alpha^K.C
                               \\
  \slab{Computation types}
                        &C,D  &::= & A \eff E \\
  \slab{Effect types}   &E    &::= & \{R\}\\
  \slab{Row types}      &R    &::= & \ell : P;R \mid \mu \mid \cdot \\
  \slab{Presence types} &P    &::= & \Abs \mid A \sto^{\hl{Y}} B \mid \theta \\
  \slab{Handler types}  &F    &::= & C \hto D \\
  \slab{Types}          &T    &::= & A \mid R \mid P \mid C \mid E \mid F \\
  \slab{Kinds}          &K    &::= & \kind \Type {\hl{Y}}
    \mid \kind {\Row_{\mathcal{L}}} {\hl{Y}}
    \mid \kind \Presence {\hl{Y}} \mid \Effect \mid \Comp \mid \Handler\\
  \slab{Linearity}      &{Y}    &::= & \hlbox{\Unl \mid \Lin} \\
  \slab{Label sets}     &\mathcal{L} &::=& \emptyset \mid \{\ell\} \uplus \mathcal{L}\\
  \slab{Type contexts} &\Gamma   &::=& \cdot \mid \Gamma, x:A \\
  \slab{Kind contexts} &\Delta   &::=& \cdot \mid \Delta, \alpha:K \\
  \slab{Values}        &V,W  &::=  & x
                               \mid  \lambda^{\hl{Y}} x^A . M
                               \mid \Lambda^{\hl{Y}} \alpha^K . M \\
  \slab{Computations}  &M,N  &::= & V\,W \mid V\,T \mid (\Ret V)^E
                                  \mid (\Do \ell\; V)^E \\
                       &     &\mid& \Lety^{\hl Y} x \revto M \In N \mid \Handle M \With H\\
  \slab{Handlers}      &H    &::= & \{ \Ret x \mapsto M \}
                              \mid  \{ \ell \; p \; r \mapsto M \} \uplus H
  \end{syntax}

  \caption{Syntax of Types, Kinds, Contexts, Values and Computations of $\Calc$}
  \label{fig:feffpop-syntax}
\end{figure}

\Calc explicitly distinguishes between value types and computation
types as well as their terms.
Value types include type variables $\alpha$, function types $A\to^Y
C$, and polymorphic types $\forall^Y\alpha^K.C$.
Value terms include value variables $x$, $\lambda$-abstractions $\lambda^Y
x^A.M$, and type abstractions $\Lambda^Y\alpha^K.M$.
Function types, polymorphic types, and abstractions are annotated with
their value linearity $Y$.
In examples we will freely make use of base types and algebraic data
types whose treatment is quite standard.
We elect to allow polymorphic computation types rather than applying
the value restriction.

A computation type $A\eff E$ comprises a result value type $A$ and an
effect type $E$ specifying the operations that the computation might
perform.
Effect types $\{R\}$ are represented by row types $R$.
Each operation label in rows is annotated with a presence type $P$,
which indicates that the label is either absent $\Abs$, present with
signature $A\sto^Y B$, or polymorphic $\theta$ in its presence.
An operation signature $A\sto^Y B$ describes an operation with
parameter of type $A$ that returns a result of type $B$ and whose
control-flow linearity is $Y$.
Row types are either open (ending with a row variable $\mu$) or closed
(ending with $\cdot$, which we often omit).
We identify rows up to reordering of labels and ignore absent labels
in closed row types \citep{remy1994type}.
Handler types $C\hto D$ represent handlers transforming computations
of type $C$ to computations of type $D$.
By convention, we let $\alpha$ range over value type variables, $\mu$
over row type variables, and $\theta$ over presence type variables,
but we also let $\alpha$ range over all over them (e.g. when binding
quantifiers of unspecified kind).

Function application $V\,W$ and type application $V\,T$ are standard.
A computation $(\Ret V)^E$ returns the value $V$.
An operation invocation $(\Do\ell\,V)^E$ invokes the operation $\ell$
with parameter $V$.
They are both annotated with their effect types for deterministic
typing.
Sequencing $\Lety^Y\,x\revto M\In N$ evaluates $M$ and binds its
result to $x$ in $N$.
The linearity $Y$ basically indicates
the control-flow linearity of $N$.
Handling $\Handle M \With H$ handles computation $M$ with handler $H$.
Handlers are given by a return clause $\Ret x\mapsto M$, which binds
the returned value as $x$ in $M$, and a list of operation clauses
$\ell\,p\,r\mapsto M$, which bind the operation parameter to $p$ and
continuation to $r$ in $M$.

We have six kinds $K$, one for each syntactic category of types.
Kinds are parameterised by linearity $Y$.
The kinds of value types $\kind\Type Y$ denote value linearity, and
the kinds of presence types $\kind\Presence Y$ and row types
$\kind{\Row_\mL} Y$ denote control-flow linearity.
The label set $\mL$ tracks the labels that should not appear in a row,
which is used to avoid duplicated labels in rows.
The kinds of effect, computation, and handler types are not annotated
with any linearity information.
Type contexts $\Gamma$ associate value variables with types, and kind
contexts $\Delta$ associate type variables with kinds.

\Cref{fig:kinding} gives the kinding rules.
Linearity-relevant parts are highlighted.
The kinding relation $\Delta \vdash T : K$ states that type $T$ has
kind $K$ in context $\Delta$.
The subkinding relation $\vdash K\le K'$ states that $K$ is a subkind
of $K'$.
We sometimes write simply $\Delta\vdash T: Y$ for
value, row and presence types when the underlying kind is clear.
The kinding rules for effect, computation, and handler types are
standard~\citep{HillerstromLA20} and irrelevant to linearity
(\klab{Effect}, \klab{Comp}, and \klab{Handler}).

\begin{figure}[htb]
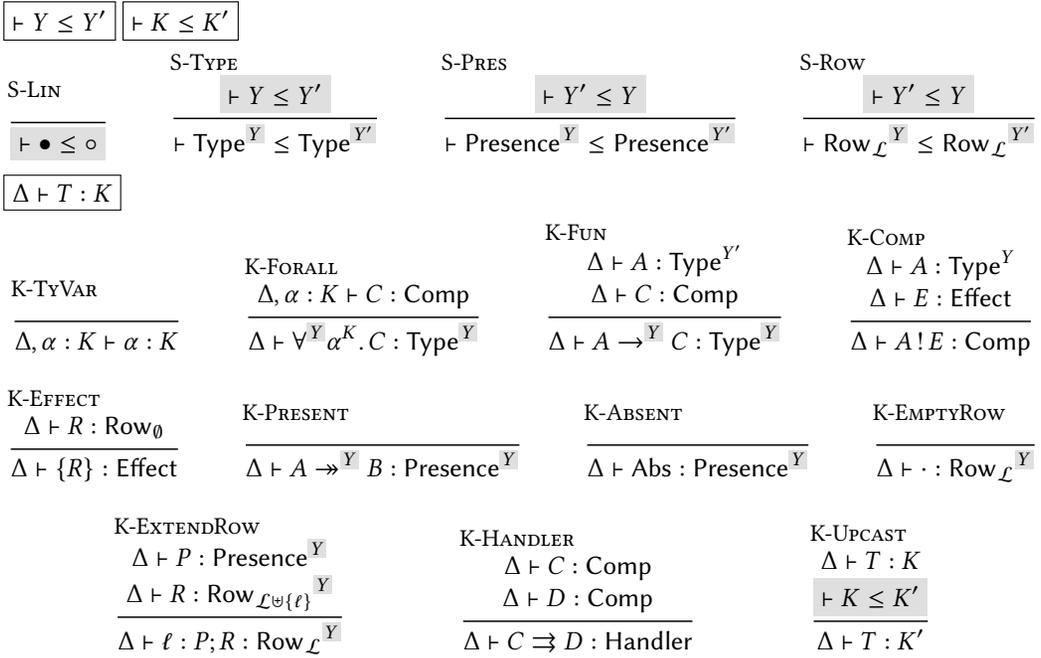

  \raggedright
  \boxed{\vdash Y \leq Y'}\ \boxed{\vdash K \leq K'}\hfill
  \begin{mathpar}
      \inferrule*[Lab=\sklab{Lin}]
      { }
      {\hlbox{\vdash \Unl \leq \Lin}}

      \inferrule*[Lab=\sklab{Type}]
      {\hlbox{\vdash Y \leq Y'} }
      {\vdash \kind \Type {\hl{Y}} \leq \kind \Type {\hl{Y'}}}

      \inferrule*[Lab=\sklab{Pres}]
      {\hlbox{\vdash Y' \leq Y} }
      {\vdash \kind \Presence {\hl{Y}} \leq \kind \Presence {\hl{Y'}}}

      \inferrule*[Lab=\sklab{Row}]
      {\hlbox{\vdash Y' \leq Y} }
      {\vdash \kind {\Row_\mathcal{L}} {\hl{Y}} \leq \kind {\Row_\mathcal{L}} {\hl{Y'}}}
  \end{mathpar}

  \boxed{\Delta \vdash T : K}\hfill
  \begin{mathpar}
    \inferrule*[Lab=\klab{TyVar}]
      { }
      {\Delta, \alpha : K \vdash \alpha : K}

    \inferrule*[Lab=\klab{Forall}]
      { \Delta, \alpha : K \vdash C : \Comp}
      {\Delta \vdash \forall^{\hl Y} \alpha^K . \, C : \kind \Type {\hl Y}}

    \inferrule*[Lab=\klab{Fun}]
      { \Delta \vdash A : \kind \Type {Y'} \\\\
        \Delta \vdash C : \Comp  \\
      }
      {\Delta \vdash A \to^{\hl Y} C : \kind \Type {\hl{Y}}}

  \inferrule*[Lab=\klab{Comp}]
  { \Delta \vdash A : \kind \Type Y \\\\
    \Delta \vdash E : \Effect \\
  }
  {\Delta \vdash A \eff E : \Comp}

    \inferrule*[Lab=\klab{Effect}]
      { \Delta \vdash R : {\Row_\emptyset} }
      {\Delta \vdash \{R\} : \Effect}

    \inferrule*[Lab=\klab{Present}]
      { }
      {\Delta \vdash A \sto^{\hl Y} B : \kind \Presence {\hl Y} }

    \inferrule*[Lab=\klab{Absent}]
      { }
      {\Delta \vdash \Abs : \kind \Presence {\hl Y} }

    \inferrule*[Lab=\klab{EmptyRow}]
      { }
      {\Delta \vdash \cdot : \kind {\Row_\mathcal{L}} {\hl Y} }

    \inferrule*[Lab=\klab{ExtendRow}]
      { \Delta \vdash P : \kind \Presence {\hl Y} \\\\
        \Delta \vdash R : \kind {\Row_{\mathcal{L} \uplus \{\ell\}}} {\hl Y}
      }
      {\Delta \vdash \ell : P;R : \kind {\Row_\mathcal{L}} {\hl Y}}

    \inferrule*[Lab=\klab{Handler}]
      { \Delta \vdash C : \Comp \\\\
        \Delta \vdash D : \Comp
      }
      {\Delta \vdash C \hto D : \Handler} %

  \inferrule*[Lab=\klab{Upcast}]
  {\Delta \vdash T : K \\\\
   \hlbox{\vdash K \leq K'} }
  {\Delta \vdash T : K'}
  \end{mathpar}

  \caption{Kinding and Subkinding Rules for \Calc}
  \label{fig:kinding}
\end{figure}

The kind context maintains kinds for variables (\klab{TyVar}).
The value linearity of function and polymorphic types comes
from their annotations (\klab{Forall} and \klab{Fun}).
Base types have their own value linearity, e.g., $\vdash\File:\Lin$
and $\vdash\Int:\Unl$.
The value linearity of (omitted) algebraic datatypes like pair types
$(A,B)$ is lifted from their components; $\vdash (A,B):\Lin$ if either
$\vdash A:\Lin$ or $\vdash B:\Lin$.

As shown in \Cref{sec:value-linearity}, for value linearity, we have a
subkinding relation $\vdash\kind\Type\Unl\le\kind\Type\Lin$ given by
subkinding rules \sklab{Lin} and \sklab{Type}.
This allows us to use unlimited value types as linear value types since
it is always safe to use unlimited values linearly (e.g., the function
$\dec{id}$ in \Cref{sec:value-linearity}).

We track control-flow linearity at the granularity of operations, and
lift it to the kinds of presence types and row types.
Absent labels and empty rows can be given any control-flow linearity
(\klab{Absent} and \klab{EmptyRow}).
The control-flow linearity of present labels comes directly from
operation signatures (\klab{Present}).
The control-flow linearity of row extensions are given by the labels
and remaining rows (\klab{ExtendRow}).

As shown in \Cref{sec:control-flow-linearity}, control-flow linearity
is dual to value linearity in some sense: we have
$\vdash\kind{\Row_\mL}\Lin\le\kind{\Row_\mL}\Unl$ and
$\vdash\kind\Presence\Lin\le\kind\Presence\Unl$ given by subkinding
rules \sklab{Lin}, \sklab{Pres}, and \sklab{Row}.
This allows linear effect rows to be used as unlimited effect rows as
it is always safe to use control-flow-linear operations in unlimited
contexts (e.g., the function $\dec{tossCoin}$ in
\Cref{sec:control-flow-linearity}).

\subsection{Typing Rules}
\label{sec:feffpop-typing}

We define two auxiliary relations in \Cref{fig:lin-split} for typing
rules.
The judgement $\Delta\vdash\Gamma:Y$ states that under kind context
$\Delta$ all types in $\Gamma$ have linearity $Y$.
As the subkinding relation for value linearity holds that
$\kind\Type\Unl\le\kind\Type\Lin$, the relation $\Delta\vdash
\Gamma:\Unl$ guarantees that all variables in $\Gamma$ are unlimited
and the relation $\Delta\vdash\Gamma:\Lin$ is a tautology.
Dually, as the subkinding relation for control-flow linearity holds
that $\kind\Row\Lin\le\kind\Row\Unl$, the relation $\Delta\vdash R :
\Lin$ guarantees that all operations in $R$ are control-flow linear
and the relation $\Delta\vdash R:\Unl$ is a tautology.
The context splitting judgement $\Delta\vdash\Gamma=\Gamma_1+\Gamma_2$
states that under kind context $\Delta$ the type context $\Gamma$ is well
formed and can be split into two contexts $\Gamma_1$ and $\Gamma_2$
such that each linear variable only appears in one of them.
We write $\Delta \vdash \Gamma_1 + \Gamma_2$ when we only care about
splitting results, and write $\Gamma_1 + \Gamma_2$ in typing rules
when the kind context $\Delta$ is clear.

\begin{figure}[htb]
  \raggedright
  \boxed{\Delta \vdash \Gamma : Y}\hfill
  \vspace{-\the\baselineskip}
  \begin{mathpar}
  \inferrule[L-Empty]
  { }
  {\Delta \vdash \cdot : \hl{Y}}

  \inferrule[L-Extend]
  {\Delta \vdash \Gamma : \hl{Y} \\
   \Delta \vdash A : \kind \Type {\hl Y}}
  {\Delta \vdash (\Gamma, x:A) : \hl{Y}}
  \end{mathpar}

  \boxed{\Delta \vdash \Gamma = \Gamma_1 + \Gamma_2}\hfill
  \begin{mathpar}
  \inferrule[C-Empty]
  { }
  {\Delta \vdash \cdot = \cdot + \cdot}

  \inferrule[C-Unl]
  {\Delta \vdash A : \kind \Type {\hl\Unl} \\
   \Delta \vdash \Gamma = \Gamma_1 + \Gamma_2}
  {\Delta \vdash \Gamma, x:A = (\Gamma_1, x:A) + (\Gamma_2, x:A)}
\\
  \inferrule[C-LinLeft]
  {\Delta \vdash A : \kind \Type {\hl\Lin} \\
   \Delta \vdash \Gamma = \Gamma_1 + \Gamma_2}
  {\Delta \vdash \Gamma,x:A = (\Gamma_1, x:A) + \Gamma_2}

  \inferrule[C-LinRight]
  {\Delta \vdash A : \kind \Type {\hl\Lin} \\
   \Delta \vdash \Gamma = \Gamma_1 + \Gamma_2}
  {\Delta \vdash \Gamma,x:A = \Gamma_1 + (\Gamma_2, x:A)}
  \end{mathpar}

  \caption{Linearity of Contexts and Context Splitting}
  \label{fig:lin-split}
\end{figure}

\begin{figure}[htb]

\raggedright
\boxed{\typ{\Delta; \Gamma}{V : A}}
\boxed{\typ{\Delta; \Gamma}{M : C}}
\boxed{\typ{\Delta; \Gamma}{H : C \hto D}}
\hfill

\begin{mathpar}

\inferrule*[Lab=\tylab{Var}]
{\typ{\Delta}{\Gamma : \Unl}}
{\typ{\Delta;\Gamma, x : A}{x : A}}

\inferrule*[Lab=\tylab{Abs}]
{
  \hlbox{\typ{\Delta}{\Gamma : Y}} \\
  \typ{\Delta}{A : \kind\Type {Y'}} \\\\
  \typ{\Delta;\Gamma, x : A}{M : C}
  }
{\typ{\Delta;\Gamma}{\lambda^{\hl{Y}} x^A .\, M : A \to^{\hl{Y}} C}}

  \inferrule*[Lab=\tylab{TAbs}]
  {
    \hlbox{\typ{\Delta}{\Gamma : Y}} \\
    \alpha \notin \ftv{\Gamma} \\\\
    \typ{\Delta,\alpha : K;\Gamma}{M : C}
  }
  {\typ{\Delta;\Gamma}{\Lambda^{\hl{Y}} \alpha^K .\, M : \forall^{\hl{Y}} \alpha^K . \,C}}

  \inferrule*[Lab=\tylab{App}]
    {\typ{\Delta;\Gamma_1}{V : A \to^{\hl{Y}} C} \\\\
     \typ{\Delta;\Gamma_2}{W : A}
    }
    {\typ{\Delta;\Gamma_1 + \Gamma_2}{V\,W : C}}

  \inferrule*[Lab=\tylab{TApp}]
    {\typ{\Delta;\Gamma}{V : \forall^{\hl{Y}} \alpha^K . \, C} \\\\
     \Delta \vdash T : K
    }
    {\typ{\Delta;\Gamma}{V\,T : C[T/\alpha]}}

  \inferrule*[Lab=\tylab{Return}]
    {
      \typ{\Delta;\Gamma}{V : A} \\
      \typ{\Delta}{E : \Effect}
    }
    {\typc{\Delta;\Gamma}{(\Ret V)^E : A}{E}}

  \inferrule*[Lab=\tylab{Do}]
    {
      E = \{\ell : A \sto^{\hl{Y}} B; R\} \\\\
      \typ{\Delta;\Gamma}{V : A} \\
      \typ{\Delta}{E : \Effect}
    }
    {\typc{\Delta;\Gamma}{(\Do \ell \; V)^E : B}{E}}

  \inferrule*[Lab=\tylab{Seq}]
    {\typc{\Delta;\Gamma_1}{M : A}{\{R\}} \\
      \typc{\Delta;\Gamma_2, x : A}{N : B}{\{R\}} \\\\
      \hlbox{\Delta \vdash \Gamma_2 : Y} \\
      \hlbox{\Delta \vdash R : Y} \\
    }
    {\typc{\Delta;\Gamma_1 + \Gamma_2}{\Lety^{\hl Y} x \revto M \In N : B}{{\{R\}}}}

\inferrule*[Lab=\tylab{Handle}]
{
  \typ{\Delta;\Gamma_1}{H : C \hto D} \\
  \typ{\Delta;\Gamma_2}{M : C} \\
}
{\Delta;\Gamma_1 + \Gamma_2 \vdash \Handle M \With H : D}

  \inferrule*[Lab=\tylab{Handler}]
  {
      H = \{\Ret x \mapsto M\} \uplus \{ \ell_i\;p_i\;r_i \mapsto N_i \}_i \\
      C = A \eff \{(\ell_i : A_i \sto^{\hl{Y_i}} B_i)_i; R\} \\
      D = B \eff \{(\ell_i : P)_i; R\}\\\\
      \hlbox{\Delta\vdash \Gamma:\Unl}\\
      \typ{\Delta;\Gamma, x : A}{M : D}\\\\
      [\typ{\Delta;\Gamma, p_i : A_i, r_i : B_i \to^{\hl{Y_i}} D}{N_i : D}]_i
  }
  {\typ{\Delta;\Gamma}{H : C \hto D}}

\end{mathpar}

\caption{Typing Rules for \Calc}
\label{fig:typing}
\end{figure}

The typing rules for values, computations, and handlers are given in
\Cref{fig:typing}. Linearity-relevant parts are highlighted.
The relations $\typ{\Delta;\Gamma}{V:A}$, $\typ{\Delta;\Gamma}{M:C}$,
and $\typ{\Delta;\Gamma}{H:C\hto D}$, state respectively that: value
$V$ has type $A$, computation $M$ has type $C$ and handler $H$ has
type $C\hto D$ in contexts $\Delta$ and $\Gamma$.
As usual, the type contexts and types are well formed under the kind
contexts.

The \tylab{Var} rule requires the remaining context to be unlimited.
The \tylab{Abs} and \tylab{TAbs} rules check the value linearity
of functions and polymorphic computations against that of the context
via the premise $\Delta\vdash\Gamma:Y$.
The typing rules for function application and type application are
standard (\tylab{App} and \tylab{TApp}).
Note that we need to split the context in the \tylab{App} rule to
avoid duplicating linear variables.
The \tylab{Return} rule does not constrain the effects. The
\tylab{Do} rule ensures that the operation $\ell$ and its parameter
$V$ agree with the effect signature $E$.
The \tylab{Handle} rule uses a handler of type $C\hto D$ to handle a
computation of type $C$.

The \tylab{Handler} rule checks that (deep) handlers must not use any
linear variables via the premise $\Delta\vdash\Gamma:\Unl$ because
they are recursively applied during evaluation. 
More importantly, it connects the control-flow linearity of operations
with the value linearity of resumption functions.
In the typing judgement of each operation clause $\ell_i : A_i
\sto^{Y_i} B_i$, the continuation $r_i$ is given the value linearity
$Y_i$, which is exactly the control-flow linearity of $\ell_i$ that
restricts the use of $\ell_i$'s continuation.
Concretely, when $Y_i = \Lin$, the continuation of $\ell_i$ may use
some linear resources.
Making $r_i$ linear guarantees that they are used exactly once.
When $Y_i = \Unl$, the continuation of $\ell_i$ must not use any
linear resources and $r_i$ is unlimited.
Note that the subkinding relation $\kind\Row\Lin \le \kind\Row\Unl$
does not influence the handling behaviour, because the \tylab{Handler}
rule uses the linearity annotations on operation signatures.

The \tylab{Seq} rule for sequencing is the most important rule for
tracking control-flow linearity, because this is the primary source of
sequential control flow in a fine-grain call-by-value calculus.
Though handling is another source of sequential control flow, deep
handlers are unlimited and cannot influence control-flow linearity.
We will discuss the extension of shallow handlers which may capture
linear resources and influence control-flow linearity in
\Cref{sec:extensions}.

Remember that for $\Lety^Y x \revto M \In N$, the linearity annotation
$Y$ indicates the control-flow linearity of $N$ which determines
how many times the control can enter $N$.
Concretely, when $Y = \Lin$, $N$ may use some linear variables bound
outside ($\Delta\vdash\Gamma_2:\Lin$), and all operations in $M$
should be control-flow linear ($\Gamma\vdash R:\Lin$); when $Y =
\Unl$, $N$ cannot use any linear variables from the context
($\Delta\vdash\Gamma_2:\Unl$), and operations in $M$ have no
restriction on their control-flow linearity ($\Delta\vdash R:\Unl$).
The $\dec{dubiousWrite}_\tmark$ in \Cref{sec:control-flow-linearity}
is an example.
Note that technically, the third sequencing $\Lety^\Lin f'\revto
\WriteFile\,(s,f)\In\CloseFile\,f'$ can be changed to $\Lety^\Unl$
because no linear variable bound outside is used by the context
$\Lety\, f'\revto \_\In\CloseFile\,f'$.

As we observed by the function $\dec{verboseClose}$ in
\Cref{sec:qualified-effect-types}, the fact that the \tylab{Seq} rule
requires the $M$ and $N$ to have the same effect type is too
restrictive for tracking control-flow linearity.
We can improve it by using a trivial subtyping relation between effect
types as follows.
\begin{mathpar}
  \inferrule*[Lab=\tylab{SeqSub}]
    {\typc{\Delta;\Gamma_1}{M : A}{\{R_1\}} \\
     \typc{\Delta;\Gamma_2, x : A}{N : B}{\{R_2\}} \\\\
     \hlbox{\typ{\Delta}{\Gamma_2 : Y}} \\
     \hlbox{\Delta \vdash R_1 : Y} \\
     \hlbox{\Delta \vdash R_1 \subtype R : K} \\
     \hlbox{\Delta \vdash R_2 \subtype R : K}
    }
    {\typc{\Delta;\Gamma_1 + \Gamma_2}{\Lety^Y x \revto M \In N : B}{{\{R\}}}}
\end{mathpar}
The trivial subtyping relation on effect row types are shown in
\Cref{fig:row-subtyping}.
\begin{figure}[tbp]
  \raggedright
  {\boxed{\Delta \vdash R \subtype R' : K}}\hfill
  \begin{mathpar}
      \inferrule
      {\Delta \vdash R : K}
      {\Delta\vdash R \subtype R : K}

      \inferrule
      {\Delta\vdash R_1 \subtype R_2:K \\ \Delta\vdash R_2 \subtype R_3:K}
      {\Delta\vdash R_1 \subtype R_3:K}

      \inferrule
      {\Delta\vdash \mu : K}
      {\Delta\vdash \cdot \subtype \mu : K}

      \inferrule
      { \Delta \vdash P : \kind \Presence {Y} \\\\
        \Delta\vdash R_1 \subtype R_2 : \kind {\Row_{\mL \uplus \{\ell\}}} {Y}}
      {\Delta\vdash \ell:\Abs ; R_1 \subtype \ell:P ; R_2 : \kind{\Row_\mL}{Y}}

      \inferrule
      { \Delta \vdash P : \kind \Presence {Y} \\\\
        \Delta\vdash R_1 \subtype R_2 : \kind {\Row_{\mL \uplus \{\ell\}}} {Y}}
      {\Delta\vdash \ell:P ; R_1 \subtype \ell:P ; R_2 : \kind{\Row_\mL}{Y}}

  \end{mathpar}
  \caption{Trivial Subtyping for Effect Row Types}
  \label{fig:row-subtyping}
\end{figure}
The judgement $\Delta \vdash R\subtype R' : K$ makes it explicit that
$R$ and $R'$ are well kinded and can be given kind $K$ under kind
context $\Delta$. It simply requires that all operation labels with
their signatures and row variable in $R$ must also appear in $R'$.
This subtyping relation does not allow non-trivial subtyping between
row variables.
We consider a more expressive alternative using qualified types in
\Cref{sec:qeffpopsub}.

\subsection{Operational Semantics}
\label{sec:feffpop-semantics}

\begin{figure}[htb]
  \flushleft
  \begin{reductions}
  \semlab{App}   & (\lambda^Y x^A.M)\,V &\reducesto& M[V/x] \\
  \semlab{TApp} & (\Lambda^Y \alpha^K.M)\,T &\reducesto& M[T/\alpha] \\
  \semlab{Seq} &
    \Lety^Y x \revto (\Ret V)^E \In N &\reducesto& N[V/x] \\
  \semlab{Ret} &
    \Handle (\Ret V)^E \With H &\reducesto& N[V/x], \hfill \text{where }
      (\Ret x \mapsto N) \in H \\
  \semlab{Op} &
    \Handle \EC[(\Do \ell \; V)^E] \With H
      &\reducesto& N[V/p, (\lambda^Y y^B.\Handle \EC[(\Ret y)^E] \With H)/r],\\
  \multicolumn{4}{@{}r@{}}{
        \text{ where } \ell \notin \BL{\EC},
        (\ell \; p \; r \mapsto N) \in H, \text{ and }
        (\ell:A\to^Y B) \in E
  } \\
  \semlab{Lift} &
    \EC[M] &\reducesto& \EC[N],  \hfill\text{if } M \reducesto N \\
  \end{reductions}
  \begin{syntax}
  \slab{Evaluation contexts} &  \EC &::= & [~] \mid \Lety^Y x \revto \EC \In N \mid
                                           \Handle \EC \With H \\ %
  \end{syntax}
\[\ba{@{}c@{}}
\BL{[~]}                        = \emptyset           \qquad
\BL{\Lety^Y x \revto \EC \In N} = \BL{\EC}               \qquad
\BL{\Handle \EC \With H}     = \BL{\EC} \cup \dom{H}
\ea\]
  \caption{Small-step Operational Semantics of \Calc}
  \label{fig:feffpop-small-step}
\end{figure}

\Cref{fig:feffpop-small-step} gives a standard small-step operational
semantics for \Calc~\citep{HillerstromLA20}.
It is clear from the definition of evaluation contexts that
let-binding and handling are indeed the only two constructs that
influence the control flow. The function $\BL{-}$ computes the set of
\emph{bound operation labels} in an evaluation context $\EC$, i.e. the
operation labels for which a suitable handler has been installed. The
purpose of this function is to ensure that any operation invocation
($\Do~\ell~V$) is always handled by the innermost suitable handler.

\subsection{Metatheory}
\label{sec:metatheory}

We now prove a type soundness result for \Calc. First we define normal
forms of computations.

\begin{definition}[Computation Normal Forms]
  \label{def:normal-form}
  We say a computation $M$ is in a normal form with respect to $E$, if
  it is either of the form $M = (\Ret V)^{E'}$ or $M = \EC[(\Do \ell\;
  V)^{E'}]$ for $\ell\in E$ and $\ell \notin \BL{\EC}$.
\end{definition}

Syntactic type soundness of \Calc relies on progress and subject
reduction.
The proofs can be found in
\Cref{app:proof-progress,app:proof-subject-reduction}.

\begin{restatable}[Progress]{theorem}{progress}
  \label{thm:progress}
  If $\typc{}{M : A}{E}$, then either there exists $N$ such that $M
  \reducesto N$ or $M$ is in a normal form with respect to $E$.
\end{restatable}

\begin{restatable}[Subject reduction]{theorem}{preservation}
  \label{thm:preservation}
  If $\typ{\Delta; \Gamma}{M:C}$ and $M \reducesto N$, then
  $\typ{\Delta;\Gamma}{N:C}$.
\end{restatable}

We now show that our tracking of value linearity and control-flow
linearity in the type system is sound, by proving that linear
variables never appear in terms that are claimed to be unlimited.
In \Calc, a term is claimed to be unlimited if it appears in an
unlimited value, a control-flow-unlimited context, or a deep handler.
The following theorem covers all three of these cases.

\begin{restatable}[Unlimited is unlimited]{theorem}{unlisunl}
  \label{thm:unl-is-unl}~
  \begin{enumerate}[label=\arabic*.]
    \item Unlimited values are unlimited: if
    $\typ{\Delta;\Gamma}{V:A}$ and $\Delta\vdash A:\Unl$, then
    $\Delta\vdash\Gamma : \Unl$.
    \label{item:unl-is-unl-a}
    \item Unlimited continuations are unlimited: if
    $\typ{\Delta;\Gamma}{\EC[(\Do \ell\; V)^E] : C}$ for $E =
    {\{\ell:A\sto^\Unl B\lsep R\}}$ and $\ell \notin \BL{\EC}$, then
    there exists $\Delta \vdash \Gamma = \Gamma_1 + \Gamma_2$ such
    that $\Delta \vdash \Gamma_1 : \Unl$ and
    $\typ{\Delta;\Gamma_1,y:B}{\EC[(\Ret y)^E] : C}$.
    \item Deep handlers are unlimited: if $\typ{\Delta;\Gamma}{H : C
    \hto D}$, then $\Delta\vdash\Gamma:\Unl$.
  \end{enumerate}
\end{restatable}

The proof can be found in \Cref{app:proof-unlimited-is-unlimited}.

However, \Cref{thm:unl-is-unl} only cares about the static tracking of
linear variables.
It says nothing about the use of linear values during evaluation directly.
In the next section, we prove that in \Calc no linear value is ever
discarded or duplicated during evaluation, by defining a
linearity-aware semantics inspired by \citet{walker2005substructural},
\citet{fpop}, and \citet{quill}.

\subsection{Linearity Safety of Evaluation}
\label{sec:feffpop-linear-semantics}

In this section, we design a linearity-aware semantics of \Calc,
extending the small-step operational semantics to track the
introduction and elimination of linear values, and prove that all
linear values are used exactly once during evaluation.

We first extend the syntax of values with values marked with linear
tags $\lintag{V}$ to indicate linear values during evaluation.
The typing rules simply ignore the linear tags.
\begin{syntax}
  \slab{Values}         &V  &::=  & \dots \mid \lintag{V} \\
\end{syntax}
We restrict attention to closed computations and define two auxiliary
functions $\islin(V)$ and $\taglin(V)$ for closed values as follows.
\[\ba{rcl}
\islin(V) &=& \left\{
  \ba{l@{\quad}l}
  \True  & \text{if } \typ{\cdot;\cdot}{V:A} \text{ and } \cdot\not\vdash A : \Unl \\
  \False & \text{otherwise}
  \ea \right. \\[2ex]
\taglin(V) &=& \left\{
  \ba{l@{\quad}l}
  (\lintag{V}, \{\lintag{V}\}) & \text{if } \islin(V) \text{ and } V \neq \lintag{W} \text{ for any } W\\
  (V, \emptyset)               & \text{otherwise}
  \ea \right.
\ea\]
The predicate $\islin(V)$ holds when $V$ is a genuine linear value as
opposed to an unlimited value that has been upcast to be linear by
subkinding.
The operation $\taglin(V)$ tags a value as linear if it is and has not
been tagged, and yields a pair of the possibly tagged $V$ and a
multiset containing the value if it is newly tagged and nothing
otherwise.

\begin{figure}
  \flushleft
  \begin{bigreductions}
  \lsemlab{App}   & (\lambda^Y x^A.M)\,V &\lreducesto{\mS}{\emptyset}& M[V'/x],
    \text{ where } (V', \mS) = \taglin(V) \\[1.2ex]
  \lsemlab{TApp} & (\Lambda^Y \alpha^K.M)\,T &\lreducesto{\emptyset}{\emptyset}& M[T/\alpha] \\[1.2ex]
  \lsemlab{Seq} &
    \Lety^Y x \revto \Ret V \In N &\lreducesto{\mS}{\emptyset}& N[V'/x],
    \text{ where } (V', \mS) = \taglin(V) \\[1.2ex]
  \lsemlab{Ret} &
    \Handle (\Ret V)^E \With H &\lreducesto{\mS}{\emptyset}& N[V'/x],\\
    \multicolumn{4}{@{}r@{}}{
        \text{where } (\Ret x \mapsto N) \in H, (V', \mS) = \taglin(V)
    } \\[1.2ex]
  \lsemlab{Op} &
    \Handle \EC[(\Do \ell \; V)^E] \With H
      &\lreducesto{\mS}{\emptyset}& N[V'/p, W'/r],\\
    \multicolumn{4}{@{}r@{}}{
          \text{where } \ell \notin \BL{\EC},\
          (\ell \ p \ r \mapsto N) \in H,\
          (\ell:A \sto^{Y} B) \in E,
    }\\
    \multicolumn{4}{@{}r@{}}{
          W = \lambda^{Y} y^{B}.\Handle \EC[(\Ret y)^E] \With H,
    } \\
    \multicolumn{4}{@{}r@{}}{
          (V',\mS_1) = \taglin(V),
          (W',\mS_2) = \taglin(W), \mS = \mS_1 \umul \mS_2
    } \\[1.2ex]
  \lsemlab{Remove} & \FC[\lintag{V}] &\lreducesto{\emptyset}{\{\lintag{V}\}} &\FC[V] \\[1.2ex]
  \lsemlab{Lift} &
    \EC[M] &\lreducesto{\mS}{\mT}& \EC[N],  \hfill\text{ if } M \lreducesto{\mS}{\mT} N \\
  \end{bigreductions}
  
  \begin{syntax}
  \slab{Evaluation contexts} &  \EC &::= & [~] \mid \Lety^Y x \revto \EC \In N \mid
                                           \Handle \EC \With H \\
  \slab{Tag-removing contexts} &  \FC &::= & [~]\; V \mid [~]\; T
  \end{syntax}
  \caption{Linearity-aware Small-step Operational Semantics of \Calc}
  \label{fig:feffpop-linear-small-step}
\end{figure}

The linearity-aware semantics is given in
\Cref{fig:feffpop-linear-small-step}.
We augment the previous reduction relation $M \reducesto N$ with two
multi-sets $M \lreducesto{\mS}{\mT} N$, where $\mS$ contains the
linear values introduced by this reduction step, and $\mT$ contains
the linear values eliminated by this reduction step.
Note that in \Calc, we cannot duplicate or discard a value before we
bind it.
We introduce linear values at the first time they are bound to
variables (\lsemlab{App}, \lsemlab{Seq}, \lsemlab{Ret} and
\lsemlab{Op}).
Take \lsemlab{App} for example.
When $V$ is a non-tagged real linear value (the first case of
$\taglin(V)$), we tag it and add it to the multiset of introduced
linear values.
Otherwise, $V$ is either not really linear or has been tagged already
(which implies that we have already introduced it). We do not need to
update the multisets.
We eliminate linear values when they are destructed (\lsemlab{Remove}).
As we only have term abstraction and type abstraction as value
constructors, the tag-removing contexts $\FC$ capture the elimination
of these two cases.
It is easy to extend the linearity-aware semantics with other value
constructors.
The relationship between the two semantics is straightforward: erasing
the linear tags from the linearity-aware semantics yields the original
semantics.

We write $\TL{M}$, $\TL{V}$, $\TL{\EC}$ and $\TL{\FC}$ for the
multisets of tagged linear values within $M$, $V$, $\EC$, and $\FC$,
respectively. They are given by the homomorphic extension of the
following equation.
\[
\TL{\lintag{V}} = \{\lintag{V}\}\umul\TL{V}
\]

We define the notion of linear safety similarly to
\Cref{thm:unl-is-unl}.
A term is linear safe if there are no tagged linear values in terms
that are claimed to be unlimited.

\begin{definition}[Linear Safety]
  \label{def:linear-safety}
  A well-typed computation $M$ or value $V$ is \emph{linear safe} if
  and only if:
  \begin{enumerate}
    \item For every value subterm $W$ of the form $\lambda^\Unl x^A.N$ or
      $\Lambda^\Unl \alpha^K.N$, $\TL{W} = \emptyset$.\label{def:linear-safety-1}
    \item For every computation subterm $N$ of the form $\EC[(\Do \ell\;
    V)^{\{\ell : A\sto^\Unl B\lsep R\}}]$ where $\ell \notin \BL{\EC}$,
    $\TL{\EC} = \emptyset$.
    \item For every handler subterm $H$, $\TL{H} = \emptyset$.
  \end{enumerate}
  (An alternative way to read Item~\ref{def:linear-safety-1} is as
  ``for every value subterm $W$ with an unlimited type''.)
\end{definition}

Finally, the following theorem states that linear safety is preserved
by evaluation, and tagged linear values are not duplicated or
discarded during evaluation.

\begin{restatable}[Reduction Safety]{theorem}{reductionSafety}
  \label{thm:reduction-safety}
  For any closed, well-typed and linear safe computation $M$ in \Calc,
  if $M \lreducesto{\mS}{\mT} N$, then $N$ is linear safe and
  $\TL{M}\umul\mS = \TL{N}\umul\mT$.
\end{restatable}

The proof can be found in \Cref{app:proof-reduction-safety}.
Note that tracking linear values explicitly during evaluation is
important for showing that they are indeed used safely.
Otherwise, it is even unclear how to state what reduction safety means
in the original semantics.

\section{Control-Flow Linearity in Links}
\label{sec:links}

In this section, we describe our implementation of control-flow
linearity tracking in \Links.
The implementation fixes a long-standing type soundness bug in \Links
arising from the interaction between session types and effect
handlers, as we described in the introduction.

\Links is an ML-style language with type inference, linearly typed
session types (based on $\Fpop$~\citep{lindley2017lfst}), and a
row-based effect type system~\citep{linksrow}.
In \Links we write \lstinline$Unl$ for $\Unl$ and \lstinline$Any$ for
$\Lin$.
The latter is \lstinline$Any$ as \emph{any} value can be soundly used
once.
The subkinding relation $\vdash\kind\Type\Unl\le\kind\Type\Lin$
(\lstinline$Unl$ $\le$ \lstinline$Any$) allows type variables of kind
\lstinline$Any$ to be unified with types of either kind.
This allows us to write functions that may accept both linear and
nonlinear values, e.g. the identity function
\lstinline$fun id(x){x} : (a::Any) -> (a::Any)$.
Here, we can instantiate the type variable \lstinline$a$ to a linear
type, such as \lstinline$!Int.End$, or an unlimited type, such as
\lstinline$Int$.

To make type inference deterministic, \Links makes use of two
different keywords for defining unlimited functions and linear
functions, which are \lstinline$fun$ and \lstinline$linfun$
respectively.
For instance, we can define a channel version of the function
$\dec{faithfulWrite}$ in \Cref{sec:value-linearity} as follows.
\begin{lstlisting}[numbers=none]
fun faithfulSend(c) { linfun (s) { var c = send(s, c); close(c) } }
\end{lstlisting}
The inferred type is \lstinline$(!(a::Any).End) -> (a::Any) ~@ ()$.
The \lstinline{faithfulSend} function takes a polymorphic channel
\lstinline{c} and returns a linear function (indicated by
\lstinline{~@} instead of the usual arrow \lstinline{~>}) that sends a
polymorphic value $s$ over the channel \lstinline{c}. If we wanted to
we could restrict the inferred type of the channel \lstinline$c$ and
the input $s$ by supplying a type annotation to either.

To track control-flow linearity we repurpose the existing effect
system and add two new control flow kinds \lstinline$Any$ (for $\Unl$)
and \lstinline$Lin$ (for $\Lin$) to signify whether a given context
allows control flow to be unlimited or linear. We further add a new
effectful operation space for control-flow-linear operations, which is
syntactically denoted by the arrow \lstinline$=@$, in addition to the
existing operation space denoted by \lstinline$=>$.
The subkinding relation $\vdash\kind\Row\Lin\le\kind\Row\Unl$
(\lstinline$Lin$ $\le$ \lstinline$Any$) is implemented by allowing row
variables of kind \lstinline{Any} to be unified with both
control-flow-linear and unlimited operations and other row variables
of arbitrary kinds.
In contrast, row variables of kind \lstinline{Lin} can only be unified
with control-flow-linear operations and row variables of kind
\lstinline{Lin}.
The change from \lstinline{Unl} to \lstinline{Lin}
is consistent with the duality between value linearity and
control-flow linearity.

Since \Links is a practical programming language, sequencing is often
implicit.
Instead of writing linearity annotations on all sequencing, we assume
that control-flow linearity is unlimited by default, and introduce the
keyword \lstinline$xlin$ to switch the control-flow linearity to
linear.
We also add the construct \lstinline$lindo$ to invoke
control-flow-linear operations in addition to the existing
\lstinline$do$ for control-flow-unlimited operations.
To illustrate the use of these extensions, let us consider a channel
version of the function $\dec{dubiousWrite}_\tmark$ from
\Cref{sec:control-flow-linearity}.
\begin{lstlisting}[numbers=none]
sig dubiousSend : (!String.End) {Choose:() =@ Bool|_::Lin}~> ()
fun dubiousSend(c) {xlin; var c = send(if (lindo Choose) "A" else "B", c); close(c)}
\end{lstlisting}
The \lstinline{dubiousSend} takes a channel \lstinline{c},
non-deterministically sends \lstinline{"A"} or \lstinline{"B"} through
it depending on the result of the operation \lstinline{Choose}, and
closes the remaining channel.
We use \lstinline{xlin} to switch the control-flow linearity to linear
so that we can use the linear channel \lstinline{c} and must use the
control-flow-linear operation \lstinline{Choose:() =@ Bool} with the
keyword \lstinline{lindo}.
If we replace \lstinline{lindo} with \lstinline{do} then \Links
correctly rejects the code as the continuation captures the linear
endpoint \lstinline{c}.
The example from the introduction will be rejected for the same
reason.
For linear effect handlers, we use the linear arrow syntax
\lstinline$=@$ to bind linear continuations of control-flow-linear
operations.
\begin{lstlisting}[numbers=none]
fun(c) {handle ({xlin; dubiousSend(c)}) {case <Choose =@ r> -> xlin; r(true)} }
\end{lstlisting}
Here, we interpret the operation \lstinline$Choose$ as
\lstinline$true$.
The use of \lstinline$xlin$ in the \lstinline$Choose$-clause is
necessary because the reified continuation $r$ is linear.
As the continuation is used linearly, \Links correctly accepts this
program.

Our implementation works well with previous programs using the effect
handler feature in \Links and fixes the type soundness bug.
However, being based on \Fpop, \Links suffers from the limitations
outlined in Section~\ref{sec:overview}.
In the next section, we present a considerably more expressive
calculus, \CalcS, which uses qualified types for both linearity and
effects, enabling a much more fine-grained analysis of control-flow
linearity, and avoiding the need to distinguish between linear and
non-linear variants of term syntax.
We leave the implementation of \CalcS to future work.

\section{An Implicit Calculus with Qualified Types}
\label{sec:qeffpopsub}

In this section, we propose \CalcS, an ML-style calculus which
enhances \Calc (and its implementation in \Links) in two directions:
minimising syntactic overheads and improving accuracy of control-flow
linearity tracking.
The core idea is to use qualified types for both linear types and
effect types.
The qualified linear type system is inspired by \Quill~\citep{quill},
which eliminates the linearity annotations on terms and supports
principal types.
The qualified effect system is inspired by the row containment
predicate of \Rose~\citep{rose} and the subtyping-based effect system
of \Eff~\citep{pretnar14, KarachaliasPSVS20}, which allows non-trivial
subtyping constraints between row variables.

\subsection{Syntax}

\Cref{fig:qeffpopsub-syntax-abbrev} shows the syntax of qualified
types of \CalcS.
We name some syntactic categories for defining meta functions.
The remaining syntax is given in full in \Cref{app:qfp-syntax}, which
is mostly identical to that of \Calc, except that we introduce
generalising let-bindings $\Let x=V\In M$ to replace explicit type
abstraction and implicit instantiation in place of type application
and remove all type annotations and linearity annotations.

\begin{figure}[htb]
  \flushleft
  \begin{minipage}[t]{0.5\textwidth}
  \begin{syntax}
  \slab{Linearity}   &Y    &::= & \phi \mid \Unl \mid \Lin \\
  \slab{Types}          &\tau &::= & A \mid R \mid Y \\
  \slab{Predicates}         &\Belong{\pi}{\meta{Pred}} &::= & \tau_1 \linle \tau_2
    \mid R_1 \rle R_2 \\
    & &\mid& R \rlack \mL \\
  \end{syntax}
  \end{minipage}
  \begin{minipage}[t]{0.45\textwidth}
  \begin{syntax}
  \slab{Qualified types} &\rho &::=& A \mid {\pi \Rightarrow \rho} \\
  \slab{Type schemes}   &\Belong{\sigma}{\meta{TySch}} &::=& \rho \mid \forall \alpha . \sigma \\
  \slab{Type contexts}  &\Belong{\Gamma}{\meta{Env}}   &::=& \cdot \mid \Gamma, x:\sigma \\
  \slab{Predicate sets}     &\Belong{P}{\meta{PSet}} &::= & \cdot \mid P, \pi \\
  \end{syntax}
  \end{minipage}

  \caption{Syntax of Qualified Types of \CalcS}
  \label{fig:qeffpopsub-syntax-abbrev}
\end{figure}

\paragraph{Linearity}
In addition to concrete linearities $\Lin$ and $\Unl$, \CalcS has linearity
variables $\phi$. This is essential to have principal types and more expressive
constraints.
For example, the identity function $\lambda x.\Ret x$ can be given
the principal type
$\forall\alpha\,\mu\,\phi.\,\alpha\to^\phi\alpha\eff\{\mu\}$, which
can be instantiated to either a linear function (by instantiating
$\phi$ to $\Lin$) or an unlimited function (by instantiating $\phi$ to
$\Unl$).

\paragraph{Qualified types}
The syntactic category $\tau$ includes value types, row types, and
linearity types.
Qualified types $\rho$ restrict value types by predicates.
The linearity predicate $\tau_1\linle\tau_2$ means the linearity of
$\tau_1$ is less than $\tau_2$ (e.g., $\Unl\linle\Lin$).
Note that we allow directly using value types and row types in the
linearity predicates, since every value type has its value linearity,
and every effect row type has its control-flow linearity.
The row predicates $R_1\rle R_2$ means $R_1$ is a sub-row of $R_2$,
and $R\rlack\mL$ means $R$ does not contain labels in $\mL$.

\paragraph{Kinding}

For conciseness we omit kinds and infer the kind of a type variable
from its name.
As usual, we let $\alpha$ range over value types, $\mu$ range over row
types, and $\phi$ range over linearity types.
We also let $\alpha$ range over all of them in the definition of type
schemes $\forall\alpha.\sigma$.
All rows are assumed to be well-formed (no duplicated
labels).
To simplify type inference, the predicate $\mu \rlack \mL$ will be
used in place of kinds $\Row_\mL$ to track labels that may not occur
in rows. This is just a convenience, though, as the corresponding
kinds of row type variables can be computed from the inferred types.

\subsection{Typing}

\begin{figure}[htb]

  \raggedright
  \boxed{\typ{P\mid\Gamma}{V : A}}
  \boxed{\typ{P\mid\Gamma}{M : C}}
  \boxed{\typ{P\mid\Gamma}{H : C \hto D}}
  \hfill

  \begin{mathpar}

  \inferrule*[Lab=\qtylab{Let}]
  { \typ{Q\mid\Gamma_1, \Gamma}{V : A} \\
    \sigma = \gen((\Gamma_1, \Gamma), Q\qto A) \\\\
    \typ{P\mid\Gamma_2, \Gamma, x : \sigma}{M : C} \\
    P \vdash \Gamma\un
  }
  {\typ{P\mid\Gamma_1, \Gamma_2, \Gamma}{\Let x = V \In M : C}}

    \inferrule*[Lab=\qtylab{Do}]
      {
        \typ{P\mid\Gamma}{V : A_\ell} \\\\
        P \vdash \{\ell\lsig A_\ell\sto^Y B_\ell\} \rle R
      }
      {\typc{P\mid\Gamma}{\Do \ell \; V : B_\ell}{\{R\}}}

  \inferrule*[Lab=\qtylab{Seq}]
  { \typc{P\mid\Gamma_1, \Gamma}{M : A}{\{R_1\}} \\\\
    \typc{P\mid\Gamma_2, \Gamma, x : A}{N : B}{\{R_2\}} \\\\
    P \vdash R_1 \rle R \\ P \vdash R_2 \rle R \\\\
    P \vdash \Gamma_2 \linle R_1 \\
    P \vdash \Gamma \un
  }
  {\typc{P\mid\Gamma_1,\Gamma_2,\Gamma}{\Let x \revto M \In N : B}{\{R\}}}

  \inferrule*[Lab=\qtylab{Handler}]
  {
    H = \{\Ret x \mapsto M\} \uplus \{ \ell_i\;p_i\;r_i \mapsto N_i \}_i\\\\
    C = A \eff \{(\ell_i : A_i \sto^{Y_i} B_i)_i; R_1\} \\
    D = B \eff \{R_2\} \\\\
    \typ{P\mid\Gamma, x : A}{M : D}\\\\
    [\typ{P\mid\Gamma, p_i : A_i, r_i : B_i \to^{Y_i} D}{N_i : D}]_i \\\\
    P \vdash \Gamma \un \\
    P \vdash R_1 \rle R_2 \\
    P \vdash R_1 \rlack \{\ell_i\}_i
  }
  {\typ{P\mid\Gamma}{H : C \hto D}}

  \end{mathpar}

  \[
  \text{where} \, \gen(\Gamma, \rho) = \forall(\ftv\rho\backslash \ftv\Gamma) . \rho.
  \]

  \caption{Selected Syntax-directed Typing Rules for \CalcQ}
  \label{fig:qe-typing-abbrev}
\end{figure}

\Cref{fig:qe-typing-abbrev} gives representative syntax-directed
typing rules for \CalcS; the remaining rules are given in full in
\Cref{app:qfp-typing}. The judgement $P \mid \Gamma \vdash M : C$
states that, under predicate assumptions $P$ and typing assumptions
$\Gamma$, the term $M$ has type $C$, and similarly for the judgements
for values and handlers. As usual for qualified type systems, the
typing rules depend on an entailment relation $P \vdash \pi$ (and an
auxiliary relation $P\vdash\Gamma\linle\tau$), discussed in the
following section.

Rule \qtylab{Let} demonstrates the treatment of linearity in \CalcS.
We divide the context in three: $\Gamma_1$ is used exclusive in the
bound term $V$, $\Gamma_2$ is used exclusively in the body $M$, and
$\Gamma$ is used in both (and so its types must be unlimited).

Rule \qtylab{Do} demonstrates the use of constraints in \CalcS to
generalise subtyping between effect rows. It states that if $V$ is a
value of type $A_\ell$, then $\Do \ell \; V$ has result type $B_\ell$
and effect row $R$.
We assume that the parameter and result types of operations are given by
an implicit global context $\Pi = \{\ell_1:A_{\ell_1}\sto B_{\ell_1},
\cdots\}$.
$R$ must license effect $\ell$.  We again rely on entailment: the
constraints $P$ must be sufficient to show that the singleton row
$\{\ell\lsig A_\ell\sto^Y B_\ell\}$ is contained within $R$.

Rule \qtylab{Seq} demonstrates the remaining novelty of qualified types in
\CalcS.  Several of its uses of entailment follow the previous patterns.  The
bindings in $\Gamma$ are available in both $M$ and $N$, so $P \vdash \Gamma
\linle \Unl$ requires that their types be unlimited.  We want flexibility in
combining the effects in $M$ and $N$, so the conditions $P \vdash R_i \rle R$
assure that the effects of each are included in the effects of the entire
computation.
This allows us to avoid having to unify row types in examples like
$\dec{sandwichClose}$ (\Cref{sec:qualified-effect-types}) which
causes inaccuracy for tracking control-flow linearity.
Finally, $N$ is in the continuation of all operations in $M$, so the
value linearity of types in $\Gamma_2$ must be less than the
control-flow linearity of operations in $R_1$.
Note that the two kinding judgements in \tylab{Seq} in
\Cref{fig:typing} are now combined into one entailment judgement
$P\vdash\Gamma_2\linle R_1$.
The duality we have identified between value linearity and
control-flow linearity is reflected by the fact that value types
appear on the left of $\linle$ and effect row types appear on the
right.

Rule \qtylab{Handler} uses the lacking predicate $P\vdash
R_1\rlack\{\ell_i\}_i$ to ensure that the handled operations are not
in the remaining part of the input effect row $R_1$, and requires
$R_1$ to be a sub-row of the output effect row $R_2$.
This is used to allow the handled operations $\ell_i$ to appear in
$R_2$.

\subsection{Entailment}

\begin{figure}[htb]
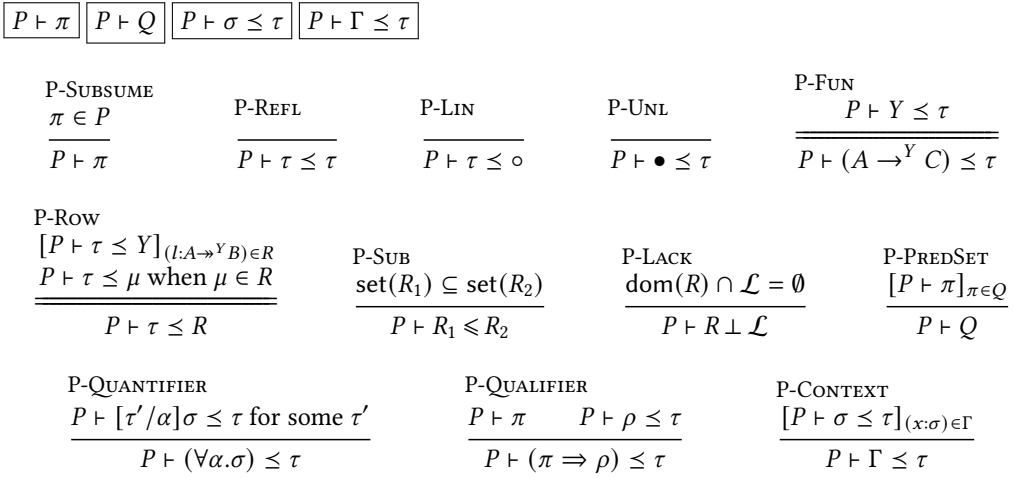

  \raggedright
  \boxed{P \vdash \pi}
  \boxed{P \vdash Q}
  \boxed{P \vdash \sigma \linle \tau}
  \boxed{P \vdash \Gamma \linle \tau}
  \hfill

  \begin{mathpar}
  \inferrule*[Lab=\plab{Subsume}]
  {\pi \in P}
  {P \vdash \pi}

  \inferrule*[Lab=\plab{Refl}]
  { }
  {P \vdash \tau \linle \tau}

  \inferrule*[Lab=\plab{Lin}]
  { }
  {P \vdash \tau \linle \Lin}

  \inferrule*[Lab=\plab{Unl}]
  { }
  {P \vdash \Unl \linle \tau}

  {\mprset{fraction={===}}
  \inferrule*[Lab=\plab{Fun}]
  {P \vdash Y \linle \tau}
  {P \vdash (A\to^Y C) \linle \tau }
  }

  {\mprset{fraction={===}}
  \inferrule*[Lab=\plab{Row}]
  {[P \vdash \tau \linle Y]_{(l\lsig A\sto^Y B) \in R} \\\\
   {P}\vdash{\tau \linle \mu} \text{ when } \mu\in R
  }
  {P \vdash \tau \linle R }
  }

    \inferrule*[Lab=\plab{Sub}]
  {
    \rowset{R_1} \subseteq \rowset{R_2}
  }
  {P \vdash R_1 \rle R_2}

  \inferrule*[Lab=\plab{Lack}]
  {\dom{R} \cap \mL = \emptyset}
  {P \vdash R \rlack \mL}

  \inferrule*[Lab=\plab{PredSet}]
  {[P \vdash \pi]_{\pi \in Q}}
  {P \vdash Q}

  \inferrule*[Lab=\plab{Quantifier}]
  {P \vdash [\tau'/\alpha] \sigma \linle \tau \text{ for some } \tau'}
  {P \vdash (\forall\alpha . \sigma) \linle \tau}

  \inferrule*[Lab=\plab{Qualifier}]
  {
  P \vdash \pi \\ P \vdash \rho \linle \tau
  }
  {P \vdash (\pi\qto\rho) \linle \tau}

  \inferrule*[Lab=\plab{Context}]
  {[P \vdash \sigma \linle \tau]_{(x:\sigma)\in\Gamma}}
  {P \vdash \Gamma \linle \tau}
\end{mathpar}

\caption{Entailment Relations for Predicates and other Judgement Relations}
\label{fig:qe-predicates-entailment}
\end{figure}

\Cref{fig:qe-predicates-entailment} defines the entailment relations
between predicates $P \vdash Q$.
It also defines an auxiliary entailment relation $P \vdash \Gamma
\linle \tau$ which compares the linearity of all variables in $\Gamma$
and $\tau$.
The algorithmic version of these relations will be given in
\Cref{sec:factorising}.

These two entailment relations are both defined as the conjunction of
sub-relations as indicated by \plab{PredSet} and \plab{Context}.
For $P\vdash Q$, we only need to use entailment relations of the form
${P}\vdash{\pi}$.
The \plab{Subsume} is standard.
The linearity predicate $\linle$ is reflexive (\plab{Refl}), with
$\Lin$ as top (\plab{Lin}) and $\Unl$ as bottom (\plab{Unl}) elements.
The two-way rules \plab{Fun} and \plab{Row} define the linearity of
functions and rows.
We make use of the fact that in the linearity predicates generated by
typing rules, functions only appear on the left, and rows only appear
on the right.
Here we do not include entailment rules for base types, but in
practice we would have axioms like $P\vdash\Int\linle\Unl$ and
$P\vdash\Lin\linle\File$.
For row predicates, we write $\rowset{R}$ for the set of all elements
(comprising operation labels with their signatures and row variables)
of $R$, and $\dom{R}$ for the set of all labels of $R$.
We define the row predicates directly by set operations (\plab{Sub}
and \plab{Lack}).

The entailment relation $P\vdash\Gamma\linle\tau$ is defined using
${P}\vdash{\sigma\linle\tau}$ which compares the linearity of a type
scheme $\sigma$ and a type $\tau$.
Our treatment of the linearity of type schemes is novel, and addresses a soundness bug in \Quill.
The rule \plab{Quantifier} which characterises the linearity of
polymorphic types may be surprising.
It states that the linearity of a polymorphic type
$\forall\alpha.\sigma$ is less than $\tau$ if there exists an
instantiation of it whose linearity is less than $\tau$.
This is because the linearity of a polymorphic
type should capture the linearity of values that inhabit that type.
A value of a polymorphic type can be understood as the intersection of
values of all possible instantiations of the type.
If one of these instantiation gives a type that is less linear
than $\tau$, then the value itself must be less linear than $\tau$ no
matter what other instantiations are.
For example, consider the identity function $\dec{id} = \lambda x.\Ret
x$ which is obviously unlimited.
We give $\dec{id}$ a polymorphic type
$\forall\phi\,\alpha\,\mu.\,\alpha\to^\phi\alpha\eff\{\mu\}$ to make
it possible to use it as both a linear function (by instantiating
$\phi$ to $\Lin$) and an unlimited function (by instantiating $\phi$
to $\Unl$).
Thus, we have expressive principal types for $\dec{id}$ without adding subtyping
between linearity types to the type system.

The rule \plab{Qualifier} may also be surprising.
To compare the linearity of a qualified type $\pi\qto\rho$ with
$\tau$, we require the predicate $\pi$ to hold and then compare
the linearity of the remaining part $\rho$ with $\tau$.
At first glance, the condition $P\vdash\pi$ may seem unnecessary: if $\pi$ must
hold in instantiations of this type, surely we can assume it in checking the
type's linearity.
However, particularly in local definitions, predicates may mention type
variables \emph{not} quantified in those schemes.  We do not want to assume
anything about the instantiation of those variables.
Consider the following function.
\[\bl
  \lambda x . \Let f=\lambda\Unit.x \In \Ret (f, f) \\
\el\]
The polymorphic function $f$ can be given the principal type $\sigma =
\forall\phi\,\mu.(\alpha\linle\phi)\qto\TUnit\to^\phi\alpha\eff\{\mu\}$
where $\alpha$ is the type of $x$.
Note that the constraint mentions $\alpha$, which is bound outside this type
scheme.
Then, since $f$ is duplicated in $\Ret (f,f)$, the typing of it
collects the constraint $\sigma\linle\Unl$.
Obviously, we want to know from $\sigma\linle\Unl$ that $\alpha$
should be unlimited since $x$ is also duplicated.
One possible derivation of $P\vdash\sigma\linle\Unl$ is shown as
follows.
\begin{mathpar}
  \inferrule*[right=\plab{Quantifier}]
  {
    \inferrule*[right=\plab{Qualifier}]
    {
      P\vdash \alpha\linle\phi' \\
      \inferrule*[Right=\plab{Function}]
      { P\vdash \phi'\linle\Unl }
      { P\vdash \TUnit\to^{\phi'}\alpha\eff\{\mu'\} \linle \Unl }
    }
    { P\vdash (\alpha\linle\phi')\qto\TUnit\to^{\phi'}\alpha\eff\{\mu'\} \linle \Unl }
  }
  {P\vdash (\forall\phi\,\mu.(\alpha\linle\phi)\qto\TUnit\to^\phi\alpha\eff\{\mu\}) \linle \Unl}
\end{mathpar}
In \plab{Quantifier} we instantiate $\phi$ and $\mu$ with variables
$\phi'$ and $\mu'$.
In order to prove $\sigma\linle\Unl$ from $P$, we must then prove
$\alpha\linle\phi'$ and $\phi'\linle\Unl$.
Note that $\phi'$ and $\mu'$ are not fresh, but should instead appear
in $P$, e.g., we might have $P =
\{\alpha\linle\phi',\phi'\linle\Unl\}$.
If we instead assumed $\alpha\linle\phi$, or removed the condition
entirely from \plab{Qualifier}, then $P$ would not need to restrict
$\alpha$ at all.  We could later instantiate $\alpha$ with a linear
type, say $\File$, and use this term to unsoundly copy file handles.

Readers may worry that the \plab{Qualifier} rule is as general as it
could be, because it always requires $P\vdash\pi$.
For example, consider $\Let f = V \In M$ where $f:\sigma$ does not
appear freely in $M$. We collect the constraint $\sigma\linle\Unl$.
Constraints of $V$ that are captured in $\sigma$ do not
necessarily need to be satisfied, because $f$ is not used.
However, we believe that binding unsatisfiable values has little
benefits and can hide potential bugs in practice.

Note that these entailment rules are intentionally made as simple as
possible. For example, we do not include any transitivity rules.
The entailment rules also do not check potentially conflicted
predicates in predicate sets since the rule \plab{Subsume} allows
collecting any predicates.
We say that predicate set $P$ is satisfiable if there exists a
substitution $\theta$ such that $\cdot \vdash \theta P$, and define
the solutions of it as $\satsubst{P} = \{\theta \mid \cdot \vdash
\theta P\}$.
Transitivity of $\linle$ is admissible when considering the solutions
of predicates, e.g., $\satsubst{\phi_1\linle\phi_2, \phi_2\linle\Unl}
= \satsubst{\phi_1\linle\phi_2, \phi_2\linle\Unl, \phi_1\linle\Unl} =
\{[\Unl/\phi_1,\Unl/\phi_2]\}$.
In \Cref{sec:constraint-solving}, we will give an algorithm to check
the satisfiability of constraint sets.

\subsection{Type Inference}
\label{sec:type-inference}

\Cref{fig:qeffpopsub-type-inference-abbrev} shows representative type inference
rules for \CalcS; the remainder are given in full in \Cref{app:qfp-inference}.
Our type inference algorithm is based on Algorithm W \citep{DamasMilner}
extended for qualified types \citep{jones94}.  In
$\infer{\Gamma}{V}{A}{\theta}{P}{\Sigma}$, the input includes the
current context $\Gamma$ and value $V$, and the output includes the
inferred type $A$, substitution $\theta$, predicate set $P$, and
variable set $\Sigma$ of used term variables.
Note that the predicates $P$ are an output of inference, not an input; rather
than checking entailment, as the syntax-directed type rules do, we will emit a
constraint set sufficient to guarantee typing.  In the next section, we discuss
our algorithm to guarantee that inferred constraint sets are not unsatisfiable.
As usual, the
substitution $\theta$ has been already applied to $A$ and $P$.

\begin{figure}[htb]
  \raggedright
  \boxed{\infer{\Gamma}{V}{A}{\theta}{P}{\Sigma}}
  \boxed{\infer{\Gamma}{M}{C}{\theta}{P}{\Sigma}}
  \boxed{\infer{\Gamma}{H}{C \hto D}{\theta}{P}{\Sigma}}

  \begin{mathpar}
  \inferrule*[Lab=\tilab{Let}]
  {
    \infer{\Gamma}{V}{A}{\theta_1}{P_1}{\Sigma_1} \\
    \sigma = \gen(\theta_1 \Gamma, P_1 \qto A) \\\\
    \infer{\theta_1\Gamma, x:\sigma}{M}{C}{\theta_2}{P_2}{\Sigma_2} \\\\
    Q = \makeunl{\theta_2\theta_1\Gamma |_{\Sigma_1 \cap \Sigma_2}} \cup
        \makeunl{\theta_2(x:\sigma) |_{\setcomplement{\Sigma_2}}}
  }
  {\infer{\Gamma}{\Let x = V \In M}{C}
    {\theta_2\theta_1}{P_2\cup Q}{\Sigma_1 \cup (\Sigma_2 \backslash x)}}

\inferrule*[Lab=\tilab{Do}]
{
  \infer{\Gamma}{V}{A}{\theta_1}{P}{\Sigma} \\
  \unify{A}{A_\ell}{\theta_2} \\\\
  \mu, \phi \fresh \\
  Q = \makesub{(\ell: A_\ell \sto^\phi B_\ell), \mu} \\
}
{\infer{\Gamma}{\Do \ell\; V}
  {B_\ell \eff \{\mu\}}{\theta_2\theta_1}{\theta_2 P\cup Q}{\Sigma}}

  \inferrule*[Lab=\tilab{Seq}]
  {
    \infer{\Gamma}{M}{A\eff\{R_1\}}{\theta_1}{P_1}{\Sigma_1} \\
    \infer{\theta_1\Gamma,x:A}{N}{B\eff\{R_2\}}{\theta_2}{P_2}{\Sigma_2} \\
    \mu \fresh \\\\
    Q = \makeunl{\theta_2\theta_1\Gamma |_{\Sigma_1 \cap \Sigma_2}} \cup
        \makeunl{\theta_2(x:A) |_{\setcomplement{\Sigma_2}}} \cup
        \makeleq{\theta_2\theta_1\Gamma |_{\Sigma_2}, \theta_2 R_1} \cup
        \makesub{\theta_2 R_1, \mu} \cup
        \makesub{R_2, \mu} \\
  }
  {\infer{\Gamma}{\Let x \revto M \In N}{B\eff \mu}
    {\theta_2\theta_1}{\theta_2 P_1 \cup P_2 \cup Q}{\Sigma_1 \cup (\Sigma_2 \backslash x)}}
  \end{mathpar}

  \[
    \makeleq{\Gamma, \tau} = \factor(\Gamma \linle \tau) \qquad
    \makeunl{\Gamma} = \makeleq{\Gamma, \Unl} \qquad
    \makesub{R_1, R_2} = \factor(R_1 \rle R_2)
  \]

  \caption{Selected Type Inference Rules for \CalcQ}
  \label{fig:qeffpopsub-type-inference-abbrev}
\end{figure}

Rule \tilab{Let} demonstrates the treatment of linearity.
We write $\Gamma |_\Sigma$ for the type context generated by
restricting $\Gamma$ to variables in $\Sigma$.
We begin by inferring types for $V$ and $M$.  Variable sets $\Sigma_1$
and $\Sigma_2$ capture those variables used in each; any variable in
$\Sigma_1 \cup \Sigma_2$ must be unlimited.  We also account for the
possibility that the variable $x$ may not be used in $M$---that is to
say, that it may appear in $\setcomplement{\Sigma_2}$, the complement
of the used variables $\Sigma_2$.  We generate the corresponding
unlimitedness constraints using the auxiliary function
$\mathsf{factorise}$, discussed next.
Rule \tilab{Do} emits the constraint that the singleton effect row be included
in the output row.  Rule \tilab{Seq} combines these techniques.

We prove soundness and completeness of type inference with respect to
the syntax-directed type system.
We write $\theta |_\Gamma$ for the substitution generated by
restricting the domain of $\theta$ to the free variables in $\Gamma$
and $(\theta = \theta')|_\Gamma$ for $\theta |_\Gamma = \theta'
|_\Gamma$.

\begin{restatable}[Soundness]{theorem}{soundness}
  If $\infer{\Gamma}{V}{A}{\theta}{P}{\Sigma}$, then
  $\typ{P \mid \theta \Gamma |_\Sigma}{V : A}$. The
  same applies to computation and handler typing.
\end{restatable}

\begin{restatable}[Completeness]{theorem}{completeness}
  If $\typ{P\mid \theta\Gamma}{V : A}$, then
  $\infer{\Gamma}{V}{A'}{\theta'}{Q}{\Sigma}$ and there exists
  $\theta''$ such that $A = \theta'' A'$, $P \vdash \theta'' Q$, and
  $(\theta = \theta'' \theta')|_\Gamma$. The same applies to
  computation and handler typing.
\end{restatable}

\noindent
The proofs can be found in \Cref{app:proof-type-inference} and depend
on the correctness of $\factor$, discussed next.
Note that we do not need to incorporate the subtyping relation into
the statement of the completeness theorem because we only have
subtyping between row types and do not allow implicit subsumption
(unlike traditional subtyping systems).

\subsection{Factorising Predicates}
\label{sec:factorising}

\begin{figure}[htb]
  \begin{minipage}[t]{0.48\textwidth}
  \[
  \ba{@{}l@{}}
  \factor : \meta{Pred} \to \meta{PSet} \\
  \factor(\tau \linle \tau) = \emptyset \\
  \factor(\tau \linle \Lin) = \emptyset \\
  \factor(\Unl \linle \tau) = \emptyset \\
  \factor(A\to^Y C \linle \tau) = \factor(Y \linle \tau) \\
  \factor(\tau \linle \cR\lsep\mu) = \\
  \quad \factor(\tau\linle \cR)
    \cup \factor(\tau\linle \mu) \\
  \factor(\tau \linle \cR) = \\
  \quad \bigcup_{(\ell:A\sto^Y B)\in \cR} {\factor(\tau\linle Y)} \\
  \factor(R_1 \rle R_2) = \emptyset, \text{ when } \rowset{R_1}\subseteq\rowset{R_2} \\
  \factor(R \rlack \mL) = \emptyset, \text{ when } \dom{R}\cap\mL = \emptyset \\
  \factor(\pi) = \pi
  \ea
  \]
  \end{minipage}
  \begin{minipage}[t]{0.48\textwidth}
  \[
  \ba{@{}l@{}}
  \factor : (\meta{TySch} \linle \meta{Type}) \to \meta{PSet} \\
  \factor((\forall\alpha.\sigma) \linle \tau) =\\
  \quad \factor([\beta/\alpha] \sigma \linle \tau)
    \text{ for some fresh }\beta \\
  \factor((\pi\qto\sigma) \linle \tau) = \\
  \quad \factor(\pi) \cup \factor(\sigma\linle\tau)
  \bigskip \\

  \factor : (\meta{Env} \linle \meta{Type}) \to \meta{PSet} \\
  \factor(\Gamma \linle \tau) =
  \bigcup_{(x:\sigma)\in \Gamma} {\factor(\sigma\linle\tau)}
  \bigskip \\

  \factor : \meta{PSet} \to \meta{PSet} \\
  \factor(P) = \bigcup_{\pi\in P} {\factor(\pi)}
  \ea
  \]
  \end{minipage}

\caption{Factorisation of Constraints}
\label{fig:constraints-factorisation}
\end{figure}

\noindent
The $\factor$ function is defined in \Cref{fig:constraints-factorisation}; it factors
constraints into simpler predicates following the entailment rules in
\Cref{fig:qe-predicates-entailment}.
We use $K$ to represent rows consisting of only operation labels.

The only surprising case is for $(\forall\alpha.\sigma)\linle\tau$.
Rule \plab{Quantifier} requires that we find some instance such that
$\sigma[\tau'/\alpha] \linle \tau$.  Rather than search for such an instance, we
simply pick a fresh type variable $\beta$.
As a result, our type inference algorithm is likely to produce \emph{ambiguous}
type schemes, in which quantified type variables appear \emph{only} in
predicates. Such type schemes are typically rejected \citep{jones94}, as the
meaning of ambiguously typed terms is undefined.  However, as our linearity
predicates do not have any intrinsic semantics, but only constrain the use of
terms, we do not believe these constraints lead to semantic ambiguity.
One interesting property of $\factor$ is that the linearity predicates
in its results are only between value type variables $\alpha$, row
type variables $\mu$, and linearity types $Y$.

We prove the correctness of $\factor$ with respect to the entailment
rules in \Cref{fig:qe-predicates-entailment}.
\begin{restatable}[Correctness of factorisation]{theorem}{constraintFactorisation}
\label{thm:correctness-of-factorisation}
If $\factor(P) = Q$, then $Q \vdash P$ and $P\vdash Q$.
If $\factor(\Gamma\linle\tau) = Q$, then $Q\vdash\Gamma\linle\tau$ and
for any $P\vdash\Gamma\linle\tau$,
there exists $\theta$ such that $P\vdash \theta Q$.
\end{restatable}
The proof can be found in \Cref{app:factorisation}.

\subsection{Constraint Solving}
\label{sec:constraint-solving}

Finally, we must check that inferred constraint sets are satisfiable; we do not
want to conclude that a program is well-typed, but only under the assumption
that a linear type is unlimited.

We define a constraint solving algorithm $\solve(P)$ for checking the
satisfiability of the predicate set $P$, inspired by solving
algorithms for general subtyping
constraints~\citep{Pottier01,pottier1998,pretnar14}.
The tricky part compared to solving usual subtyping constraints is
that we need to carefully deal with the interaction between row
subtyping constraints and linearity constraints.
For instance, $R_1\rle R_2$ and $\tau\linle R_2$ actually implies
$\tau\linle R_1$.
To resolve the interaction, the algorithm proceeds by first
transforming row subtyping constraints to those of the forms
$\mu\rle R$, so that we can always simply instantiate $\mu$ on the
left to the empty row $\cdot$ for which $\tau\linle\cdot$ always
holds.
Then, the algorithm computes the transitive closure of linearity
constraints and rejects $\Lin \linle \Unl$.
The full algorithm is given in \Cref{app:qfp-constraint-solving}.
We have the following theorem on the correctness of the constraint
solving algorithm, in which we write $\satsubst{P}\theta$ for the substitution set
$\{\theta'\theta\mid \theta'\in\satsubst{P}\}$.
\begin{restatable}[Correctness of constraint solving]{theorem}{constraintSolving}
  For any constraint set $P$ generated by the type inference of
  \CalcS, $\solve(P)$ always terminates.
\begin{itemize}
  \item If it fails, then $P$ is not satisfiable.
  \item If it returns $(\theta, Q)$, then $P$ is satisfiable and
  $\satsubst{P} = \satsubst{Q}\theta$.
\end{itemize}
\end{restatable}
The proof can be found in \Cref{app:proof-constraint-solving}, whose
main idea is to show that every step of the algorithm preserves
solutions, and the output predicate set has one solution.

We leave the design of constraint simplification algorithms as
practical concerns.
Some existing algorithms on simplifying general subtyping constraints
 are promising~\cite{pottier1998, Pottier01}. %

\section{Shallow Handlers}
\label{sec:extensions}

Up to now we have concentrated on \emph{deep} effect handlers, which
wrap the original handler around the body of captured continuations.
Given this automatic reuse of the handler, the handler itself cannot
capture any linear resources.
In contrast, shallow handlers~\citep{KammarLO13,HillerstromL18} do not
wrap the original handler around the body of captured continuations,
which means shallow handlers \emph{can} capture linear resources and
thus influence control-flow linearity.
In this section, we discuss the extensions of \Calc and \CalcS with
shallow handlers and their challenges.

Let us first consider shallow handlers in \Calc.
We write $H^\dagger$ for a shallow handler.
The only difference in the operational semantics is the new
\semlab{Op$^\dagger$} rule for handling with shallow handlers.
\begin{reductions}
  \semlab{Op$^\dagger$} &
    \Handle \EC[(\Do \ell \; V)^E] \With H^\dagger
      &\reducesto& N[V/p, (\lambda^Y y^B.\EC[(\Ret y)^E])/r], \\
  \multicolumn{4}{@{}r@{}}{
        \text{where } \ell \notin \BL{\EC},
        (\ell \; p \; r \mapsto N) \in H^\dagger \text{ and }
        (\ell:A\to^Y B) \in E
  }
\end{reductions}
Unlike in \semlab{Op}, the body of the continuation is not handled by
$H^\dagger$.
Whereas deep handlers perform a fold over a computation trees shallow
handlers perform a case-split.
As such, we know that exactly one operation clause or the return
clause will be invoked, and providing all allowed operations are
linear each clause may capture the same linear resources.
The typing rule is as follows.
\begin{mathpar}
  \inferrule*[Lab=\tylab{ShallowHandler}]
    {
      H = \{\Ret x \mapsto M\} \uplus \{ \ell_i\;p_i\;r_i \mapsto N_i \}_i
      \\\\
      C = A \eff \{(\ell_i : A_i \sto^{Y_i} B_i)_i; {R}\} \\
      D = B \eff \{(\ell_i : P)_i; {R}\}\\\\
      \hlbox{\Delta\vdash \Gamma:Y}\\
      \hlbox{\Delta\vdash R:Y}\\
      \typ{\Delta;\Gamma, x : A}{M : D}\\\\
      [\typ{\Delta;\Gamma, p_i : A_i, r_i : B_i \to^{Y_i} {C}}{N_i : D}]_i
    }
    {\typ{\Delta;\Gamma}{H^\dagger : C \hto D}}
\end{mathpar}
Instead of requiring value linearity of $\Gamma$ to be unlimited as in
the deep handler rule \tylab{Handler}, we require the value linearity
of $\Gamma$ to coincide with the control-flow linearity of $R$, the
effect row of the unhandled operations.
This is because the shallow handler may be captured as part of the
continuations of these unhandled operations in outer handlers.
Concretely, when $Y=\Lin$, the shallow handler may use linear
variables from the context, and unhandled operations are control-flow
linear; when $Y=\Unl$, the shallow handler cannot use any linear
variables from the context, and we have no restriction on the
control-flow linearity of unhandled operations.

We can also easily extend \CalcS with shallow handlers.
\begin{mathpar}
  \inferrule*[Lab=\qtylab{ShallowHandler}]
  {
    H = \{\Ret x \mapsto M\} \uplus \{ \ell_i\;p_i\;r_i \mapsto N_i \}_i\\\\
    C = A \eff \{(\ell_i : A_i \sto^{Y_i} B_i)_i; R_1\} \\
    D = B \eff \{R_2\} \\\\
    \typ{P\mid\Gamma, x : A}{M : D}\\
    [\typ{P\mid\Gamma, p_i : A_i, r_i : B_i \to^{Y_i} {C}}{N_i : D}]_i \\\\
    \hlbox{P \vdash \Gamma \linle R_1} \\
    P \vdash R_1 \rle R_2 \\
    P \vdash R_1 \rlack \{\ell_i\}_i
  }
  {\typ{P\mid\Gamma}{H : C \hto D}}
\end{mathpar}
In place of $P\vdash\Gamma\linle\Unl$ in \qtylab{Handler}, we have
$P\vdash\Gamma\linle R_1$, which restricts the value linearity of the
type context to be less than the control-flow linearity of unhandled
operations in $R_1$.

Shallow handlers are typically used together with recursive functions
to implement more general recursive behaviours than the structural
recursion of deep handlers. It is straightforward to extend \Calc and
\CalcS with recursive functions \citep{fpop,HillerstromLA20}.
Obviously recursive functions are themselves unlimited so cannot
capture linear resources, but that does not preclude explicitly
threading a linear resource through a recursive function that installs
a shallow handler.
We use the syntax $\Rec f\,x.M$ to define a recursive function $f$
with parameter $x$ and function body $M$. The typing rules and
semantics rule for it in \Calc and \CalcS are as follows.
\begin{mathpar}
  \inferrule*[Lab=\tylab{Rec}]
  {
    \typ{\Delta;\Gamma, f:A\to^\Unl C, x:A}{M:C} \\
    \Delta\vdash \Gamma : \Unl
  }
  { \typ{\Delta;\Gamma}{\Rec f^{A\to^\Unl C}\, x.M : A\to^\Unl C} }

  \inferrule*[Lab=\qtylab{Rec}]
  {
    \typ{\Delta;\Gamma, f:A\to^\Unl C, x:A}{M:C} \\
    P\vdash \Gamma \linle \Unl
  }
  { \typ{P\mid\Gamma}{\Rec f\,x.M : A\to^\Unl C} }

\end{mathpar}
\begin{reductions}
  \semlab{Rec}   & (\Rec f\,x.M)\,V &\reducesto & M[(\Rec f\,x.M)/f, V/x]
\end{reductions}

As an example, we can write the following recursive function
$\dec{withFile}\,f$ which takes a file handle $f$ and interprets all
$\Print$ operations in $M$ as writing to file $f$.
\[\ba{l@{\;}c@{\;}l@{\;}l}
  \dec{withFile}\,f &=&
  \Rec \dec{withFile}\,f .\Handle M \With \span \\
    & & \{
      \Ret x &\mapsto \Close\,f \Seq x
      \\
    & & \phantom{\{}
      \Print\,s\,r &\mapsto \Let f'\revto \WriteFile\,(s,f) \In \dec{withFile}\,f'\,r
    \} \\
\ea\]
Note that this example can also be implemented with a deep handler by
requiring the handler to return a function which takes the file handle
as a parameter.
Shallow handlers provide us with a more direct programming style.

Although our two new typing rules are straightforward and entirely
backward compatible with the current systems, shallow handlers can
actually introduce more challenges to track control-flow linearity.
This is essentially because shallow handlers are more flexible than
deep handlers and do not handle all invocations of the same operation
uniformly.
With only deep handlers, it is natural for all invocations of an
operation to have the same control-flow linearity as they are handled
by the same handler.
However, with shallow handlers, different invocations of the same
operation can be handled by different handlers, resulting in different
control-flow linearity.
For example, consider the following program $\dec{hesitantClose}$
which makes choices before and after closing the file $f$.
\[\ba{l@{\;}c@{\;}l}
  \dec{hesitantClose} &=& \lambda f . \Do \Choose\,\Unit; \CloseFile\,f; \Do \Choose\,\Unit
\ea\]
The continuation of the first $\Choose$ contains the linear file
handle $f$, whereas the second one does not.
Technically, the handler for the second $\Choose$ can resume any
number of times.
However, neither the effect system of \Calc nor that of \CalcS is able
to ascribe a different control-flow linearity to the two invocations
of $\Choose$, which means we must handle both invocations linearly.
One potential solution is to track the order and duplication of
effects in the effect system.
However, this kind of information is known to be too cumbersome for
effect systems.
A more lightweight solution is to exploit named
handlers~\citep{XieCIL22,BiernackiPPS20} to assign $\Choose$
operations in different positions to different shallow handlers.
We leave the design of an ergonomic and expressive effect system for
tracking control-flow linearity of shallow handlers to future work.

\section{Related Work}
\label{sec:related-work}

\newcommand{\csharp}{\text{C}\#}

\paragraph{Linear Resources and Control Effects}
Exception handlers with finally clauses are a common way of managing
linear resources. Exception handlers provide a form of unwind
protection, which enables the programmer to supply the logic to
release acquired resources in the finally clause, which gets executed
irrespective of whether a fault occurs.
Similarly, the \lstinline$defer$ statement in \Go{}~\citep{DonovanK15}
defers the execution of its operand until the defining function
returns either successfully or via a fault. Thus the programmer can
conveniently acquire a particular resource and include the deferred
logic for releasing it on the next line of code.
Another variation is automatic resource block management as in the \Cpp{}
RAII idiom~\citep{CombetteM18} and \Java{}'s
\lstinline$try-with-resource$~\citep{GoslingJSBBSB23}, both of which
offer a means for automatically acquiring and releasing resources in
the static scope.
In \Scheme{} the fundamental resource protection mechanism is the
procedure \lstinline$dynamic-wind$~\citep{FriedmanH85}.
It is a generalisation of unwind protection intended to be used in the
presence of first-class control, where control may enter and leave the
same computation multiple times.
It takes three functional arguments: the first is the resource
acquisition procedure, which gets applied when control enters
\lstinline$dynamic-wind$; the second is the main computation, which
may use the acquired resources; and the third is the resource release
procedure, which is applied when control is about to leave
\lstinline$dynamic-wind$.

\citet{BrachthauserL23} present a constraint system based on qualified
types for programming with multi-shot effect handlers and linear
resources in \Koka. They use these constraints to mark some effects as
linear. However, they do not include a linear type system and instead
rely on pre-declaring the linearity of operations (i.e., no inference
for control-flow linearity) and a syntactic check to ensure that
resumptions are not invoked more than once.
Compared to the qualified effect system of \CalcS, their system does
not support effect subtyping and abstraction over linearity.

\newcommand{\ural}{\ensuremath{\lambda^{\mathsf{URAL}}}\xspace}
\newcommand{\uralc}{\ensuremath{\lambda^{\mathsf{URAL}}(\mathscr{C})}\xspace}
\paragraph{Structural Types and Control Effects}

\citet{TovP11} propose a calculus \uralc which extends the
substructural $\lambda$-calculus \ural~\citep{AhmedFM05} with abstract
control effects $\mathscr{C}$ given by a set of effects, a pure
effect, and an effect-sequencing operator.
They show how to instantiate \uralc with concrete control effects
including exceptions and shift/reset~\citep{DanvyF90} separately.
Similar to \Calc and \CalcS, the \uralc calculus also uses
type-and-effect system to check that control effects do not violate
the substructural usage guarantees for values.
It includes a judgement on effect types to determine whether control
effects may discard or duplicate their continuations, which roughly
corresponds to our notion of control-flow linearity.
The main difference between our work and \uralc is that we consider
the tracking of control-flow linearity in the presence of algebraic
effects and effect handlers, which are more involved than exceptions
and shift/reset both statically and dynamically.
While it is theoretically possible to instantiate \uralc to effect
handlers, this task is itself highly non-trivial due to the richer
effect systems of effect handlers.
Conversely, we can also easily encode exceptions and shift/reset as
user-defined effects in \Calc and \CalcS using effect
handlers~\citep{ForsterKLP19, PirogPS19}.

\paragraph{Linear Type Systems}
Type inference with linear types is a well-studied area. \citet{fpop}
propose using kinds to track linearity, using subkinding to enable
polymorphism over linearities. \citet{alms} develop an expanded
approach to tracking structural restrictions in kinds; among other
differences they introduce subtyping for function types and require
fewer explicit linearity annotations than \citeauthor{fpop}.
\citet{GanTM14} use qualified types to characterise types that admit
structural rules in a substructural type system: for example, in a
linear type system, unlimited types are exactly types $\tau$ that
support operations $\mathsf{dup} : \tau \to (\tau, \tau)$ and
$\mathsf{drop} : \tau \to ()$.
\citet{quill} extends the approach of \citeauthor{alms} to generalise
the treatment of function types, introducing the linearity ordering
constraint $\tau \linle \upsilon$; he also generalises their
description of unlimited types to type schemes, but does so unsoundly.
In contrast, the current work does not interpret unlimited types via
operations like $\mathsf{dup}$ and $\mathsf{drop}$; we also avoid
Morris's unsoundness in the treatment of type schemes. An alternative
approach tracks linearity exclusively in function types, rather than
in kinds.  This approach is developed by \citet{GhicaS14},
\citet{McBride16}, and \citet{Atkey18}, and has been implemented in
Idris \citep{Brady21} and an extension to the GHC Haskell compiler
\citep{BernardyBNJS18}.

\paragraph{Row-based Effect Types}
Row types and row polymorphism are a popular way of implementing
effect systems in programming languages.
\Links \citep{linksrow} adopts R{\'e}my style row polymorphism
\citep{remy1994type}, where the row types are able to represent the
absence of labels and each label is restricted to appear at most once.
\Koka~\citep{koka} and \Frank~\citep{frank} use row polymorphism based
on scoped labels \citep{leijen2005extensible} which allows duplicated
labels.
We believe the idea of tracking control-flow linearity in \Calc should
work well with all kinds of different row-based effect systems.

\paragraph{Subtyping-based Effect Types}
Some versions of \Eff~\citep{bauer13, pretnar14} use an effect system
based on subtyping.
\citet{KarachaliasPSVS20} describe an explicit target calculus
\textsc{ExEff} with a subtyping-based effect system and a type
inference algorithm that elaborates \Eff source code into it.
\Eff uses a row-like representation of effect types and defines a
subtyping relation for effect types similar to the that of \CalcS.
One difference is that \Eff incorporates full subtyping relations
between all types and implicit subsumption, whereas we only introduce
subtyping between row types and allow explicit subsumption in
necessary positions (like \qtylab{Seq} and \qtylab{Handle}). In this
respect our qualified effect system is more lightweight.
Algebraic subtyping~\citep{Dolan16,DolanM17} combines subtyping and
parametric polymorphism with elegant principal types.
It would be interesting to explore the possibility of combining linear
types and effect types based on algebraic subtyping with control-flow
linearity.

\paragraph{One-shot control operators}
One-shot continuations were first introduced by \citet{FriedmanH85} in
the form of a linear variant of call/cc.
Similarly, \citet{Filinski92} considers a one-shot variant of the
$\mathcal{C}$ operator~\citep{FelleisenFKD87}.

\paragraph{One-shot Effect Handlers}
\OCaml{} 5~\citep{multicore}, the \Cpp{}-effects library~\citep{GhicaLBP22},
and the typed continuations proposal for adding effect handlers to
\Wasm{}~\citep{Phipps-CostinRGLHSPL23,wasm-typed-continuations}
all implement dynamically-checked one-shot effect handlers.
Continuations captured by such effect handlers can be thought of as
linear resources themselves, and thus play nicely with other linear
resources.
Any attempt to invoke a continuation more than once throws a runtime
error.
In contrast, our type systems can be used to statically ensure that
handlers are one-shot.
In fact, its considerably easier to build a system that ensures that
\emph{all} handlers are uniformly one-shot than a system like ours
that supports both one-shot and multi-shot handlers, as in the former
case there is no need to track the use of linear resources specially.
Another advantage of one-shot continuations is that they admit
efficient implementations which are compatible with linear resources,
as a one-shot continuation need not copy its underlying
stack~\citep{DBLP:conf/pldi/BruggemanWD96}.
\citet{HillerstromLL23} present a substructural type system for a
calculus with effect handlers based on dual intuitionistic linear
logic~\citep{Barber96} which restricts all effect handlers to be
one-shot (actually one- or zero-shot). They use it to show an
asymptotic performance gap between one-shot and multi-shot effect
handlers, but are not concerned with linear resources other than
continuations.

\paragraph{Multi-shot Effect Handlers}
\Eff~\citep{BauerP15}, \Effekt~\citep{effekt}, \Koka~\citep{koka}, and
\Helium~\citep{BiernackiPPS19} are research programming languages with
multi-shot handlers. In contrast to one-shot handlers, multi-shot
handlers can invoke the captured continuations an arbitrary number of
times. This enables a range of interesting applications. For instance,
asymptotic efficient backtracking search~\citep{HillerstromLL20},
nondeterminism~\citep{KammarLO13}, and UNIX fork-style
concurrency~\citep{Hillerstrom22} can all be given a direct semantics
in terms of multi-shot handlers.
However, one obstacle is that the aforementioned languages cannot
statically optimise uses of one-shot continuations, as they must
conservatively expect the ambient context to have nonlinear control
flow, thus requiring them to copy the continuation a
priori~\citep{HillerstromLS16,Hillerstrom16}.
Our type systems can enable static optimisation of one-shot
continuations through static identification of linear and nonlinear
contexts.

\section{Conclusion and Future Work}
\label{sec:conclusion}

We have explored the interplay between effect handlers and linear
types.
We have demonstrated that in order to soundly combine potentially
non-linear effect handlers with linear types, it is necessary to add a
mechanism for tracking control-flow linearity too. We incorporated
control-flow linearity into two quite different core languages as well
as realising control-flow linearity in \Links.

Directions for future work include:
  implementing a programming language based on \CalcS;
  developing more precise type systems for combining control-flow
  linearity with shallow handlers;
  combining control-flow linearity with other forms of effect type
  systems, such as those that support generative effects, duplicate
  effects, capabilities, and modal effect types;
  adapting the constraints of \CalcS to algebraic
  subtyping~\citep{DolanM17}; and
  adapting control-flow linearity for uniqueness types and for
  quantitive type theory~\citep{McBride16, Atkey18}.

\section*{Data Availability Statement}
The implementation of \Calc in \Links is available on
Zenodo~\citep{TangHLM23}.

\begin{acks}                            %
  This work was supported by the UKRI Future Leaders Fellowship
  ``Effect Handler Oriented Programming'' (reference number
  MR/T043830/1).
\end{acks}

\bibliography{reference}

\appendix
\section{Proofs of \Calc}

In this section, we prove the theorems in \Cref{sec:feffpop}.

\subsection{Unlimited is Unlimited}
\label{app:proof-unlimited-is-unlimited}

\unlisunl*

\begin{proof}~\\
\noindent \emph{1. Unlimited values are unlimited.} By induction on
the typing derivation $\typ{\Delta;\Gamma}{V:A}$.
\begin{description}
  \item[Case] \tylab{Var}. Trivial.
  \item[Case] \tylab{Abs}. $\Delta\vdash A\to^Y C:\Unl$ gives
  $Y=\Unl$, which then gives $\Delta\vdash\Gamma:\Unl$.
  \item[Case] \tylab{TAbs}. $\Delta\vdash \forall^Y \alpha^K . C:\Unl$
  gives $Y=\Unl$, which then gives $\Delta\vdash\Gamma:\Unl$.
\end{description}

\noindent \emph{2. Unlimited continuations are unlimited.}
By $\ell\notin\BL{\EC}$ and straightforward induction on typing
derivations, we have $C = \_\eff\{\ell : A\sto^\Unl B\lsep\_\}$.
By induction on $\typ{\Delta;\Gamma}{\EC[(\Do\ell\;V)^E]:C}$.
\begin{description}
  \item[Case]
  \begin{mathpar}
  \inferrule*[Lab=\tylab{Do}]
    {
      E = \{\ell : A \sto^{Y} B\lsep R\} \\\\
      \typ{\Delta;\Gamma}{V : A} \\
      \typ{\Delta}{E : \Effect}
    }
    {\typc{\Delta;\Gamma}{(\Do \ell \; V)^E : B}{E}}
  \end{mathpar}
  Immediately, we have $\typ{\Delta;y:B}{(\Ret y)^E} : B\eff E$ and
  $\Delta\vdash\cdot:\Unl$.

  \item[Case]
  \begin{mathpar}
  \inferrule*[Lab=\tylab{Seq}]
  { \refa{\typ{\Delta;\Gamma_1}{\EC'[(\Do\ell\;V)^E] : A'\eff E'}} \\
    \typ{\Delta;\Gamma_2, x : A'}{N : B'\eff E'} \\\\
    E' = \{\ell:A\sto^\Unl B\lsep R'\} \\
    \refb{\Delta \vdash \Gamma_2 : Y} \\
    \refc{\Delta \vdash (\ell:A\sto^\Unl B\lsep R') : Y} \\
  }
  {\typ{\Delta;\Gamma_1 + \Gamma_2}
  {\Lety^{Y} x \revto \EC'[(\Do\ell\;V)^E] \In N : B'\eff E'}}
  \end{mathpar}

  By \refc{}, we have $Y = \Unl$.
  Then, by \refb{}, we have $\Delta \vdash \Gamma_2 : \Unl$.
  By the IH on \refa{}, there exists $\Delta\vdash \Gamma_1 =
  \Gamma_{11} + \Gamma_{12}$ such that $\Delta\vdash\Gamma_{11}:\Unl$
  and $\typ{\Delta;\Gamma_{11}, y:B}{\EC'[(\Ret y)^E] : A'\eff E'}$.
  Applying \tylab{Seq} to it, we have
  $\typ{\Delta;\Gamma_3,y:B}{\Lety^Y x\revto \EC'[(\Ret y)^E] \In N :
  B'\eff E'}$, $\Delta\vdash\Gamma = \Gamma_{12} + \Gamma_3$ and
  $\Delta\vdash\Gamma_3:\Unl$ where $\Delta\vdash \Gamma_3 = \Gamma_2
  + \Gamma_{11}$.

  \item[Case]
  \begin{mathpar}
  \inferrule*[Lab=\tylab{Handle}]
  {
    \refa{\typ{\Delta;\Gamma_1}{\EC'[(\Do\ell\;V)^E] : A'\eff E'}} \\
    \refb{\typ{\Delta;\Gamma_2}{H : A'\eff E' \hto B'\eff F'}} \\
  }
  {\Delta;\Gamma_1 + \Gamma_2 \vdash \Handle \EC'[(\Do\ell\;V)^E] \With H : B'\eff F'}
  \end{mathpar}
  By \refb{}, we have $\Delta\vdash \Gamma_2 : \Unl$.
  By the IH on \refa{}, there exists $\Delta\vdash \Gamma_1 =
  \Gamma_{11} + \Gamma_{12}$ such that $\Delta\vdash\Gamma_{11}:\Unl$
  and $\typ{\Delta;\Gamma_{11}, y:B}{\EC'[(\Ret y)^E] : A'\eff E'}$.
  Applying \tylab{Handle} to it, we have
  $\typ{\Delta;\Gamma_3,y:B}{\Handle \EC'[(\Ret y)^E] \With H : B'\eff
  F'}$, $\Delta\vdash\Gamma = \Gamma_{12} + \Gamma_3$ and
  $\Delta\vdash\Gamma_3:\Unl$ where $\Delta\vdash \Gamma_3 = \Gamma_2
  + \Gamma_{11}$.

\end{description}

\noindent \emph{3. Deep handlers are unlimited.} Directly follows
from \tylab{Handler}.

\end{proof}

\subsection{Progress}
\label{app:proof-progress}

\begin{lemma}[Canonical forms]~
  \label{lemma:canonical-forms}
  \begin{enumerate}[label=\arabic*.]
    \item If $\typ{}{V:A\to^Y B}$, then $V$ is of shape $\lambda^Y x^A . M$.
    \item If $\typ{}{V:\forall^Y \alpha^K. C}$, then $V$ is of shape
    $\Lambda^Y \alpha^K . M$.
  \end{enumerate}
\end{lemma}

\begin{proof}
  Directly follows from the typing rules.
\end{proof}

\progress*

\begin{proof}
By induction on the typing derivation $\typ{}{M : A\eff E}$.
\begin{description}
  \item[Case]
  \begin{mathpar}
  \inferrule*[Lab=\tylab{App}]
    {\typ{}{V : A \to^{Y} C} \\
     \typ{}{W : A}
    }
    {\typ{}{V\,W : C}}
  \end{mathpar}
  By \Cref{lemma:canonical-forms}, we have $V = \lambda^Y x^A.M$.
  Reduced by \semlab{App}.

  \item[Case]
  \begin{mathpar}
  \inferrule*[Lab=\tylab{TApp}]
    {\typ{\Delta;\Gamma}{V : \forall^{\hl{Y}} \alpha^K . \, C} \\
     \Delta \vdash T : K
    }
    {\typ{\Delta;\Gamma}{V\,T : C[T/\alpha]}}
  \end{mathpar}
  By \Cref{lemma:canonical-forms}, we have $V = \Lambda^Y\alpha^K.M$.
  Reduced by \semlab{TApp}.

  \item[Case] \tylab{Return}. In a normal form with respect to $E$.

  \item[Case] \tylab{Do}. In a normal form with respect to $E$.

  \item[Case]
  \begin{mathpar}
  \inferrule*[Lab=\tylab{Seq}]
    {\typc{\Delta;\Gamma_1}{M : A}{E} \\
      \typc{\Delta;\Gamma_2, x : A}{N : B}{E} \\\\
      E = \{R\} \\
      {\Delta \vdash \Gamma_2 : Y} \\
      {\Delta \vdash R : Y} \\
    }
    {\typc{\Delta;\Gamma_1 + \Gamma_2}{\Lety^{Y} x \revto M \In N : B}{E}}
  \end{mathpar}
  By a case analysis on $M$.
  \begin{description}
    \item[Subcase] $M = (\Ret N)^E$. Reduced by \semlab{Seq}.
    \item[Subcase] Otherwise. By the IH, if $M \reducesto N$, then the
    original term is reduced by \semlab{Lift}. Otherwise, $M$ is in a
    normal form with respect to $E$, which implies the original term
    is also in a normal form with respect to $E$.
  \end{description}

  \item[Case]
  \begin{mathpar}
  \inferrule*[Lab=\tylab{Handle}]
  {
    \typ{\Delta;\Gamma_1}{H : C \hto D} \\
    \typ{\Delta;\Gamma_2}{M : C} \\
    C = A\eff E' \\ D = B\eff E
  }
  {\Delta;\Gamma_1 + \Gamma_2 \vdash \Handle M \With H : D}
  \end{mathpar}
  By a case analysis on $M$.
  \begin{description}
    \item[Subcase] $M = (\Ret N)^{E'}$. Reduced by \semlab{Ret}.
    \item[Subcase] $M = \EC[(\Do\ell\,V)^{E''}]$ with
    $\ell\notin\BL{\EC}$ and $(\ell\,p\,r\mapsto N)\in H$. The
    original term is reduced by \semlab{Op}.
    \item[Subcase] Otherwise. By the IH, if $M \reducesto N$, then the
    original term is reduced by \semlab{Lift}. Otherwise, $M$ is in a
    normal form with respect to $E'$. By \Cref{def:normal-form}, $M =
    \EC[(\Do\ell\,V)^{E''}]$ for $\ell\in E'$ and
    $\ell\notin\BL{\EC}$. By the last subcase, $\ell$ is also not handled
    by $H$. Thus, the original term is also in a normal form with
    respect to $E$.
  \end{description}
\end{description}
\end{proof}

\subsection{Subject Reduction}
\label{app:proof-subject-reduction}

\begin{lemma}[Substitution]~
  \label{lemma:substitution}
  \begin{enumerate}[label=\arabic*.]
    \item Preservation of kinds under type substitution:
    if $\Delta,\alpha : K'\vdash T:K$
    and $\Delta \vdash T' : K'$, then $\Delta \vdash T[T' / \alpha] : K$.
    \item Preservation of types under type substitution:
    if $\Delta \vdash T : K$,
    then $\Delta,\alpha : K;\Gamma\vdash M:C$
    implies $\Delta;\Gamma[T/ \alpha]\vdash M[T/ \alpha] : C[T/ \alpha]$,
    and $\Delta,\alpha : K;\Gamma\vdash V:A$
    implies $\Delta;\Gamma[T/ \alpha]\vdash V[T/ \alpha] : A[T/ \alpha]$,
    and $\Delta,\alpha : K;\Gamma\vdash H:C\hto D$
    implies $\Delta;\Gamma[T/ \alpha]\vdash H[T/\alpha]:(C\hto D)[T/\alpha]$.
    \item Preservation of types under value substitution:
    if $\Delta\vdash\Gamma_1:Y$, $\Delta;\Gamma_1\vdash V:A$ and
    $\Delta\vdash A:Y$, then $\Delta;\Gamma_2,x:A\vdash M:C$ implies
    $\Delta;\Gamma_1 + \Gamma_2\vdash M[V/x]:C$,
    and $\Delta;\Gamma_2,x:A\vdash W:B$ implies $\Delta;\Gamma_1 +
    \Gamma_2\vdash W[V/x]:B$,
    and $\Delta;\Gamma_2,x:A\vdash H:C\hto D$ implies
    $\Delta;\Gamma_1+\Gamma_2\vdash H[V/x]:C\hto D$.
  \end{enumerate}
\end{lemma}

\begin{proof} We apply various structural lemmas like weakening,
permutation of contexts, and properties of context splitting in the
following proofs.

\noindent \emph{1. Preservation of kinds under type substitution.}
Straightforward induction on the kinding derivations.

\noindent \emph{2. Preservation of types under type substitution.}
By \Cref{lemma:substitution}.1 and straightforward mutual induction on
the typing derivations.

\noindent \emph{3. Preservation of types under value substitution.}
By mutual induction on the typing derivations.
\end{proof}

\preservation*
\begin{proof}

By induction on the typing derivation $\typ{\Delta;\Gamma}{M:C}$.
\begin{description}
  \item[Case]
  \begin{mathpar}
    \inferrule*[Lab=\tylab{App}]
    {\refa{\typ{\Delta;\Gamma_1}{V : A \to^Y C}} \\
     \refb{\typ{\Delta;\Gamma_2}{W : A}}
    }
    {\typ{\Delta;\Gamma_1 + \Gamma_2}{V\,W : C}}
  \end{mathpar}
  The reduction can only be derived using \semlab{App}, which implies
  $V = \lambda^Y x^A . N$ and $(\lambda^Y x^A . N)\ W \reducesto
  N[W/x]$.
  Inversion on \refa{} gives \refc{$\typ{\Delta;\Gamma_1,x:A}{N:C}$}.
  Case analysis on the linearity of $A$:
  \begin{description}
    \item[Subcase] \refd{$\Delta\vdash A:\Unl$}. Applying
    \Cref{thm:unl-is-unl}.1 to \refb{} gives
    \refe{$\Delta\vdash\Gamma_2 : \Unl$}.
    Applying \Cref{lemma:substitution}.3 to \refb{}, \refc{}, \refd{}
    and \refe{} gives $\typ{\Delta;\Gamma_1 + \Gamma_2}{N[W/x] : C}$.

    \item[Subcase] \refd{$\Delta\vdash A:\Lin$}. We always have
    \refe{$\Delta\vdash\Gamma_2:\Lin$}.
    Applying \Cref{lemma:substitution}.3 to \refb{}, \refc{}, \refd{}
    and \refe{} gives $\typ{\Delta;\Gamma_1 + \Gamma_2}{N[W/x] : C}$.
  \end{description}

  \item[Case]
  \begin{mathpar}
    \inferrule*[Lab=\tylab{TApp}]
    {\refa{\typ{\Delta;\Gamma}{V : \forall^Y \alpha^K . \, C}} \\
     \refb{\Delta \vdash T : K}
    }
    {\typ{\Delta;\Gamma}{V\,T : C[T/\alpha]}}
  \end{mathpar}
  The reduction can only be derived using \semlab{TApp}, which implies
  $V = \Lambda^Y \alpha^K . N$ and\\ $(\Lambda^Y \alpha^K . N)\; T
  \reducesto N[T/\alpha]$.
  Inversion on \refa{} gives \refc{$\typ{\Delta,\alpha:K;
  \Gamma}{N:C}$}. By $\alpha\notin \ftv{\Gamma}$, applying
  \Cref{lemma:substitution}.2 to \refb{} and \refc{} gives
  $\typ{\Delta;\Gamma}{N[T/\alpha] : C[T/\alpha]}$.

  \item[Case] \tylab{Return}. No reduction. $M$ is in a normal form.

  \item[Case] \tylab{Do}. No reduction. $M$ is in a normal form.

  \item[Case]
  \begin{mathpar}
  \inferrule*[Lab=\tylab{Seq}]
    { \refa{\typc{\Delta;\Gamma_1}{M : A}{\{R\}}} \\
      \refb{\typc{\Delta;\Gamma_2, x : A}{N : B}{\{R\}}} \\\\
      {\Delta \vdash \Gamma_2 : Y} \\
      {\Delta \vdash R : Y} \\
    }
    {\typc{\Delta;\Gamma_1 + \Gamma_2}{\Lety^{Y} x \revto M \In N : B}{{\{R\}}}}
  \end{mathpar}
  By a case analysis on the next rule used by reduction:
  \begin{description}
    \item[Subcase] \semlab{Lift}. Suppose $M \reducesto M'$. The IH on
    \refa{} gives $\typc{\Delta;\Gamma_1}{M' : A}{\{R\}}$. Then, by
    \tylab{Seq} we have ${\typc{\Delta;\Gamma_1 + \Gamma_2}{\Lety^Y x
    \revto M' \In N : B}{\{R_2\}}}$.
    \item[Subcase] \semlab{Seq}. $M = (\Ret V)^{\{R\}}$. Inversion on
    \refa{} gives \refc{$\typ{\Delta;\Gamma_1}{V:A}$}.
    With \refb{} and \refc{}, our goal follows from a case analysis on
    the linearity of $A$ similar to the \tylab{App} case.
  \end{description}

  \item[Case]
  \begin{mathpar}
\inferrule*[Lab=\tylab{Handle}]
{
  \refa{\typ{\Delta;\Gamma_1}{M : C}} \\
  \refb{\typ{\Delta;\Gamma_2}{H : C \hto D}} \\
}
{\Delta;\Gamma_1 + \Gamma_2 \vdash \Handle M \With H : D}
  \end{mathpar}
  By a case analysis on the next rule used by reduction:
  \begin{description}
    \item[Subcase] \semlab{Lift}. Suppose $M \reducesto M'$. The IH on
    \refa{} gives $\typ{\Delta;\Gamma_1}{M' : C}$. Then, by
    \tylab{Handle} we have ${\typ{\Delta;\Gamma_1 + \Gamma_2}{\Handle
    M' \With H:D}}$.
    \item[Subcase] \semlab{Ret}. $M = (\Ret V)^E$ and $(\Ret x \mapsto
    N) \in H$. Suppose $C = A\eff E$.
    Inversion on \refa{} gives \refc{$\typ{\Delta;\Gamma_1}{V:A}$}.
    Inversion on \refb{} gives \refd{$\typ{\Delta;\Gamma_2,
    x:A}{N:D}$}.
    With \refc{} and \refd{}, our goal follows from a case analysis on
    the linearity of $A$ similar to the \tylab{App} case.

    \item[Subcase] \semlab{Op}. $M = \EC[(\Do \ell \; V)^E]$, $\ell
    \notin \BL{\EC}$ and $(\ell \ p \ r \mapsto N) \in H$.
    Suppose $(\ell : A \to^Y B) \in E$ and $W = \lambda^Y y^B.\Handle
    \EC[(\Ret y)^E] \With H$.
    The reduction is $\Handle M \With H \reducesto N[V/p,W/r]$.
    Inversion on \refb{} gives
    \refc{$\typ{\Delta;\Gamma_2,p:A,r:B\to^Y D}{N:D}$}.
    By a straightforward induction on \refa{} similar to the proof of
    \Cref{thm:unl-is-unl}.2, it is easy to show that there exists
    $\Delta\vdash\Gamma_1 = \Gamma_{11}+\Gamma_{12}$ such that
    \refd{$\typ{\Delta;\Gamma_{11}}{V : A}$} and
    \refe{$\typ{\Delta;\Gamma_{12}, y:B}{\EC[(\Ret y)^E]} : C$}.
    With \refc{} and \refd{}, by a case analysis on the linearity of
    $A$ similar to the \tylab{App} case, we have
    \reff{$\typ{\Delta;\Gamma_{11} + (\Gamma_2, r:B\to^Y D)}{N[V/p] :
    D}$}.
    Then by another case analysis on $Y$:
    \begin{description}
      \item[subcase] $Y=\Unl$. By \Cref{thm:unl-is-unl}.2 we have
      $\Delta\vdash\Gamma_{12}:\Unl$.
      Applying \tylab{Handle} and \tylab{Abs} to \refe{}, we have
      \refg{$\typ{\Delta;\Gamma_{12}+\Gamma_2}{W : B\to^Y D}$}.
      Applying \Cref{thm:unl-is-unl}.3 to \refb{} we have
      $\Delta\vdash\Gamma_2:\Unl$.
      Finally, applying \Cref{lemma:substitution}.3 to \reff{} and
      \refg{}, we have
      $\typ{\Delta;\Gamma_{11}+\Gamma_{12}+\Gamma_2}{N[V/p,W/r]:D}$.

      \item[subcase] $Y=\Lin$.
      Applying \tylab{Handle} and \tylab{Abs} to \refe{}, we have
      \refg{$\typ{\Delta;\Gamma_{12}+\Gamma_2}{W : B\to^Y D}$}.
      We always have $\Delta\vdash\Gamma_{12}+\Gamma_2 : \Lin$.
      Finally, applying \Cref{lemma:substitution}.3 to \reff{} and
      \refg{}, we have
      $\typ{\Delta;\Gamma_{11}+\Gamma_{12}+\Gamma_2}{N[V/p,W/r]:D}$.

    \end{description}
  \end{description}

\end{description}
\end{proof}

\subsection{Linearity Safety of Evaluation}
\label{app:proof-reduction-safety}

\begin{lemma}[Linear variables appear exactly once]
  \label{lemma:linear-var-exactly-once}
  If $\typ{\Delta;\Gamma,x:A}{V:B}$ and $\Delta\not\vdash A:\Unl$,
  then $x$ appears exactly once in $V$.
  If $\typ{\Delta;\Gamma,x:A}{M:C}$ and $\Delta\not\vdash A:\Unl$,
  then $x$ appears exactly once in $M$.
\end{lemma}
\begin{proof}
  By the definition of the context splitting relation and
  straightforward induction on typing derivations.
\end{proof}

\begin{lemma}[Preservation of linear safety under substitution]
  \label{lemma:linear-safe-substitution}
  Given closed and linear safe $V$ and $M$, if $\typ{}{V:A}$ and
  $\typ{\cdot;x:A}{M:C}$, then $M[V'/x]$ is linear safe where $(V',\_) =
  \taglin(V)$.
\end{lemma}
\begin{proof}
  Case analysis on the linearity of $A$.
  \begin{description}
    \item[Case] $\vdash A:\Unl$. We have $V' = V$.
    By the linear safety of $V$, we have $\TL{V}=\emptyset$.
    The linear safety of $M[V'/x]$ follows from the linear safety of
    $M$.

    \item[Case] $\not\vdash A:\Unl$. By \Cref{thm:unl-is-unl}, $x$
    does not appear in unlimited values, continuations and handlers of
    $M$.
    Thus, $V'$ does not appear in unlimited values, continuations and
    handlers of $M[V'/x]$.
    The linear safety of $M[V'/x]$ then directly follows from the
    linear safety of $M$ and $V$.
  \end{description}
\end{proof}

\reductionSafety*

\begin{proof}

We proceed by induction on the linearity-aware reduction rules defined
in \Cref{fig:feffpop-linear-small-step}.
To avoid name conflicts, we consider $\hat{M} \lreducesto{\mS}{\mT}
\hat{N}$.

\begin{description}
  \item[Case]
  \begin{bigreductions}
    \lsemlab{App}   & (\lambda^Y x^A.M)\,V &\lreducesto{\mS}{\emptyset}& M[V'/x],
    \text{ where } (V', \mS) = \taglin(V)
  \end{bigreductions}
  The linear safety of $\hat{M}$ gives the linear safety of $M$ and
  $V$.
  The linear safety of $\hat{N}$ follows from
  \Cref{lemma:linear-safe-substitution}.
  By inversion on $\hat{M}$, $V$ has type $A$. Case analysis on the
  linearity of $A$:
  \begin{description}
    \item[Subcase] $\vdash A :\Unl$. We have $\islin(V) = \False$ and
    $\taglin(V) = \{V,\emptyset\}$.
    By the fact that $V$ is closed and linear safe, we have $\TL{V} =
    \emptyset$.
    Our goal follows from $\TL{\hat{M}} \umul \emptyset = \TL{M}
    = \TL{\hat{N}} \umul \emptyset$.
    \item[Subcase] $\not\vdash A:\Unl$.  We have $\islin(V) = \True$.
    By \Cref{lemma:linear-var-exactly-once}, $x$ appears in $M$
    exactly once.
    If $V = \lintag{W}$ for some $W$, then we have $\TL{\hat{M}}
    \umul \emptyset = \TL{M} \umul \TL{V} = \TL{M[V/x]} =
    \TL{\hat{N}} \umul \emptyset$.
    Otherwise, we have $\TL{\hat{M}} \umul \{\lintag{V}\} = \TL{M}
    \umul \TL{V}\umul \{\lintag{V}\} = \TL{M} \umul \TL{\lintag{V}}
    = \TL{M[\lintag{V} / x]} = \TL{\hat{N}} \umul \emptyset$.
  \end{description}

  \item[Case]
  \begin{bigreductions}
    \lsemlab{TApp} & (\Lambda^Y \alpha^K.M)\,T &\lreducesto{\emptyset}{\emptyset}& M[T/\alpha]
  \end{bigreductions}
  The linear safety of $\hat{N}$ directly follows from the linear
  safety of $\hat{M}$.
  We have $\TL{\hat{M}}\umul\emptyset = \TL{M} = \TL{\hat{N}}\umul\emptyset$.

  \item[Case]
  \begin{bigreductions}
  \lsemlab{Seq} &
    \Lety^Y x \revto \Ret V \In N &\lreducesto{\mS}{\emptyset}& N[V'/x],
    \text{ where } (V', \mS) = \taglin(V) \\
  \end{bigreductions}
  The linear safety of $\hat{M}$ gives the linear safety of $N$ and
  $V$.
  The linear safety of $\hat{N}$ follows from
  \Cref{lemma:linear-safe-substitution}.
  Suppose $\typ{}{V:A}$. Our goal follows from a case analysis on the
  linearity of $A$ similar to the \lsemlab{App} case.

  \item[Case]
  \begin{bigreductions}
  \lsemlab{Ret} &
    \Handle (\Ret V)^E \With H &\lreducesto{\mS}{\emptyset}& N[V'/x],\\
    \multicolumn{4}{@{}r@{}}{
        \text{where } (\Ret x \mapsto N) \in H, (V', \mS) = \taglin(V)
    }
  \end{bigreductions}
  The linear safety of $\hat{M}$ gives the linear safety of $V$, $H$ and
  $N$.
  The linear safety of $\hat{N}$ follows from
  \Cref{lemma:linear-safe-substitution}.
  Suppose $\typ{}{V:A}$. Our goal follows from a case analysis on the
  linearity of $A$ similar to the \lsemlab{App} case.

  \item[Case]
  \begin{bigreductions}
  \lsemlab{Op} &
    \Handle \EC[(\Do \ell \; V)^E] \With H
      &\lreducesto{\mS}{\emptyset}& N[V'/p, W'/r],\\
    \multicolumn{4}{@{}r@{}}{
          \text{where } \ell \notin \BL{\EC},\
          (\ell \ p \ r \mapsto N) \in H,\
          (\ell:A \sto^{Y} B) \in E,
    }\\
    \multicolumn{4}{@{}r@{}}{
          W = \lambda^{Y} y^{B}.\Handle \EC[(\Ret y)^E] \With H,
    } \\
    \multicolumn{4}{@{}r@{}}{
          (V',\mS_1) = \taglin(V),
          (W',\mS_2) = \taglin(W), \mS = \mS_1 \umul \mS_2
    }
  \end{bigreductions}
  The linear safety of $\hat{M}$ gives the linear safety of $V$, $H$,
  $N$ and $\EC$.
  We need to show the linear safety of $W$.
  If $Y = \Lin$, the linear safety of $W$ directly follows from the
  linear safety of $\EC$ and $H$.
  If $Y = \Unl$, by the linear safety of $\EC[(\Do\ell\;V)^E]$ we have
  $\TL{\EC} = \emptyset$.
  By the linear safety of $H$ we have $\TL{H} = \emptyset$.
  Thus, $\TL{W} = \emptyset$, which gives us the linear safety of $W$.
  The linear safety of $\hat{N}$ follows from
  \Cref{lemma:linear-safe-substitution}.
  Then, we prove the equation.
  By inversion on $(\Do\ell\;V)^E$, we have $\typ{}{V:A}$.
  Suppose $\typ{}{W:B\to^Y C}$.
  By the linear safety of $H$, we have $\TL{H} = \TL{N} = \emptyset$.
  By a case analysis on the linearity of $A$.
  \begin{description}
    \item[Subcase] $\vdash A :\Unl$. We have $\islin(V) = \False$ and
    $\taglin(V) = \{V,\emptyset\}$.
    By the fact that $V$ is closed and linear safe, we have $\TL{V} =
    \emptyset$.
    By a case analysis on the linearity of $B \to^Y C$.
    \begin{description}
      \item[subcase] $\vdash B\to^Y C:\Unl$. We have $\islin(W)=\False$ and
      $\taglin(W) = \{W,\emptyset\}$.
      By the fact that $W$ is closed and linear safe, we have $\TL{W}
      = \emptyset$.
      Our goal follows from $\TL{\hat{M}} \umul \emptyset = \emptyset
      = \TL{\hat{N}} \umul \emptyset$.

      \item[subcase] $\vdash B\to^Y C:\Lin$. We have $\islin(W)=\True$
      and $\taglin(W) = \{\lintag{W}, \{\lintag{W}\}\}$.
      By \Cref{lemma:linear-var-exactly-once}, $r$ appears in $N$
      exactly once.
      We have $\TL{\hat{M}}\umul\{\lintag{W}\} =
      \TL{\EC}\umul\{\lintag{W}\} = \TL{\lintag{W}} = \TL{\hat{N}}$.
    \end{description}
    \item[Subcase] $\not\vdash A:\Unl$.  We have $\islin(V) = \True$.
    By \Cref{lemma:linear-var-exactly-once}, $p$ appears in $N$
    exactly once.
    If $V = \lintag{V_1}$ for some $V_1$, we have $\lintag{V} =
    (V,\emptyset)$.
    By a case analysis on the linearity of $B\to^Y C$.
    \begin{description}
      \item[subcase] $\vdash B\to^Y C:\Unl$. We have $\islin(W)=\False$ and
      $\taglin(W) = \{W,\emptyset\}$.
      By the fact that $W$ is closed and linear safe, we have $\TL{W}
      = \emptyset$.
      Our goal follows from $\TL{\hat{M}} \umul \emptyset = \TL{V}
      = \TL{\hat{N}} \umul \emptyset$.

      \item[subcase] $\vdash B\to^Y C:\Lin$. We have $\islin(W)=\True$
      and $\taglin(W) = \{\lintag{W}, \{\lintag{W}\}\}$.
      By \Cref{lemma:linear-var-exactly-once}, $r$ appears in $N$
      exactly once.
      We have $\TL{\hat{M}}\umul\{\lintag{W}\} =
      \TL{V}\umul\TL{\EC}\umul\{\lintag{W}\} =
      \TL{V}\umul\TL{\lintag{W}} = \TL{\hat{N}}$.
    \end{description}
    Otherwise, we have $\lintag{V} = (\lintag{V}, \{\lintag{V}\})$.
    By a case analysis on the linearity of $B\to^Y C$.
    \begin{description}
      \item[subcase] $\vdash B\to^Y C:\Unl$. We have $\islin(W)=\False$ and
      $\taglin(W) = \{W,\emptyset\}$.
      By the fact that $W$ is closed and linear safe, we have $\TL{W}
      = \emptyset$.
      Our goal follows from $\TL{\hat{M}} \umul \emptyset =
      \TL{V}\umul\{\lintag{V}\} = \TL{\lintag{V}} = \TL{\hat{N}}
      \umul \emptyset$.

      \item[subcase] $\vdash B\to^Y C:\Lin$. We have $\islin(W)=\True$
      and $\taglin(W) = \{\lintag{W}, \{\lintag{W}\}\}$.
      By \Cref{lemma:linear-var-exactly-once}, $r$ appears in $N$
      exactly once.
      We have $\TL{\hat{M}}\umul\{\lintag{W},\lintag{V}\} =
      \TL{V}\umul\TL{\EC}\umul\{\lintag{W},\lintag{V}\} =
      \TL{\lintag{V}}\umul\TL{\lintag{W}} = \TL{\hat{N}}$.
    \end{description}
  \end{description}

  \item[Case]
  \begin{bigreductions}
  \lsemlab{Remove} & \FC[\lintag{V}] &\lreducesto{\emptyset}{\{\lintag{V}\}} &\FC[V]
  \end{bigreductions}
  The linear safety of $\hat{N}$ directly follows from the linear safety of $\hat{M}$.
  We have $\TL{\hat{M}}\umul\emptyset = \TL{\FC}\umul\TL{\lintag{V}}
  = \TL{\FC}\umul\TL{V}\umul\{\lintag{V}\} =
  \TL{\hat{N}}\umul\{\lintag{V}\}$.

  \item[Case]
  \begin{bigreductions}
  \lsemlab{Lift} &
    \EC[M] &\lreducesto{\mS}{\mT}& \EC[N],  \hfill\text{ if } M \lreducesto{\mS}{\mT} N
  \end{bigreductions}
  The linear safety of $\hat{M}$ gives the linear safety of $\EC$ and $M$.
  By IH, we have the linear safety of $N$.
  The linear safety of $\hat{N}$ follows from the linear safety of
  $\EC$ and $N$.
  By IH, we have $\TL{M}\umul\mS = \TL{N}\umul\mT$.
  Our goal follows from $\TL{\hat{M}}\umul\mS = \TL{\EC} \umul
  \TL{M}\umul\mS = \TL{\EC} \umul \TL{N}\umul\mT =
  \TL{\hat{N}}\umul\mT$.

\end{description}

\end{proof}

\section{Full Specification of \CalcS}

In this section, we give the full syntax, typing rules, type
inference, and constraint solving algorithm of \CalcS in
\Cref{sec:qeffpopsub}.

\subsection{Full Syntax}
\label{app:qfp-syntax}

The full syntax of $\CalcS$ is given in \Cref{fig:qeffpopsub-syntax}.
Note that we introduce the syntactic category of concrete rows to
simplify the presentation of the constraint solving algorithm.

\begin{figure}[htb]
  \begin{syntax}
  \slab{Value types}    &A,B  &::= & \alpha \mid A \to^{Y} C \\
  \slab{Computation types}
                        &C,D  &::= & A \eff E \\
  \slab{Handler types}  &F    &::= & C \hto D \\
  \slab{Effect types}   &E    &::= & \{R\}\\
  \slab{Concrete row types} &\Belong{\cR}{\meta{CRow}} &::= & \cdot \mid \ell \lsig A\sto^Y B \lsep \cR \\
  \slab{Row types}          &\Belong{R}{\meta{Row}}  &::= & \mu \mid \cR \mid \cR \lsep R \\
  \slab{Linearity types}&Y    &::= & \hl{\phi} \mid \Unl \mid \Lin \\
  \slab{Types}          &\tau &::= & A \mid R \mid Y \\
  \slab{Predicates}         &\Belong{\pi}{\meta{Pred}} &::= & \hlbox{\tau_1 \linle \tau_2
    \mid R_1 \rle R_2 \mid R \rlack \mL} \\
  \slab{Qualified types} &\rho &::=& A \mid {\pi \Rightarrow \rho} \\
  \slab{Type schemes}   &\Belong{\sigma}{\meta{TySch}} &::=& \rho \mid \forall \alpha . \sigma \\
  \slab{Label sets}     &\mathcal{L} &::=& \emptyset \mid \{\ell\} \uplus \mathcal{L}\\
  \slab{Type contexts}  &\Belong{\Gamma}{\meta{Env}}   &::=& \cdot \mid \Gamma, x:\sigma \\
  \slab{Predicate sets}     &\Belong{P}{\meta{PSet}} &::= & \cdot \mid P, \pi \\

  \slab{Values}        &V,W  &::=  & x \mid  \lambda x . M \\
  \slab{Computations}  &M,N  &::= & V\,W \mid \Ret V
                                  \mid \Do \ell\; V
                                  \mid \hlbox{\Let x = V \In M} \\
                       &     &\mid& \Let x \revto M \In N \mid \Handle M \With H\\
  \slab{Handlers}      &H    &::= & \{ \Ret x \mapsto M \}
                              \mid  \{ \ell \; p \; r \mapsto M \} \uplus H
  \end{syntax}
  \caption{The Syntax of \CalcS}
  \label{fig:qeffpopsub-syntax}
\end{figure}

\subsection{Full Typing Rules}
\label{app:qfp-typing}

The full syntax-directed typing rules for \CalcS is given in
\Cref{fig:qe-typing}.
Note that in the qualified effect system of \CalcS, we only have
subtyping between row types and use them in \qtylab{Do}, \qtylab{Seq},
\qtylab{Handle}, and \qtylab{Handler}.
This is different from other type systems with general subtyping,
where the subtyping relation is used everywhere.
For example, in the \qtylab{App} rule, we require the argument type to
be equal to the parameter type of the function, instead of requiring a
subtyping relation.
Having a full subtyping relation between any types does not help
improve the accuracy of tracking control-flow linearity; subtyping
between effect rows is enough.

\begin{figure}[htb]

  \raggedright
  \boxed{\typ{P\mid\Gamma}{V : A}}
  \boxed{\typ{P\mid\Gamma}{M : C}}
  \boxed{\typ{P\mid\Gamma}{H : C \hto D}}
  \hfill

  \begin{mathpar}

  \inferrule*[Lab=\qtylab{Var}]
  { P \vdash \Gamma\un \\\\
    (P \qto A) \sqgen \sigma}
  {\typ{P\mid\Gamma, x : \sigma}{x : A}}

  \inferrule*[Lab=\qtylab{Abs}]
  { \typ{P\mid\Gamma, x : A}{M : C}\\\\
    P \vdash \Gamma \linle Y }
  {\typ{P\mid\Gamma}{\lambda x . M : A \to^Y C}}

  \inferrule*[Lab=\qtylab{App}]
  { \typ{P\mid\Gamma_1, \Gamma}{V : A \to^Y C} \\\\
    \typ{P\mid\Gamma_2, \Gamma}{W : A} \\
    P \vdash \Gamma\un
  }
  {\typ{P\mid\Gamma_1, \Gamma_2, \Gamma}{V\,W : C}}

  \inferrule*[Lab=\qtylab{Let}]
  { \typ{Q\mid\Gamma_1, \Gamma}{V : A} \\
    \sigma = \gen((\Gamma_1, \Gamma), Q\qto A) \\\\
    \typ{P\mid\Gamma_2, \Gamma, x : \sigma}{M : C} \\
    P \vdash \Gamma\un
  }
  {\typ{P\mid\Gamma_1, \Gamma_2, \Gamma}{\Let x = V \In M : C}}
  
    \inferrule*[Lab=\qtylab{Return}]
      {\typ{P\mid\Gamma}{V : A}}
      {\typc{P\mid\Gamma}{\Ret V : A}{\{R\}}}

    \inferrule*[Lab=\qtylab{Do}]
      {
        \typ{P\mid\Gamma}{V : A_\ell} \\\\
        P \vdash \{\ell\lsig A_\ell\sto^Y B_\ell\} \rle R
      }
      {\typc{P\mid\Gamma}{\Do \ell \; V : B_\ell}{\{R\}}}

  \inferrule*[Lab=\qtylab{Seq}]
  { \typc{P\mid\Gamma_1, \Gamma}{M : A}{\{R_1\}} \\
    \typc{P\mid\Gamma_2, \Gamma, x : A}{N : B}{\{R_2\}} \\\\
    P \vdash R_1 \rle R \\ P \vdash R_2 \rle R \\
    P \vdash \Gamma_2 \linle R_1 \\
    P \vdash \Gamma \un
  }
  {\typc{P\mid\Gamma_1,\Gamma_2,\Gamma}{\Let x \revto M \In N : B}{\{R\}}}

  \inferrule*[Lab=\qtylab{Handle}]
  {
    \typ{P\mid\Gamma_1,\Gamma}{H : A\eff\{R_1\} \hto D} \\\\
    \typc{P\mid\Gamma_2, \Gamma}{M : A}{\{R\}} \\\\
    P \vdash \Gamma \un \\ P \vdash R \rle R_1
  }
  {P\mid\Gamma_1, \Gamma_2,\Gamma\vdash \Handle M \With H : D}

  \inferrule*[Lab=\qtylab{Handler}]
  {
    H = \{\Ret x \mapsto M\} \uplus \{ \ell_i\;p_i\;r_i \mapsto N_i \}_i\\\\
    C = A \eff \{(\ell_i : A_i \sto^{Y_i} B_i)_i; R_1\} \\
    D = B \eff \{R_2\} \\\\
    \typ{P\mid\Gamma, x : A}{M : D}\\\\
    [\typ{P\mid\Gamma, p_i : A_i, r_i : B_i \to^{Y_i} D}{N_i : D}]_i \\\\
    P \vdash \Gamma \un \\
    P \vdash R_1 \rle R_2 \\
    P \vdash R_1 \rlack \{\ell_i\}_i
  }
  {\typ{P\mid\Gamma}{H : C \hto D}}

  \end{mathpar}

  \caption{Syntax-directed Typing Rules for \CalcQ}
  \label{fig:qe-typing}
\end{figure}

\subsection{Type Inference Algorithm}
\label{app:qfp-inference}

The full type inference of \CalcS is given in
\Cref{fig:qeffpopsub-type-inference}.
It uses the unification relations $\unify{\tau}{\tau'}{\theta}$ which
states that $\theta$ is the principal unifier of types $\tau$ and
$\tau'$, and $\unify{C}{C'}{\theta}$ which states that $\theta$ is the
principal unifier for computation types $C$ and $C'$.
The unification relations are directly defined by the unification
function.

\begin{mathpar}
  \inferrule*[Lab=\ulab{Type}]
  {\munify(\tau\sim \tau') = \theta}
  {\unify{\tau}{\tau'}{\theta}}

  \inferrule*[Lab=\ulab{Comp}]
  {\munify(C\sim C') = \theta}
  {\unify{C}{C'}{\theta}}
\end{mathpar}

\Cref{fig:unification} gives unification function $\munify(U)$ which
takes a set of unification predicates and returns the principal
unifiers of them.
It is relatively standard~\citep{MartelliM82}.
The arrow $\rightharpoonup$ indicates a meta function that might fail.
Following \citet{Leijen08} we explicitly indicate the successful
return of a result by $\mreturn$.
The auxiliary functions $\unifyrow$ and $\unifylin$ are given and
explained in \label{sec:constraint-solver}.
The unification predicates and predicate sets are defined as follows.
\begin{syntax}
  \slab{Unification predicates} &\Belong{u}{\meta{UPred}} &::= & \tau\sim\tau' \mid C\sim C' \\
  \slab{Unification sets} &\Belong{U}{\meta{USet}} &::= & U, u \\
\end{syntax}

Note that it is possible to postpone the solving of unification
constraints to the constraint solving algorithm.
We opt for this mixed style presentation for \CalcS in order to keep
close to the original presentation of qualified types \citep{jones94},
and to keep the constraint set cleaner.

\begin{figure}[htb]
  \raggedright
  \boxed{\infer{\Gamma}{V}{A}{\theta}{P}{\Sigma}}
  \boxed{\infer{\Gamma}{M}{C}{\theta}{P}{\Sigma}}
  \boxed{\infer{\Gamma}{H}{C \hto D}{\theta}{P}{\Sigma}}

  \begin{mathpar}
  \inferrule*[Lab=\tilab{Var}]
  {
    (x : \forall \ol{\alpha} . P \qto A) \in \Gamma \\\\
    \ol\beta \fresh \\
    \theta = [\ol{\beta}/ \ol{\alpha}]
  }
  {\infer{\Gamma}{x}{\theta A}{\theta}{\theta P}{\{x\}}}

  \inferrule*[Lab=\tilab{Let}]
  {
    \infer{\Gamma}{V}{A}{\theta_1}{P_1}{\Sigma_1} \\
    \sigma = \gen(\theta_1 \Gamma, P_1 \qto A) \\\\
    \infer{\theta_1\Gamma, x:\sigma}{M}{C}{\theta_2}{P_2}{\Sigma_2} \\\\
    Q = \makeunl{\theta_2\theta_1\Gamma |_{\Sigma_1 \cap \Sigma_2}} \cup
        \makeunl{\theta_2(x:\sigma) |_{\setcomplement{\Sigma_2}}}
  }
  {\infer{\Gamma}{\Let x = V \In M}{C}
    {\theta_2\theta_1}{P_2\cup Q}{\Sigma_1 \cup (\Sigma_2 \backslash x)}}

  \inferrule*[Lab=\tilab{Abs}]
  {
    \alpha,\phi \fresh\\
    \infer{\Gamma, x:\alpha}{M}{C}{\theta}{P}{\Sigma} \\\\
    Q = \makeleq{\theta\Gamma |_{\Sigma}, \phi} \cup
        \makeunl{\theta(x:\alpha) |_{\setcomplement{\Sigma}}}
  }
  {\infer{\Gamma}{\lambda x . M}{\theta\alpha \to^\phi C}
    {\theta}{P \cup Q}{\Sigma \backslash x}}

  \inferrule*[Lab=\tilab{App}]
  {
    \infer{\Gamma}{V}{A}{\theta_1}{P_1}{\Sigma_1} \\
    \infer{\theta_1\Gamma}{W}{B}{\theta_2}{P_2}{\Sigma_2} \\\\
    \alpha, \mu, \phi \fresh \\
    \unify{\theta_2 A}{(B \to^\phi \alpha\eff\mu)}{\theta_3} \\\\
    P = \theta_3 (\theta_2 P_1 \cup P_2) \\
    Q = \makeunl{\theta_3\theta_2\theta_1\Gamma |_{\Sigma_1 \cap \Sigma_2}}
  }
  {\infer{\Gamma}{V\; W}{\theta_3(\alpha\eff\mu)}
    {\theta_3\theta_2\theta_1}{P\cup Q}{\Sigma_1 \cup \Sigma_2}}

  \inferrule*[Lab=\tilab{Seq}]
  {
    \infer{\Gamma}{M}{A\eff\{R_1\}}{\theta_1}{P_1}{\Sigma_1} \\
    \infer{\theta_1\Gamma,x:A}{N}{B\eff\{R_2\}}{\theta_2}{P_2}{\Sigma_2} \\
    \mu \fresh \\\\
    Q = \makeunl{\theta_2\theta_1\Gamma |_{\Sigma_1 \cap \Sigma_2}} \cup
        \makeunl{\theta_2(x:A) |_{\setcomplement{\Sigma_2}}} \cup
        \makeleq{\theta_2\theta_1\Gamma |_{\Sigma_2}, \theta_2 R_1} \cup
        \makesub{\theta_2 R_1, \mu} \cup
        \makesub{R_2, \mu} \\
  }
  {\infer{\Gamma}{\Let x \revto M \In N}{B\eff \mu}
    {\theta_2\theta_1}{\theta_2 P_1 \cup P_2 \cup Q}{\Sigma_1 \cup (\Sigma_2 \backslash x)}}

  \inferrule*[Lab=\tilab{Return}]
  {
    \infer{\Gamma}{V}{A}{\theta}{P}{\Sigma} \\
    \mu \fresh
  }
  {\infer{\Gamma}{\Ret V}{A \eff \{\mu\}}{\theta}{P}{\Sigma}}

  \inferrule*[Lab=\tilab{Do}]
  {
    \infer{\Gamma}{V}{A}{\theta_1}{P}{\Sigma} \\
    \unify{A}{A_\ell}{\theta_2} \\\\
    \mu, \phi \fresh \\
    Q = \makesub{(\ell: A_\ell \sto^\phi B_\ell), \mu} \\
  }
  {\infer{\Gamma}{\Do \ell\; V}
    {B_\ell \eff \{\mu\}}{\theta_2\theta_1}{\theta_2 P\cup Q}{\Sigma}}

  \inferrule*[Lab=\tilab{Handle}]
  {
    \infer{\Gamma}{H}{A\eff\{R_1\} \hto D}{\theta_1}{P_1}{\Sigma_1} \\
    \infer{\theta_1\Gamma}{M}{A'\eff\{R\}}{\theta_2}{P_2}{\Sigma_2} \\\\
    \unify{\theta_2 A}{A'}{\theta_3} \\
    P = \theta_3 (\theta_2 P_1 \cup P_2) \\
    Q = \makesub{\theta_3 R,\theta_3\theta_2 R_1}\cup
        \makeunl{\theta_3\theta_2\theta_1\Gamma |_{\Sigma_1\cap\Sigma_2}}
  }
  {\infer{\Gamma}{\Handle M \With H}{\theta_3\theta_2 D}
    {\theta_3\theta_2\theta_1}{P \cup Q}{\Sigma_1 \cup \Sigma_2}}

  \inferrule*[Lab=\tilab{Handler}]
  {
    \alpha,\phi_i,\mu \fresh \\
    \infer{\Gamma, x:\alpha}{M}{D}{\theta_0}{P_0}{\Sigma_0} \\\\
    [\infer{\theta_{i-1}(\Gamma, p_i : A_{\ell_i}, r_i : {B_{\ell_i} \to^{\phi_i} D})}
           {N_i}{D_i}{\theta_i'}{P_i}{\Sigma_i} \\
      \unify{D_i}{\theta_i'\theta_{i-1} D}{\theta_i''} \\
      \theta_i = \theta_i''\theta_i'\theta_{i-1}]_{i=1}^n \\\\
    C = \theta_n (\alpha\eff\{(\ell_i : A_{\ell_i} \sto^{\phi_i} B_{\ell_i})_i\lsep\mu\}) \\
    B\eff\{R\} = \theta_n D \\\\
    \Sigma = (\Sigma_0 \backslash \{x\}) \cup
              (\cup_{i=1}^{n} (\Sigma_i \backslash \{p_i, r_i\})) \\
    P = (\cup_{i=0}^{n} \theta_n P_i)
      \cup \makeunl{\theta_n\Gamma |_{\Sigma}}
      \cup \makesub{\mu, R}
      \cup \makelack{\mu, \{\ell_i\}_i}
      \\
    Q = \makeunl{\theta_n(x:\alpha)|_{\setcomplement{\Sigma_0}}}
      \cup (\cup_{i=1}^n\makeunl{\theta_n(p_i:A_{\ell_i},r_i:B_{\ell_i}\to^{\phi_i}D)})
  }
  {\infer{\Gamma}{\{\Ret x \mapsto M\} \uplus
                  \{ \ell_i\;p_i\;r_i \mapsto N_i \}_{i=1}^n }{C \hto \theta_n D}
    {\theta_{n}}{P\cup Q}{\Sigma}}

  \end{mathpar}

  \[
    \begin{split}
    \makeleq{\Gamma, \tau} &= \factor(\Gamma \linle \tau) \\
    \makeunl{\Gamma} &= \makeleq{\Gamma, \Unl}
    \end{split}\qquad
    \begin{split}
    \makesub{R_1, R_2} &= \factor(R_1 \rle R_2) \\
    \makelack{R, \mL} &= \factor(R \rlack \mL)
    \end{split}
  \]

  \caption{Type Inference of \CalcQ}
  \label{fig:qeffpopsub-type-inference}
\end{figure}

\begin{figure}[htbp]
\begin{minipage}[t]{0.48\textwidth}
\[
\bl
\munify : \meta{USet} \rightharpoonup \meta{Subst}
\medskip \\
\munify(\cdot) = \mreturn\,\iota
  \medskip
  \\
\munify(\alpha\sim\alpha, U) = \munify(U)
  \medskip
  \\

\munify(\alpha\sim\tau, U) = \\
\quad\bl
  \massert\ \alpha\notin\ftv{\tau} \\
  \mlet\ \theta = [\tau/\alpha] \\
  \munify(\theta U) \theta
  \el
  \medskip
  \\

\munify(\tau\sim\alpha, U) = \\
\quad \munify(\alpha\sim\tau, U)
  \medskip
  \\

\munify(A\eff \{R\}\sim A'\eff \{R'\}, U) = \\
\quad \munify({A\sim A', R\sim R'}, U)
  \medskip
  \\

\munify((A\to^Y C)\sim (A'\to^{Y'} C'), U) = \\
\quad \munify({A\sim A', C\sim C', Y\sim Y'}, U)
  \medskip
  \\

\munify(Y \sim Y', U) = \\
\quad\bl
  \mlet\ \theta = \unifylin(Y, Y') \\
  \munify(\theta U)\theta
  \el
  \medskip
  \\
\el
\]
\end{minipage}
\begin{minipage}[t]{0.48\textwidth}
\[\bl
\munify(\cR_1\sim \cR_2, U) = \\
\quad\bl
  \mlet\ (\cR_1', \cR_2', \theta) = \unifyrow(\cR_1, \cR_2) \\
  \massert\ \rowset{\cR_1'} = \rowset{\cR_2'} \\
  \munify(\theta U) \theta
  \el
  \medskip
  \\

\munify(\cR_1\lsep\mu_1\sim \cR_2, U) = \\
\quad\bl
  \mlet\ (\cR_1',\cR_2',\theta) = \unifyrow(\cR_1, \cR_2) \\
  \massert\ \rowset{\cR_1'}\subseteq \rowset{\cR_2} \\
  \mfresh\ \mu \\
  \mlet\ \theta' = [((\cR_2'\backslash \cR_1')\lsep\mu)/\mu_1] \\
  \munify(\theta'\theta U) \theta'\theta
  \el
  \medskip
  \\

\munify(\cR_2\sim \cR_1\lsep\mu_1, U) = \\
\quad \munify(\cR_1\lsep\mu_1\sim \cR_2, U)
  \medskip
  \\

\munify(\cR_1\lsep\mu_1\sim \cR_2\lsep\mu_2, U) = \\
\quad\bl
  \mlet\ (\cR_1',\cR_2',\theta) = \unifyrow(\cR_1, \cR_2) \\
  \mfresh\ \mu \\
  \ba{@{}l@{}l}
  \mlet\ \theta' = [&((\cR_2'\backslash \cR_1')\lsep\mu)/\mu_1, \\
                    &((\cR_1'\backslash \cR_2')\lsep\mu)/\mu_2] \\
  \ea \\
  \munify(\theta'\theta U) \theta'\theta
  \el
  \medskip
  \\
\el\]
\end{minipage}

\caption{Unification of \CalcS}
\label{fig:unification}
\end{figure}

\subsection{Constraint Solving Algorithm}
\label{app:qfp-constraint-solving}

The constraint solving algorithm of \CalcS is given in
\Cref{fig:constraint-solving}.

The function $\unifylin$ unifies two linearity types.
The function $\unifylab$ unifies the signatures of shared labels of
two concrete rows.
The function $\unifyrow$ wraps $\unifylab$.
The function $\solvelin$ computes the transitive closure of linearity
constraints.

The function $\solverow(\theta,P,Q)$ solves row constraints.
It takes the current substitution $\theta$ and the currently solved
predicate set $P$, and solves the predicates in $Q$.
The basic idea is to transform the row subtyping predicates to forms
of $\mu\rle R$ and row lacking predicates to forms of $\mu\rlack\mL$,
which we call \emph{solved forms}.
It does a case analysis on the first predicate in $Q$.
For instance, consider the most complicated case
$\cR_1\lsep\mu_1\rle\cR_2\lsep\mu_2$.
It first unifies the common labels of $\cR_1$ and $\cR_2$.
When $\cR_1$ is a subset of $\cR_2$, we can directly transform it to
the solved form; otherwise, we allocate a fresh row variable to
substitute $\mu_2$ and transform it to the solved form.
Note that we also need to move all previously solved predicates to the
unsolved predicate set, because the row variable $\mu_2$ is
substituted, which might turn some predicates in solved forms to
unsolved forms.

The main function $\solve$ sequentially solves row constraints using
$\solverow$ and linearity constraints using $\solvelin$.
Note that we use $\factor$ to factorise the output predicate set to
transform the linearity constraints into the simplest form (i.e., only
between value type variables, row variables, and linearity), which is
suitable for computing the transitive closure using $\solvelin$.

\begin{figure}[h]
  \begin{minipage}[t]{0.48\textwidth}
  \[
  \bl
  \ba{@{}l@{\;}l@{\;}l@{}}
  \solverow &: &(\meta{Subst} \times \meta{PSet} \times \meta{PSet}) \\
  &\rightharpoonup &(\meta{Subst} \times \meta{PSet})
  \ea
  \medskip \\
  \solverow(\theta,P,\cdot) = \mreturn\ (\theta, P)
    \medskip \\

  \solverow(\theta,P,(\tau_1\linle\tau_2, Q)) = \\
  \quad\bl
    \solverow(\theta,(P,\tau_1\linle\tau_2), Q)
    \el
    \medskip \\

  \solverow(\theta, P, (\cR_1 \rle \cR_2, Q)) = \\
  \quad\bl
    \mlet\ (\cR_1', \cR_2', \theta') = \unifyrow(\cR_1, \cR_2) \\
    \massert\ \rowset{\cR_1'} \subseteq \rowset{\cR_2'} \\
    \solverow(\theta'\theta, \theta'P, \theta'Q)
    \el
    \medskip \\

  \solverow(\theta, P, (\cR_1 \lsep \mu \rle \cR_2 \lsep \mu, Q)) = \\
  \quad\bl
    \mlet\ (\cR_1', \cR_2', \theta') = \unifyrow(\cR_1, \cR_2) \\
    \massert\ \rowset{\cR_1'} \subseteq \rowset{\cR_2'} \\
    \solverow(\theta'\theta, \theta' P, \theta' Q)
    \el
    \medskip \\

  \solverow(\theta, P, (\cR_1 \lsep \mu \rle \cR_2, Q)) = \\
  \quad\bl
    \mlet\ (\cR_1', \cR_2', \theta') = \unifyrow(\cR_1, \cR_2) \\
    \massert\ \rowset{\cR_1'} \subseteq \rowset{\cR_2'} \\
    \solverow(\theta'\theta,
      (\theta'P, \mu \rle (\cR_2'\backslash \cR_1')), \theta' Q)
    \el
    \medskip \\

  \solverow(\theta, P, (\cR_1 \rle \cR_2 \lsep \mu_2, Q)) = \\
  \quad\bl
    \mlet\ (\cR_1', \cR_2', \theta') = \unifyrow(\cR_1, \cR_2) \\
    \meta{if}\ \rowset{\cR_1'}\subseteq\rowset{\cR_2'}\\
    \meta{then}\
      \solverow(\theta'\theta, \theta'P, \theta'Q) \\
    \meta{else}\
    \bl
      \mfresh\ \mu \\
      \mlet\ \theta'' = [((\cR_1'\backslash \cR_2')\lsep \mu) / \mu_2]\theta' \\
      \solverow(\theta''\theta, \cdot, \theta''(Q,P))
      \el
    \el
    \medskip \\

  \solverow(\theta, P, (\cR_1 \lsep \mu_1 \rle \cR_2 \lsep \mu_2, Q)) = \\
  \quad\bl
    \mlet\ (\cR_1', \cR_2', \theta') = \unifyrow(\cR_1, \cR_2) \\
    \meta{if}\ \rowset{\cR_1'}\subseteq\rowset{\cR_2'}\\
    \meta{then}\
    \bl
      \solverow(\theta'\theta,(\theta'P,\mu_1\rle(\cR_2'\backslash\cR_1')\lsep\mu_2),\theta'Q)
      \el\\
    \meta{else}\
    \bl
      \mfresh\ \mu \\
      \mlet\ \theta'' = [((\cR_1'\backslash \cR_2')\lsep \mu) / \mu_2]\theta' \\
      \solverow(\theta''\theta,%
      \mu_1 \rle (\cR_2'\backslash \cR_1')\lsep \mu, \theta''(Q,P))
      \el
    \el
    \medskip \\

  \solverow(\theta, P, (\cR \rlack \mL, Q)) = \\
  \quad\bl
    \massert\ \dom{\cR} \cap \mL = \emptyset \\
    \solverow(\theta, P, Q)
    \el
    \medskip \\

  \solverow(\theta, P, (\cR\lsep\mu \rlack \mL, Q)) = \\
  \quad\bl
    \massert\ \dom{\cR} \cap \mL = \emptyset \\
    \solverow(\theta, (P, \mu \rlack \mL), Q)
    \el
  \el
  \]
  \end{minipage}
  \begin{minipage}[t]{0.48\textwidth}
  \[
  \bl

  \ba{@{}l@{\;}l@{\;}l@{}}
  \unifyrow &: &(\meta{CRow} \times \meta{CRow}) \\
  &\rightharpoonup &(\meta{CRow}\times\meta{CRow}\times\meta{Subst})\\
  \ea
  \\
  \unifyrow(\cR, \cR') = \\
  \quad\bl
    \mlet\ \theta = \unifylab(\cR,\cR') \\
    \mreturn\ (\theta\cR, \theta\cR', \theta)
    \el
    \bigskip\\

  \ba{@{}l@{\;}l@{\;}l@{}}
  \unifylab &: &(\meta{CRow} \times \meta{CRow}) \rightharpoonup \meta{Subst}
  \ea
  \\

  \unifylab(\cdot, \cR) = \mreturn\ \iota
    \\

  \unifylab(\cR, \cdot) = \mreturn\ \iota
    \\

  \unifylab((\ell:Y_1\lsep \cR_1), (\ell:Y_2\lsep \cR_2)) = \\
  \quad\bl
    \mlet\ \theta = \unifylin(Y_1, Y_1) \\
    \mlet\ \theta' = \unifylab(\theta \cR_1, \theta \cR_2) \\
    \mreturn\  \theta'\theta
    \el
    \\

  \unifylab((\ell:Y\lsep \cR_1), \cR_2) = \unifylab(\cR_1, \cR_2)
    \\

  \unifylab(\cR_1, (\ell:Y\lsep \cR_2)) = \unifylab(\cR_1, \cR_2)
    \bigskip \\

  \unifylin : (\meta{Lin} \times \meta{Lin}) \rightharpoonup \meta{Subst}
  \\

  \ba{@{}l@{\;}l@{}}
  \unifylin(Y, Y) &= \mreturn\ \iota \\
  \unifylin(\Unl, \Lin) &= \mfail \\
  \unifylin(\Lin, \Unl) &= \mfail \\
  \unifylin(\phi, Y) &= \mreturn\ [Y/\phi] \\
  \unifylin(Y, \phi) &= \mreturn\ [Y/\phi] \\
  \ea
  \bigskip \\

  \solvelin : (\meta{PSet} \times \meta{PSet}) \to \meta{PSet}
  \\

  \solvelin(P, \cdot) = P \\
  \solvelin(P, (R_1\rle R_2, Q)) = \solvelin(P, Q) \\
  \solvelin(P, (\tau_1\linle\tau_2, Q)) = \solvelin(P\cup P'', Q) \\
  \quad\meta{where}\\
  \quad\quad\bl
    P' = P \cup \{\tau_1\linle\tau_1,\tau_2\linle\tau_2\} \\
    P'' =
      \{\tau_1'\linle\tau_2' \mid
      \{\tau_1'\linle\tau_1,\tau_2\linle\tau_2'\} \subseteq P'\}\\
    \el
  \bigskip \\

  \solve : \meta{PSet} \rightharpoonup (\meta{Subst}\times\meta{PSet}) \\
  \solve(P) = \\
  \quad\bl
    \mlet\ (\theta,Q) = \solverow(\iota,\cdot,P) \\
    \mlet\ Q' = \solvelin(\cdot, \factor(Q)) \\
    \massert\ (\Lin\linle\Unl) \notin Q' \\
    \mreturn\ (\theta, Q)
    \el
  \el
  \]
  \end{minipage}

\caption{Constraint Solving of \CalcS}
\label{fig:constraint-solving}
\end{figure}

\FloatBarrier
\section{Proofs of \CalcS}

In this section, we prove the theorems in \Cref{sec:qeffpopsub}.

\subsection{Correctness of Factorisation}
\label{app:factorisation}

We first prove some useful properties of the entailment relations.

\begin{restatable}[Properties of the entailment relation]{theorem}{propertyEntailmentRelation}
\label{thm:properties-of-the-entailment-relation}
The entailment relation between predicate sets satisfies the following
properties:
\begin{itemize}
  \item Monotonicity. If $Q \subseteq P$, then $P\vdash Q$.
  \item Transitivity. If $P_1\vdash P_2$ and $P_2\vdash P_3$, then $P_1\vdash P_3$.
  \item Closure property. If $P\vdash Q$, then $\theta P\vdash \theta Q$.
  \item Weakening. If $P\vdash Q$, then $P, P'\vdash Q$.
\end{itemize}
\end{restatable}

\begin{proof} ~\\
\noindent\emph{Monotonicity.} Directly follows from \plab{Subsume} and
\plab{PredSet}.

\noindent\emph{Transitivity.} By \plab{PredSet}, we only need to prove
that if $P_1\vdash P_2$ and ${P_2}\vdash{\pi}$, then
${P_1}\vdash{\pi}$. By straightforward induction on
${P_2}\vdash{\pi}$.

\noindent\emph{Closure property.} By \plab{PredSet}, we only need to
prove that if ${P}\vdash{\pi}$ then ${\theta P}\vdash{\theta\pi}$. By
straightforward induction on ${P}\vdash{\pi}$.

\noindent\emph{Weakening.} By \plab{PredSet}, we only need to prove
that if ${P}\vdash{\pi}$ then
${P,P'}\vdash{\pi}$. By straightforward induction
on ${P}\vdash{\pi}$.

\end{proof}

\begin{lemma}[Inverse closure property]
  \label{lemma:closure-extended}
  If $P\vdash\theta(\sigma\linle\tau)$, then there exists
  $P'\vdash\sigma\linle\tau$ such that $P\vdash\theta P'$.
\end{lemma}
\begin{proof}
  By induction on the entailment relations.
  \begin{description}
    \item[Case]
    \begin{mathpar}
  \inferrule*[Lab=\plab{Quantifier}]
  {\refa{P \vdash [\tau'/\alpha] \theta (\sigma \linle \tau)} \text{ for some } \tau'}
  {P \vdash \theta ((\forall\alpha . \sigma) \linle \tau)}
    \end{mathpar}
    Assume that $\alpha\notin\dom{\theta}$ and
    $\alpha\notin\ftv{\tau}$ without loss of generality. We can
    commute $[\tau'/\alpha]$ and $\theta$ in \refa{}.
    By the IH on \refa{}, there exists $P'\vdash
    [\tau'/\alpha](\sigma\linle\tau)$ such that $P\vdash\theta P'$.
    By \plab{Quantifier}, we have $P'\vdash(\forall\alpha.\sigma)\linle\tau$.

    \item[Case]
    \begin{mathpar}
  \inferrule*[Lab=\plab{Qualifier}]
  {
  \refa{P \vdash \theta\pi} \\ \refb{P \vdash \theta(\rho \linle \tau)}
  }
  {P \vdash \theta((\pi\qto\rho) \linle \tau)}
    \end{mathpar}
    By the IH on \refa{}, there exists $P_1\vdash\pi$ such that
    $P\vdash\theta P_1$.
    By the IH on \refb{}, there exists $P_2\vdash\rho\linle\tau$ such
    that $P\vdash\theta P_2$.
    By \plab{Qualifier}, we have $P_1\cup
    P_2\vdash(\pi\qto\rho)\linle\tau$.
    By \plab{PredSet}, we have $P\vdash\theta(P_1\cup P_2)$.

    \item[Case]
    For all other cases of $P\vdash\theta\pi$, just take $P'=\pi$.
  \end{description}
\end{proof}

\constraintFactorisation*

\begin{proof}
  The first part of the theorem is kind of obvious because
  $\factor(P)$ is almost directly defined from the entailment rules in
  \Cref{fig:qe-predicates-entailment}.
  We prove the auxiliary lemma that if $\factor(\pi) = Q$, then
  ${Q}\vdash{\pi}$ and ${\pi}\vdash{Q}$.
  Both directions follow from straightforward induction on the
  definition of $\factor$.
  Note that in the proof of ${\pi}\vdash{Q}$, we apply the bottom-up
  direction of the two-way rules \plab{Fun} and \plab{Row}.
  Then, given $\factor(P) = \bigcup_{\pi\in P}\factor(\pi)$, by the
  lemma we have ${\factor(\pi)}\vdash{\pi}$ for all $\pi\in P$, which
  then give $\bigcup_{\pi\in P}\factor(\pi)\vdash P$ by \plab{PredSet}
  and the \emph{weakening} of
  \Cref{thm:properties-of-the-entailment-relation}.
  We also have that ${\pi}\vdash{\factor(\pi)}$ for all $\pi\in P$,
  which then give $P \vdash \bigcup_{\pi\in P}{\factor(\pi)}$ by
  \plab{PredSet} and the \emph{weakening} of
  \Cref{thm:properties-of-the-entailment-relation}.

  For the second part of the theorem, we prove the auxiliary lemma
  that if $\factor(\sigma\linle\tau) = Q$, then
  {${Q}\vdash{\sigma\linle\tau}$} and {for any
  ${P}\vdash{\sigma\linle\tau}$, there exists $\theta$ such that
  $P\vdash \theta Q$}.
  The ${Q}\vdash{\sigma\linle\tau}$ follows from straightforward
  induction on the definition of $\factor$.
  The other direction is more involved.
  We proceed by induction on the definition of $\factor$.
  \begin{description}
    \item[Case]
  $$
\factor((\forall\alpha.\sigma) \linle \tau) = \refa{\factor([\beta/\alpha] \sigma \linle \tau)}
  \text{ for some fresh }\beta
  $$

  Suppose $\factor((\forall\alpha.\sigma) \linle \tau) = Q$.
  We want to show that for any \refb{$P\vdash (\forall\alpha.\sigma) \linle
  \tau$}, there exists $\theta$ such that $P\vdash\theta Q$.
  By \refb{} and \plab{Quantifier}, there exists
  $\theta_1=[\tau'/\alpha]$ such that
  $P\vdash\theta_1\sigma\linle\tau$.
  Let $\theta_2 = [\tau'/\beta]$. We have
  $P\vdash\theta_2[\beta/\alpha]\sigma\linle\tau$.
  By \Cref{lemma:closure-extended}, there exists $P'$ such that
  \refc{$P'\vdash[\beta/\alpha]\sigma\linle\tau$} and $P\vdash\theta_2
  P'$.
  By \refc{} and the IH on \refa{}, there exists $\theta_3$ such that
  $P'\vdash\theta_3 Q$.
  Then, by the closure property and transitivity of
  \Cref{thm:properties-of-the-entailment-relation}, we have
  $P\vdash\theta_2\theta_3 Q$.

  \item[Case]
  $$
\factor((\pi\qto\sigma) \linle \tau) = \refa{\factor(\pi)} \cup \refb{\factor(\sigma\linle\tau)}
  $$

  Suppose ${\factor(\pi)} = Q_1$ and ${\factor(\sigma\linle\tau)} =
  Q_2$.
  For any $P\vdash(\pi\qto\sigma)\linle\tau$, by \plab{Qualifier}, we
  have $P\vdash\pi$ and $P\vdash\sigma\linle\tau$.
  By the IH on \refa{},  there exists $\theta_1$ such that
  $P\vdash\theta_1 Q_1$.
  By the IH on \refb{}, there exists $\theta_2$ such that
  $P\vdash\theta_2 Q_2$.
  Note that $\dom{\theta_1} \cap \dom{\theta_2} = \emptyset$.
  Thus, we have $P\vdash\theta_1\theta_2 (Q_1\cup Q_2)$.

  \item[Case]
  $$
  \factor(\pi) = Q
  $$

  By the first part of the theorem which has been proved, we have
  $\pi\vdash Q$.
  For any $P\vdash\pi$, by the transitivity of
  \Cref{thm:properties-of-the-entailment-relation}, we have $P\vdash
  Q$.
  \end{description}
  With this lemma, our goal follows from a similar analysis to the
  proof of the first part of the theorem since \plab{Context} and
  \plab{PredSet} are both conjunction rules.

\end{proof}

\subsection{Principal Unifier}
\label{app:unification}

We have the following lemmas for the unification function in
\Cref{fig:unification} and its auxiliary functions.

\begin{lemma}[Principal auxiliary unifiers]
  \label{lemma:principal-auxiliary-unification}
  Given $\cR_1$ and $\cR_2$, let $\cR_1' = (\cR_1
  |_{\dom{\cR_1}\cap\dom{\cR_2}})$ and $\cR_2' = (\cR_2
  |_{\dom{\cR_1}\cap\dom{\cR_2}})$. If $\unifylab(\cR_1,\cR_2) =
  \theta$, then
  for any $\theta'\cR_1' = \theta'\cR_2'$, there exists $\theta''$
  such that $\theta'=\theta''\theta$; if it fails, then $\cR_1'$ and
  $\cR_2'$ cannot be unified.
\end{lemma}
\begin{proof}
  By straightforward induction on the definition of $\unifyrow$,
  $\unifylab$ and $\unifylin$.
\end{proof}

\begin{lemma}[Principal unifiers]
  \label{lemma:principal-unifiers}
  If $\unify{A}{B}{\theta}$, then for any $\theta'A=\theta'B$, there
  exists $\theta''$ such that $\theta' = \theta''\theta$; if it fails,
  then $A$ and $B$ cannot be unified. The same applies to computation
  types.
\end{lemma}
\begin{proof}
  By straightforward induction on the definition of $\munify(U)$.
\end{proof}

\subsection{Soundness and Completeness of Type Inference}
\label{app:proof-type-inference}

We prove the soundness and completeness of type inference as well as
auxiliary lemmas.

\begin{lemma}[Closure property of typing]
  \label{lemma:typing-closure}
  If $\typ{P\mid\Gamma}{V:A}$, then $\typ{\theta
  P\mid\theta\Gamma}{V:\theta A}$. The same applies to computation and
  handler typing.
\end{lemma}
\begin{proof}
  By the closure property of
  \Cref{thm:properties-of-the-entailment-relation} and straightforward
  induction on the typing derivations.
\end{proof}

\begin{lemma}[Weakening of predicates]
  \label{lemma:typing-weakening}
  If $\typ{P\mid\Gamma}{V:A}$, then $\typ{P, P'\mid\Gamma}{V: A}$. The
  same applies to computation and handler typing.
\end{lemma}
\begin{proof}
  By the weakening property of
  \Cref{thm:properties-of-the-entailment-relation} and straightforward
  induction on the typing derivations.
\end{proof}

\begin{lemma}[Extra is unlimited]
  \label{lemma:extra-is-unlimited}
If $\typ{P\mid\Gamma}{V:A}$, then $\typ{P'\mid\Gamma,x:\sigma}{V:A}$
for any $P'\vdash P$ and $P'\vdash \sigma\linle\Unl$. The same applies
to computation and handler typing.
\end{lemma}
\begin{proof}
  By straightforward induction on the typing derivations.
\end{proof}

\soundness*
\begin{proof}
  By mutual induction on the type inference derivations
  $\infer{\Gamma}{V}{A}{\theta}{P}{\Sigma}$,
  $\infer{\Gamma}{M}{C}{\theta}{P}{\Sigma}$, and
  $\infer{\Gamma}{H}{C\hto D}{\theta}{P}{\Sigma}$.

\begin{description}
  \item[Case]
  \begin{mathpar}
    \inferrule*[Lab=\tilab{Var}]
    {
      (x : \forall \ol{\alpha} . P \qto A) \in \Gamma \\
      \ol\beta \fresh \\
      \refa{\theta = [\ol{\beta}/ \ol{\alpha}]}
    }
    {\infer{\Gamma}{x}{\theta A}{\theta}{\theta P}{\{x\}}}
  \end{mathpar}
  By \refa{}, we have $\theta P\qto\theta A
  \sqgen\theta(\forall\ol\alpha.P\qto A)$.
  Our goal then follows from \qtylab{Var}.

  \item[Case]
  \begin{mathpar}
    \inferrule*[Lab=\tilab{Let}]
    {
      \refa{\infer{\Gamma}{V}{A}{\theta_1}{P_1}{\Sigma_1}} \\
      \sigma = \gen(\theta_1 \Gamma, P_1 \qto A) \\\\
      \refb{\infer{\theta_1\Gamma, x:\sigma}{M}{C}{\theta_2}{P_2}{\Sigma_2}} \\\\
      Q = \makeunl{\theta_2\theta_1\Gamma |_{\Sigma_1 \cap \Sigma_2}} \cup
          \makeunl{\theta_2(x:\sigma) |_{\setcomplement{\Sigma_2}}}
    }
    {\infer{\Gamma}{\Let x = V \In M}{C}
      {\theta_2\theta_1}{P_2\cup Q}{\Sigma_1 \cup (\Sigma_2 \backslash x)}}
  \end{mathpar}
  By the IH on \refa{}, we have $\typ{P_1\mid\theta_1\Gamma|_{\Sigma_1}}{V:A}$.
  By \Cref{lemma:typing-closure}, we have \refc{$\typ{\theta_2
  P_1\mid\theta_2\theta_1\Gamma}{V:\theta_2 A}$}.
  By the IH on \refb{}, we have \refd{$\typ{P_2\mid
  \theta_2(\theta_1\Gamma,x:\sigma)|_{\Sigma_2}}{M:\theta_2 C}$}.
  Let $\sigma'=\gen(\theta_2\theta_1\Gamma,\theta_2 P_1\qto\theta_2
  A)$.
  Notice that $\theta_2$ is generated by the type inference judgement
  \refb{}, which cannot substitute any variables bound by $\sigma$
  (i.e., variables in $\ftv{P_1\qto
  A}\backslash\ftv{\theta_1\Gamma}$).
  Thus, we have $\theta_2\sigma = \sigma'$.
  Let $\Sigma_2'=\Sigma_2\backslash x$, $\Gamma_1 =
  (\theta_2\theta_1\Gamma)|_{\Sigma_1\backslash\Sigma_2'}$, $\Gamma_2
  = (\theta_2\theta_1\Gamma)|_{\Sigma_2'\backslash\Sigma_1}$, $\Gamma'
  = (\theta_2\theta_1\Gamma)|_{\Sigma_1\cap\Sigma_2'}$.
  By \refc{} and \refd{}, we have \refe{$\typ{\theta_2
  P_1\mid\Gamma_1,\Gamma'}{V:\theta_2 A}$} and {$\typ{P_2\mid
  \Gamma_2,\Gamma',(x:\sigma')|_{\Sigma_2}}{M:\theta_2 C}$}.
  By \Cref{lemma:extra-is-unlimited} we have
  \reff{$\typ{P_2\cup\makeunl{(x:\sigma')|_{\setcomplement{\Sigma_2}}}
  \mid\Gamma_2,\Gamma',x:\sigma'}{M:\theta_2 C}$}.
  By \Cref{thm:correctness-of-factorisation}, we have $Q\vdash
  \Gamma'\linle\Unl$.
  Our goal follows from \qtylab{Let}, \refe{}, \reff{} and
  \Cref{lemma:typing-weakening}.

  \item[Case]
  \begin{mathpar}
    \inferrule*[Lab=\tilab{Abs}]
    {
      \alpha,\phi \fresh\\
      \refa{\infer{\Gamma, x:\alpha}{M}{C}{\theta}{P}{\Sigma}} \\\\
      Q = \makeleq{\theta\Gamma |_{\Sigma}, \phi} \cup
          \makeunl{\theta(x:\alpha) |_{\setcomplement{\Sigma}}}
    }
    {\infer{\Gamma}{\lambda x . M}{\theta\alpha \to^\phi C}
      {\theta}{P \cup Q}{\Sigma \backslash x}}
  \end{mathpar}
  By the IH on \refa{}, we have
  $\typ{P\mid(\theta\Gamma,x:\theta\alpha)|_\Sigma}{M:C}$.
  Our goal follows from \Cref{lemma:extra-is-unlimited},
  \Cref{thm:correctness-of-factorisation} and \qtylab{Abs}.

  \item[Case]
  \begin{mathpar}
    \inferrule*[Lab=\tilab{App}]
    {
      \refa{\infer{\Gamma}{V}{A}{\theta_1}{P_1}{\Sigma_1}} \\
      \refb{\infer{\theta_1\Gamma}{W}{B}{\theta_2}{P_2}{\Sigma_2}} \\\\
      \alpha, \mu, \phi \fresh \\
      \refc{\unify{\theta_2 A}{(B \to^\phi \alpha\eff\mu)}{\theta_3}} \\\\
      P = \theta_3 (\theta_2 P_1 \cup P_2) \\
      Q = \makeunl{\theta_3\theta_2\theta_1\Gamma |_{\Sigma_1 \cap \Sigma_2}}
    }
    {\infer{\Gamma}{V\; W}{\theta_3(\alpha\eff\mu)}
      {\theta_3\theta_2\theta_1}{P\cup Q}{\Sigma_1 \cup \Sigma_2}}
  \end{mathpar}
  By the IH on \refa{}, we have $\typ{P_1\mid\theta_1\Gamma}{V:A}$.
  By the IH on \refb{}, we have $\typ{P_2\mid\theta_2\theta_1\Gamma}{W:B}$.
  By \Cref{lemma:typing-closure}, we have \refd{$\typ{\theta_3\theta_2
  P_1\mid\theta_3\theta_2\theta_1\Gamma}{V:\theta_3\theta_2 A}$} and
  \refe{$\typ{\theta_3 P_2\mid\theta_3\theta_2\theta_1\Gamma}{W:\theta_3 B}$}.
  By \refc{}, we have $\theta_3\theta_2 A = \theta_3
  (B\to^\phi\alpha\eff\mu)$.
  Let $\Gamma_1 =
  (\theta_3\theta_2\theta_1\Gamma)|_{\Sigma_1\backslash\Sigma_2}$, $\Gamma_2
  = (\theta_3\theta_2\theta_1\Gamma)|_{\Sigma_2\backslash\Sigma_1}$, $\Gamma'
  = (\theta_3\theta_2\theta_1\Gamma)|_{\Sigma_1\cap\Sigma_2}$.
  By \refd{} and \refe{}, we have \reff{$\typ{\theta_3\theta_2 P_1\mid
  \Gamma_1,\Gamma'}{V:\theta_3\theta_2 A}$} and \refg{$\typ{\theta_3
  P_2\mid\Gamma_2,\Gamma'}{W:\theta_3 B}$}.
  By \Cref{thm:correctness-of-factorisation}, we have
  $Q\vdash\Gamma'\linle\Unl$.
  Our goal follows from \qtylab{App}, \reff{}, \refg{},
  and \Cref{lemma:typing-weakening}.

  \item[Case]
  \begin{mathpar}
    \inferrule*[Lab=\tilab{Return}]
    {
      \refa{\infer{\Gamma}{V}{A}{\theta}{P}{\Sigma}} \\
      \mu \fresh
    }
    {\infer{\Gamma}{\Ret V}{A \eff \{\mu\}}{\theta}{P}{\Sigma}}
  \end{mathpar}
  Our goal follows from the IH on \refa{} and \qtylab{Return}.

  \item[Case]
  \begin{mathpar}
    \inferrule*[Lab=\tilab{Do}]
    {
      \refa{\infer{\Gamma}{V}{A}{\theta_1}{P}{\Sigma}} \\
      \unify{A}{A_\ell}{\theta_2} \\\\
      \mu, \phi \fresh \\
      Q = \makesub{(\ell: A_\ell \sto^\phi B_\ell), \mu} \\
    }
    {\infer{\Gamma}{\Do \ell\; V}
      {B_\ell \eff \{\mu\}}{\theta_2\theta_1}{\theta_2 P\cup Q}{\Sigma}}
  \end{mathpar}
  Our goal follows from the IH on \refa{}, \qtylab{Do},
  \Cref{thm:correctness-of-factorisation}, and
  \Cref{lemma:typing-closure}.

  \item[Case]
  \begin{mathpar}
    \inferrule*[Lab=\tilab{Seq}]
    {
      \refa{\infer{\Gamma}{M}{A\eff\{R_1\}}{\theta_1}{P_1}{\Sigma_1}} \\
      \refb{\infer{\theta_1\Gamma,x:A}{N}{B\eff\{R_2\}}{\theta_2}{P_2}{\Sigma_2}} \\
      \mu \fresh \\\\
      Q = \makeunl{\theta_2\theta_1\Gamma |_{\Sigma_1 \cap \Sigma_2}} \cup
          \makeunl{\theta_2(x:A) |_{\setcomplement{\Sigma_2}}} \cup
          \makeleq{\theta_2\theta_1\Gamma |_{\Sigma_2}, \theta_2 R_1} \cup
          \makesub{\theta_2 R_1, \mu} \cup
          \makesub{R_2, \mu} \\
    }
    {\infer{\Gamma}{\Let x \revto M \In N}{B\eff \mu}
      {\theta_2\theta_1}{\theta_2 P_1 \cup P_2 \cup Q}{\Sigma_1 \cup (\Sigma_2 \backslash x)}}
  \end{mathpar}
  Similar to the \tilab{Let} and \tilab{App} cases.
  Let $\Sigma_2'=\Sigma_2\backslash x$, $\Gamma_1 =
  (\theta_2\theta_1\Gamma)|_{\Sigma_1\backslash\Sigma_2'}$, $\Gamma_2
  = (\theta_2\theta_1\Gamma)|_{\Sigma_2'\backslash\Sigma_1}$, $\Gamma'
  = (\theta_2\theta_1\Gamma)|_{\Sigma_1\cap\Sigma_2'}$.
  By the IH on \refa{} and \Cref{lemma:typing-closure}, we have
  \refc{$\typ{\theta_2 P_1\mid\Gamma_1,\Gamma'}{M:\theta_2(A\eff\{R_1\})}$}.
  By the IH on \refb{}, we have
  \refd{$\typ{P_2\mid\Gamma_2,\Gamma',(x:A)|_{\Sigma_2}}{N:B\eff\{R_2\} }$}.
  Our goal follows from \qtylab{Seq}, \refc{}, \refd{},
  \Cref{thm:correctness-of-factorisation},
  \Cref{lemma:typing-weakening} and \Cref{lemma:extra-is-unlimited}.

  \item[Case]
  \begin{mathpar}
    \inferrule*[Lab=\tilab{Handle}]
    {
      \refa{\infer{\Gamma}{H}{A\eff\{R_1\} \hto D}{\theta_1}{P_1}{\Sigma_1}} \\
      \refb{\infer{\theta_1\Gamma}{M}{A'\eff\{R\}}{\theta_2}{P_2}{\Sigma_2}} \\\\
      \unify{\theta_2 A}{A'}{\theta_3} \\
      P = \theta_3 (\theta_2 P_1 \cup P_2) \\
      Q = \makesub{\theta_3 R,\theta_3\theta_2 R_1}\cup
          \makeunl{\theta_3\theta_2\theta_1\Gamma |_{\Sigma_1\cap\Sigma_2}}
    }
    {\infer{\Gamma}{\Handle M \With H}{\theta_3\theta_2 D}
      {\theta_3\theta_2\theta_1}{P \cup Q}{\Sigma_1 \cup \Sigma_2}}
  \end{mathpar}
  By a similar proof to the \qtylab{App} case, our goal follows from
  the IHs on \refa{} and \refb{},
  \Cref{thm:correctness-of-factorisation}, \Cref{lemma:typing-closure}
  and \Cref{lemma:typing-weakening}.

  \item[Case]
  \begin{mathpar}
    \inferrule*[Lab=\tilab{Handler}]
    {
      \alpha,\phi_i,\mu \fresh \\
      \refa{\infer{\Gamma, x:\alpha}{M}{D}{\theta_0}{P_0}{\Sigma_0}} \\\\
      [\refb{\infer{\theta_{i-1}(\Gamma, p_i : A_{\ell_i}, r_i : {B_{\ell_i} \to^{\phi_i} D})}
             {N_i}{D_i}{\theta_i'}{P_i}{\Sigma_i}} \\
        \unify{D_i}{\theta_i'\theta_{i-1} D}{\theta_i''} \\
        \theta_i = \theta_i''\theta_i'\theta_{i-1}]_{i=1}^n \\\\
      C = \theta_n (\alpha\eff\{(\ell_i : A_{\ell_i} \sto^{\phi_i} B_{\ell_i})_i\lsep\mu\}) \\
      B\eff\{R\} = \theta_n D \\\\
      \Sigma = (\Sigma_0 \backslash \{x\}) \cup
                (\cup_{i=1}^{n} (\Sigma_i \backslash \{p_i, r_i\})) \\
      P = (\cup_{i=0}^{n} \theta_n P_i)
        \cup \makeunl{\theta_n \Gamma |_{\Sigma}}
        \cup \makesub{\mu, R}
        \cup \makelack{\mu, \{\ell_i\}_i}
        \\
    Q = \makeunl{\theta_n(x:\alpha)|_{\setcomplement{\Sigma_0}}}
      \cup (\cup_{i=1}^n\makeunl{\theta_n(p_i:A_{\ell_i},r_i:B_{\ell_i}\to^{\phi_i}D)})
    }
    {\infer{\Gamma}{\{\Ret x \mapsto M\} \uplus
                    \{ \ell_i\;p_i\;r_i \mapsto N_i \}_{i=1}^n }{C \hto \theta_n D}
      {\theta_{n}}{P\cup Q}{\Sigma}}
  \end{mathpar}
  The type inference for handlers is the most complicated, but there
  is nothing really new about the proof compared to previous cases.
  By the IH on \refa{}, we have
  $\typ{P_0\mid\theta_0(\Gamma,x:\alpha)|_{\Sigma_0}}{M:D}$.
  By the IH on \refb{}, we have
  $\typ{P_i\mid\theta_i'\theta_{i-1}(\Gamma,p_i : A_{\ell_i}, r_i :
  B_{\ell_i} \to^{\phi_i} D)|_{\Sigma_i}}{N_i:D_i}$.
  By \Cref{lemma:typing-closure}, we have {$\typ{\theta_n
  P_0\mid\theta_n(\Gamma,x:\alpha)|_{\Sigma_0}}{M:\theta_n D}$} and
  {$\typ{\theta_n P_i\mid\theta_n(\Gamma,p_i : A_{\ell_i}, r_i :
  B_{\ell_i} \to^{\phi_i} D)|_{\Sigma_i}}{N_i:\theta_n D_i}$}.
  By \Cref{lemma:extra-is-unlimited}, we have
  $$
  \refc{\typ{\theta_n P_0 \cup \makeunl{\theta_n \Gamma|_{\Sigma}} \cup
  \makeunl{\theta_n(x:\alpha)|_{\setcomplement{\Sigma_0}}}\mid\theta_n(\Gamma|_{\Sigma},x:\alpha)}{M:\theta_n D}}
  $$
  and
  $$
  \refd{\typ{\theta_n P_i \cup \makeunl{\theta_n\Gamma|_{\Sigma}} \cup
  \makeunl{\theta_n(p_i:A_{\ell_i},r_i:B_{\ell_i}\to^{\phi_i}D)}
  \mid\theta_n(\Gamma|_{\Sigma},p_i :
  A_{\ell_i}, r_i : B_{\ell_i} \to^{\phi_i}
  D)}{N_i:\theta_n D_i}}
  $$
  By \Cref{thm:correctness-of-factorisation}, we have $P\cup Q \vdash
  \{\mu\rle R, \mu\rlack\{\ell_i\}_i\}$.
  Our goal follows from \qtylab{Handler}, \refc{}, \refd{}, and
  \Cref{lemma:typing-weakening}.
\end{description}
\end{proof}

\begin{lemma}[More general contexts]
  \label{lemma:more-general-contexts}
  If $\typ{P\mid\Gamma,x:\sigma}{V:A}$ and $\sigma\sqgen \sigma'$,
  then $\typ{P\mid\Gamma,x:\sigma'}{V:A}$. The same applies to
  computation and handler typing.
\end{lemma}
\begin{proof}
  By straightforward induction on the typing derivation.
\end{proof}

\begin{lemma}[Zero is unlimited]
  \label{lemma:zero-is-unlimited}
  If $\typ{P\mid\Gamma,x:\sigma}{V:A}$ and $x$ does not appear in $V$,
  then $P \vdash \sigma \linle \Unl$. The same applies to computation
  and handler typing.
\end{lemma}
\begin{proof}
  By straightforward induction on the typing derivation.
\end{proof}

\begin{lemma}[Closure property of factorisation]
  \label{lemma:closure-factor}
  If $\factor(P) = Q$, then $\factor(\theta P) = \theta Q$.
  If $\factor(\Gamma\linle\tau) = Q$, then
  $\factor(\theta(\Gamma\linle\tau)) = \theta Q$.
\end{lemma}
\begin{proof}
  By the closure property of
  \Cref{thm:properties-of-the-entailment-relation} and straightforward
  induction on the definition of $\factor$.
\end{proof}

\completeness*

\begin{proof}
  By mutual induction on the syntax-directed typing derivations
  $\typ{P\mid\Gamma}{V:A}$, $\typ{P\mid\Gamma}{M:C}$, and
  $\typ{P\mid\Gamma}{H:C\hto D}$.

\begin{description}
  \item[Case]
  \begin{mathpar}
  \inferrule*[Lab=\qtylab{Var}]
  { P \vdash \Gamma\un \\
    P \qto A \sqgen \forall\ol\alpha.Q\qto B}
  {\typ{P\mid\theta(\Gamma, x : \forall\ol\alpha.Q\qto B)}{x : A}}
  \end{mathpar}
  By $P \qto A \sqgen \forall\ol\alpha.Q\qto B$, there exists
  $\theta_1$ such that $A = \theta_1 B$ and $P \vdash \theta_1 Q$.
  By \tilab{Var}, we have the following derivation
  \begin{mathpar}
  \inferrule*[Lab=\tilab{Var}]
  {
    \ol\beta \fresh \\
    \theta' = [\ol{\beta}/ \ol{\alpha}]
  }
  {\infer{\Gamma, x : \forall\ol\alpha.Q\qto B}{x}{\theta' B}{\theta'}{\theta' Q}{\{x\}}}
  \end{mathpar}
  Let $\theta'' = \theta \theta_1 [\ol\alpha / \ol\beta]$, we have $A
  = \theta_1 B = \theta''\theta' B$, $P\vdash \theta_1 Q =
  \theta''\theta' Q$, and $(\theta = \theta''\theta') |_\Gamma$.

  \item[Case]
  \begin{mathpar}
  \inferrule*[Lab=\qtylab{Let}]
  { \refa{\typ{P_1\mid\theta(\Gamma_1, \Gamma)}{V : A}} \\
    \sigma = \gen(\theta(\Gamma_1, \Gamma), P_1\qto A) \\\\
    \refb{\typ{P_2\mid\theta(\Gamma_2, \Gamma), x : \sigma}{M : C}} \\
    P_2 \vdash \theta\Gamma\un
  }
  {\typ{P_2\mid\theta(\Gamma_1, \Gamma_2, \Gamma)}{\Let x = V \In M : C}}
  \end{mathpar}
  By the IH on \refa{}, we have
  $\infer{\Gamma_1,\Gamma}{V}{A'}{\theta_1}{P_1'}{\Sigma_1}$
  and there exists $\theta_1'$ such that $A=\theta_1'A'$,
  $P_1\vdash\theta_1' P_1'$, and $(\theta =
  \theta_1'\theta_1)|_{\Gamma_1,\Gamma}$.
  By context weakening, we have
  \refc{$\infer{\Gamma_1,\Gamma_2,\Gamma}{V}{A'}{\theta_1}{P_1'}{\Sigma_1}$}.
  We also have
  $\sigma = \gen(\theta(\Gamma_1,\Gamma_2,\Gamma),P_1\qto A)$.
  Let $\sigma' = \gen(\theta_1(\Gamma_1,\Gamma_2,\Gamma),P_1'\qto A')$.
  By $(\theta = \theta_1'\theta_1)|_{\Gamma_1,\Gamma}$, it is easy to
  see that $\sigma \sqgen \theta_1'\sigma'$.
  Then by \refb{} and \Cref{lemma:more-general-contexts}, we have
  $\typ{P_2\mid \theta(\Gamma_2,\Gamma), x:\theta_1'\sigma'}{M:C}$,
  which further implies
  \refd{$\typ{P_2\mid\theta_3\theta_1'\theta_1(\Gamma_2,\Gamma,x:\sigma')}{M:C}$}
  for some $\theta_3$ with $\theta = \theta_3\theta_1'\theta_1$. %
  By the IH on \refd{}, we have
  $\infer{\theta_1(\Gamma_2,\Gamma,x:\sigma')}{M}{C'}{\theta_2}{P_2'}{\Sigma_2}$
  and there exists $\theta_2'$ such that $C=\theta_2' C'$,
  $P_2\vdash\theta_2' P_2'$, and
  $(\theta_3\theta_1'=\theta_2'\theta_2)|_{\Gamma_2,\Gamma}$.
  By context weakening and $\theta_1\sigma'=\sigma'$, we have
  \refe{$\infer{\theta_1(\Gamma_1,\Gamma_2,\Gamma),x:\sigma'}{M}{C'}{\theta_2}{P_2'}{\Sigma_2}$}.
  Let $Q = \makeunl{\theta_2\theta_1(\Gamma_1,\Gamma_2,\Gamma) |_{\Sigma_1 \cap \Sigma_2}} \cup
           \makeunl{\theta_2(x:\sigma) |_{\setcomplement{\Sigma_2}}}$.
  By \tilab{Let}, \refc{} and \refe{}, we have
  $$\infer{\Gamma_1,\Gamma_2,\Gamma}{\Let x=V\In M}{C'}
  {\theta_2\theta_1}{P_2'\cup Q}{\Sigma_1\cup(\Sigma_2\backslash x)}$$

  With $\theta' = \theta_1'\theta_2'$, we have
  $(\theta=\theta'\theta_2\theta_1)|_{\Gamma_1,\Gamma_2,\Gamma_3}$.
  By $\Sigma_1\cap\Sigma_2\subseteq \dom{\Gamma}$,
  \Cref{lemma:zero-is-unlimited}, \Cref{lemma:closure-factor} and
  \Cref{thm:correctness-of-factorisation}, there exists $\theta_p$
  such that \reff{$P_2\vdash\theta_p \theta' Q$}.
  Let $\theta'' = \theta_p\theta'$.
  Our goal follows from
  $(\theta=\theta''\theta_2\theta_1)|_{\Gamma_1,\Gamma_2,\Gamma_3}$,
  $C= \theta'' C'$, and $P_2\vdash\theta'' (P_2'\cup Q)$.

  \item[Case]
  \begin{mathpar}
  \inferrule*[Lab=\qtylab{Abs}]
  { \refa{\typ{P\mid\theta\Gamma, x : A}{M : C}}\\
    P \vdash \theta\Gamma \linle Y }
  {\typ{P\mid\theta\Gamma}{\lambda x . M : A \to^Y C}}
  \end{mathpar}
  Take a fresh variable $\alpha$ and let $\theta_1 =
  \theta[A/\alpha]$. By \refa{}, we have
  $\refb{\typ{P\mid\theta_1(\Gamma, x : \alpha)}{M : C}}$.
  By the IH on \refb{}, we have
  \refc{$\infer{\Gamma,x:\alpha}{M}{C'}{\theta'}{P'}{\Sigma}$} and there
  exists $\theta''$ such that $C=\theta''C'$, $P\vdash\theta''P'$, and
  $(\theta_1 = \theta''\theta')|_{\Gamma,x:\alpha}$.
  Let $Q = \makeleq{\theta'\Gamma |_{\Sigma}, \phi} \cup
           \makeunl{\theta'(x:\alpha) |_{\setcomplement{\Sigma}}}$
  By \tilab{Abs} and  \refc{}, taking a fresh variable $\phi$, we have
  $$
  \infer{\Gamma}{\lambda x . M}{\theta'\alpha \to^\phi C'}
  {\theta'}{P' \cup Q}{\Sigma \backslash x}
  $$
  With $\theta_2 = \theta''[Y/\phi]$, we have $(\theta =
  \theta_2\theta')|_{\Gamma,x:\alpha}$.
  By $P\vdash\theta\Gamma\linle Y$, \Cref{lemma:zero-is-unlimited},
  \Cref{lemma:closure-factor}, and
  \Cref{thm:correctness-of-factorisation}, there exists $\theta_p$
  such that $P\vdash\theta_p\theta_2 Q$.
  Let $\theta_3 = \theta_p \theta_2$.
  Our goal follows from $(\theta =
  \theta_3\theta')|_{\Gamma,x:\alpha}$, $A\to^Y C =
  \theta_3(\theta'\alpha\to^\phi C')$ and $P \vdash \theta_3 (P'\cup Q)$.

  \item[Case]
  \begin{mathpar}
  \inferrule*[Lab=\qtylab{App}]
  { \refa{\typ{P\mid\theta(\Gamma_1, \Gamma)}{V : B \to^Y C}} \\
    \refb{\typ{P\mid\theta(\Gamma_2, \Gamma)}{W : B}} \\
    P \vdash \theta\Gamma\un
  }
  {\typ{P\mid\theta(\Gamma_1, \Gamma_2, \Gamma)}{V\,W : C}}
  \end{mathpar}
  By the IH on \refa{}, we have
  \refc{$\infer{\Gamma_1,\Gamma}{V}{A'}{\theta_1}{P_1}{\Sigma_1}$} and there
  exists $\theta_1'$ such that $B\to^Y C = \theta_1'A'$,
  $P\vdash\theta_1'P_1$, and $(\theta = \theta_1'\theta_1)|_{\Gamma_1,\Gamma}$.
  Let $\theta = \theta'\theta_1'\theta_1$ where $\theta'$ only
  substitutes type variables only appearing in $\Gamma_2$.
  By the IH on \refb{}, we have
  \refd{$\infer{\theta_1(\Gamma_2,\Gamma)}{W}{B'}{\theta_2}{P_2}{\Sigma_2}$}
  and there exists $\theta_2'$ such that $B = \theta_2'B'$,
  $P\vdash\theta_2'P_2$, and \refe{$(\theta'\theta_1' = \theta_2'\theta_2)|_{\Gamma_2,\Gamma}$}.
  Take fresh variables $\alpha,\mu,\phi$. By $B\to^Y C = \theta_1'A'$,
  the unification $\unify{\theta_2 A'}{B'\to^\phi\alpha\eff\mu}{\theta_3}$ succeeds.
  By \Cref{lemma:principal-unifiers} and \refe{}, there exists
  $\theta_4$ such that $\theta_4\theta_3(\theta_2 A') =
  \theta_4\theta_3(B'\to^\phi\alpha\eff\mu) = B\to^Y C$.
  Let $P_3 = \theta_3 (\theta_2 P_1 \cup P_2)$ and $Q =
  \makeunl{\theta_3\theta_2\theta_1(\Gamma_1,\Gamma_2,\Gamma) |_{\Sigma_1 \cap \Sigma_2}}$.
  By \tilab{App}, \refc{}, \refd{} and context weakening, we have
  $\infer{\Gamma_1,\Gamma_2,\Gamma}{V\,W}{\theta_3(\alpha\eff\mu)}
  {\theta_3\theta_2\theta_1}{P_3\cup Q}{\Sigma_1\cup\Sigma_2}$.
  With $\theta''=\theta_4\theta_2'\theta_1'$, we have
  $(\theta=\theta''\theta_3\theta_2\theta_1)|_{\Gamma_1,\Gamma_2,\Gamma}$.
  By $\Sigma_1\cap\Sigma_2 \subseteq \dom{\Gamma}$,
  $P\vdash\theta\Gamma\linle\Unl$, \Cref{lemma:closure-factor}, and
  \Cref{thm:correctness-of-factorisation}, we have
  $P\vdash\theta_p\theta'' Q$.
  Let $\theta_5 = \theta_p\theta''$.
  Our goal follows from $C=\theta_5{\theta_3(\alpha\eff\mu)}$,
  $P\vdash\theta_5 (P_3\cup Q)$ and
  $(\theta=\theta_5\theta_3\theta_2\theta_1)|_{\Gamma_1,\Gamma_2,\Gamma}$.

  \item[Case]
  \begin{mathpar}
    \inferrule*[Lab=\qtylab{Return}]
      {\refa{\typ{P\mid\theta\Gamma}{V : A}}}
      {\typc{P\mid\theta\Gamma}{\Ret V : A}{\{R\}}}
  \end{mathpar}
  Our goal follows from the IH on \refa{}.

  \item[Case]
  \begin{mathpar}
    \inferrule*[Lab=\qtylab{Do}]
      {
        \refa{\typ{P\mid\theta\Gamma}{V : A_\ell}} \\
        P \vdash \{\ell\lsig A_\ell\sto^Y B_\ell\} \rle R
      }
      {\typc{P\mid\theta\Gamma}{\Do \ell \; V : B_\ell}{\{E\}}}
  \end{mathpar}
  Similar to previous cases.
  Our goal follows from the IH on \refa{},
  \Cref{lemma:principal-unifiers}, and
  \Cref{thm:correctness-of-factorisation}.

  \item[Case]
  \begin{mathpar}
  \inferrule*[Lab=\qtylab{Seq}]
  { \refa{\typc{P\mid\theta(\Gamma_1, \Gamma)}{M : A}{\{R_1\}}} \\
    \refb{\typc{P\mid\theta(\Gamma_2, \Gamma), x : A}{N : B}{\{R_2\}}} \\\\
    P \vdash R_1 \rle R \\ P \vdash R_2 \rle R \\
    P \vdash \theta\Gamma_2 \linle R_1 \\
    P \vdash \theta\Gamma \un
  }
  {\typc{P\mid\theta(\Gamma_1,\Gamma_2,\Gamma)}{\Let x \revto M \In N : B}{\{R\}}}
  \end{mathpar}
  By the IH on \refa{}, we have
  \refd{$\infer{\Gamma_1,\Gamma}{M}{A'\eff\{R_1'\}}{\theta_1}{P_1}{\Sigma_1}$}
  and there exists $\theta_1'$ such that $A\eff\{R_1\} =
  \theta_1'(A'\eff\{R_1'\})$, $P\vdash\theta_1' P_1$, and
  $(\theta=\theta_1'\theta_1)|_{\Gamma_1,\Gamma}$.
  Let $\theta=\theta'\theta_1'\theta_1$ where $\theta'$ substitutes
  type variables only appearing in $\Gamma_2$.
  By \refb{}, we have
  \refc{$\typc{P\mid\theta'\theta_1'\theta_1(\Gamma_2, \Gamma, x :
  A')}{N : B}{\{R_2\}}$}.
  By the IH on \refc{}, we have
  \refe{$\infer{\theta_1(\Gamma_2,\Gamma,x:A')}{N}{B'\eff\{R_2'\}}{\theta_2}{P_2}{\Sigma_2}$}
  and there exists $\theta_2'$ such that $B\eff\{R_2\}=\theta_2'(B'\eff\{R_2'\})$,
  $P\vdash\theta_2'P_2$ and $(\theta'\theta_1' = \theta_2'\theta_2)|_{\Gamma_2,\Gamma}$.
  Take a fresh variable $\mu$.
  Let $Q = \makeunl{\theta_2\theta_1(\Gamma_1,\Gamma_2,\Gamma) |_{\Sigma_1 \cap \Sigma_2}} \cup
  \makeunl{\theta_2(x:A) |_{\setcomplement{\Sigma_2}}} \cup
  \makeleq{\theta_2\theta_1(\Gamma_1,\Gamma_2,\Gamma)|_{\Sigma_2}, \theta_2 R_1} \cup
  \makesub{\theta_2 R_1, \mu} \cup
  \makesub{R_2, \mu}$.
  By \tilab{Seq}, \refd{}, \refe{}, and context weakening, we have
  $\infer{\Gamma_1,\Gamma_2,\Gamma}{\Let x\revto M\In N}{B'\eff\{R_2\}}
  {\theta_2\theta_1}{\theta_2P_1\cup P_2\cup Q}{\Sigma_1\cup(\Sigma_2\backslash x)}$.
  With $\theta'' = [R/\mu]\theta_2'\theta_1'$, we have $(\theta =
  \theta''\theta_2\theta_1)|_{\Gamma_1,\Gamma_2,\Gamma}$.
  By $\Sigma_1\cap\Sigma_2\subseteq\dom{\Gamma}$,
  $P\vdash\theta\Gamma\linle\Unl$, \Cref{lemma:zero-is-unlimited},
  $\Gamma_2 = \Gamma|_{\Sigma_2}$, $P \vdash \theta\Gamma_2 \linle
  R_1$, $P\vdash R_1\rle R$, $P\vdash R_2\rle R$,
  \Cref{lemma:closure-factor} and
  \Cref{thm:correctness-of-factorisation}, there exists $\theta_p$
  such that $P\vdash\theta_p\theta'' Q$.
  Let $\theta_3 = \theta_p\theta''$.
  Our goal follows from $B\eff\{R_2\} = \theta_3({B'\eff\{\mu\}})$, $P
  \vdash \theta_3(\theta_2 P_1 \cup P_2)$, and $(\theta =
  \theta_3\theta_2\theta_1)|_{\Gamma_1,\Gamma_2,\Gamma}$.

  \item[Case]
  \begin{mathpar}
  \inferrule*[Lab=\qtylab{Handle}]
  {
    \refa{\typ{P\mid\theta(\Gamma_1,\Gamma)}{H : A\eff\{R_1\} \hto D}} \\
    \refb{\typc{P\mid\theta(\Gamma_2, \Gamma)}{M : A}{\{R\}}} \\\\
    P \vdash \theta\Gamma \un \\ \refc{P \vdash R \rle R_1}
  }
  {P\mid\theta(\Gamma_1, \Gamma_2,\Gamma)\vdash \Handle M \With H : D}
  \end{mathpar}
  By a similar proof to the \qtylab{App} case, our goal follows from
  the IHs on \refa{} and \refb{}, \Cref{lemma:closure-factor},
  \Cref{thm:correctness-of-factorisation}, and
  \Cref{lemma:principal-unifiers}.
  The only difference is the subtyping constraint $\makesub{\theta_3
  R, \theta_3\theta_2 R_1}$ used by \tilab{Handle}, which follows from
  \refc{}.

  \item[Case]
  \begin{mathpar}
  \inferrule*[Lab=\qtylab{Handler}]
  {
    C = A \eff \{(\ell_i : A_{\ell_i} \sto^{Y_i} B_i)_i; R_1\} \\
    D = B \eff \{R_2\} \\\\
    \refa{\typ{P\mid\theta\Gamma, x : A}{M : D}}\\
    \refb{[\typ{P\mid\theta\Gamma, p_i : A_i, r_i : B_i \to^{Y_i} D}{N_i : D}]_i} \\\\
    P \vdash \theta\Gamma \un \\
    P \vdash R_1 \rle R_2 \\
    P \vdash R_1 \rlack \{\ell_i\}_i
  }
  {\typ{P\mid\theta\Gamma}{\{\Ret x \mapsto M\} \uplus \{ \ell_i\;p_i\;r_i \mapsto N_i \}_{i=1}^n : C \hto D}}
  \end{mathpar}
  The typing rule for handler is the most complicated one, but there
  is actually nothing new of the proof compared to previous cases for
  other rules.
  For each typing derivation on the handler clauses, we do a similar
  proof to the \qtylab{Abs} case.
  Take fresh variables $\alpha, \phi_i$, and $\mu$.

  First, by \refa{} we have $\typ{P\mid\theta[A/\alpha](\Gamma, x :
  \alpha)}{M : D}$. By the IH on it, we have
  \refc{$\infer{\Gamma,x:\alpha}{M}{D'}{\theta_0}{P_0}{\Sigma_0}$} and there
  exists $\theta_0'$ such that $D=\theta_0'D'$, $P\vdash \theta_0'
  P_0$ and $(\theta[A/\alpha] =
  \theta_0'\theta_0)|_{\Gamma,x:\alpha}$.
  Let $\theta_0^a = \theta_0$ and $\theta_0^b = \theta_0'$.
  We have $(\theta_0^b \theta_0^a = \theta)|_\Gamma$.

  By the typing derivation on the first handler clause in \refb{}, we
  have $\typ{P\mid \theta_0^b[Y_1/\phi_1]\theta_0^a(\Gamma, p_1:A_1,
  r_1:B_1\to^{\phi_1}D)}{N_1: D}$.
  By the IH on it, we have $\infer{\theta_0^a (\Gamma, p_1:A_1,
  r_1:B_1\to^{\phi_1}D)}{N_1}{D_1}{\theta_1}{P_1}{\Sigma_1}$ and
  $\theta_1'$ such that $D = \theta_1'D_1$, $P\vdash\theta_1'P_1$ and
  $(\theta_0^b[Y_1/\phi_1]=\theta_1'\theta_1)|_{(\Gamma, p_1:A_1,
  r_1:B_1\to^{\phi_1}D)}$.
  By $D = \theta_1'D_1$, the unification
  $\unify{D_1}{\theta_1'\theta_1 D'}{\theta_1^x}$ succeeds. By
  \Cref{lemma:principal-unifiers}, there exists $\theta_1^y$ such that
  $\theta_1^y D_1 = D$.
  Set $\theta_1^a = \theta_x\theta_1\theta_0^a$ and $\theta_1^b
  =\theta_1'\theta_1^y$. We have $(\theta_1^b \theta_1^a =
  \theta)|_\Gamma$.

  Repeating the above process for every $i$ from $2$ to $n$, we have
  \refd{$\infer{\theta_{i-1}^a (\Gamma, p_i:A_i,
  r_i:B_i\to^{\phi_i}D)}{N_i}{D_i}{\theta_i}{P_i}{\Sigma_i}$} and
  $(\theta_i^b \theta_i^a = \theta)|_\Gamma$.
  Let
  \[\bl
  C' = \theta_n^a (\alpha\eff\{(\ell_i : A_{\ell_i} \sto^{\phi_i} B_{\ell_i})_i\lsep\mu\}) \\
  B'\eff\{R\} = \theta_n^a D' \\
  \Sigma = (\Sigma_0 \backslash \{x\}) \cup
            (\cup_{i=1}^{n} (\Sigma_i \backslash \{p_i, r_i\})) \\
  P' = (\cup_{i=0}^{n} \theta_n^a P_i)
    \cup \makeunl{\theta_n^a\Gamma |_{\Sigma}}
    \cup \makesub{\mu, R}
    \cup \makelack{\mu, \{\ell_i\}_i} \\
  Q' = \makeunl{\theta_n^a(x:\alpha)|_{\setcomplement{\Sigma_0}}}
    \cup (\cup_{i=1}^n\makeunl{\theta_n^a(p_i:A_{\ell_i},r_i:B_{\ell_i}\to^{\phi_i}D)})
  \el\]
  By \tilab{Handler}, \refc{}, and \refd{}, we have
  $\infer{\Gamma}{\{\Ret x \mapsto M\} \uplus \{ \ell_i\;p_i\;r_i
  \mapsto N_i \}_{i=1}^n}{C'\hto\theta_n^a D'}{\theta_n^a}{P'\cup Q'}{\Sigma} $.
  With $\theta' = \theta_n^b[R_1/\mu]$,
  we have $(\theta = \theta'\theta_n^a)|_\Gamma$.
  By \Cref{lemma:zero-is-unlimited},  \Cref{lemma:closure-factor}, and
  \Cref{thm:correctness-of-factorisation} there exists $\theta_p$ such
  that $P\vdash\theta_p\theta'(P\cup Q)$.
  Let $\theta'' = \theta_p\theta'$.
  Our goal follows from $C\hto D = \theta'' (C'\hto\theta_n^a D')$,
  $P\vdash\theta''(P\cup Q)$, and $(\theta =
  \theta''\theta_n^a)|_\Gamma$.
\end{description}
\end{proof}

\subsection{Correctness of Constraint Solving}
\label{app:proof-constraint-solving}

\begin{lemma}%
  \label{lemma:auxiliary-unification}
  If $\unifyrow(\cR_1,\cR_2)$ returns $(\cR_1',\cR_2',\theta)$, then
  $\satsubst{\cR_1\rle\cR_2} = \satsubst{\cR_1'\rle\cR_2'}\theta$; if
  it fails, then $\cR_1\rle\cR_2$ is not satisfiable.
\end{lemma}
\begin{proof}
  By \Cref{lemma:principal-auxiliary-unification}, the substitution
  $\theta$ returned by $\unifyrow(\cR_1,\cR_2)$ is the principal
  unifier that unifies the linearity types of the same labels in
  $\cR_1$ and $\cR_2$, which is a necessary condition for any solution
  of $\cR_1 \rle \cR_2$.
\end{proof}

\begin{lemma}
  \label{lemma:factor-preserves-solutions}
  If $\factor(P) = Q$, then $\satsubst{P} = \satsubst{Q}$.
\end{lemma}

\begin{proof}
  By \Cref{thm:correctness-of-factorisation}, we have $P\vdash Q$ and
  $Q\vdash P$.
  For any $\theta\in\satsubst{P}$, we have $\cdot\vdash\theta P$.
  By the closure property of
  \Cref{thm:properties-of-the-entailment-relation}, we have
  $\theta P\vdash \theta Q$.
  By the transitivity of
  \Cref{thm:properties-of-the-entailment-relation}, we have
  $\cdot\vdash\theta Q$, which implies $\theta\in\satsubst{Q}$.
  Symmetrically, for any $\theta\in\satsubst{Q}$, we can prove
  $\theta\in\satsubst{P}$.
  Finally, we have $\satsubst{P} = \satsubst{Q}$.
\end{proof}

\constraintSolving*

\begin{proof}
  The termination of $\solvelin$ and $\factor$ is obvious.
  It may be not very obvious that $\solverow$ always terminates since the
  $\solverow(\theta,P,Q)$ moves the solved predicates in $P$ to the
  set of unsolved constraints $Q$ in some cases.
  Note that only row subtyping constraints of forms
  $\cR_1\rle\cR_2\lsep\mu_2$ and $\cR_1\lsep\mu_1\rle\cR_2\lsep\mu_2$
  might require resolving previously solved constraints because they
  substitute row variables.
  In both cases, when $\rowset{\cR_1'}\not\subseteq\rowset{\cR_2'}$,
  we substitute $\mu_2$ with $(\cR_1'\backslash\cR_2')\lsep\mu$.
  Notice that the number of labels used in the whole predicate set is
  finite, and the $\solverow$ fails when there are duplicated labels
  in the same row, which implies that this kind of substitution
  terminates.
  Finally, we can conclude that $\solverow$ terminates.

  For the correctness, the idea is to show that every step preserves
  solutions. We first show $\solverow$ preserves solutions by proving
  a lemma that if $\solverow(\theta,P,Q)$ returns $(\theta'\theta,
  Q')$, then we have $\satsubst{P\cup Q} = \satsubst{Q'} \theta'$; if
  it fails, then $P\cup Q$ is not satisfiable.
  We prove by induction on the definition of $\solverow$.
  \begin{description}
    \item[Case]
    \[\bl
    \solverow(\theta,P,\cdot) = \mreturn\ (\theta, P)
    \el\]
    Our goal follows from $\satsubst{P} = \satsubst{P}\iota$.

    \item[Case]
    \[\bl
    \solverow(\theta,P,(\tau_1\linle\tau_2, Q)) =
      \refa{\solverow(\theta,(P,\tau_1\linle\tau_2), Q)}
    \el\]
    Our goal follows from the IH on \refa{} and
    $\satsubst{P\cup(\tau_1\linle\tau_2, Q)} =
    \satsubst{(P,\tau_1\linle\tau_2)\cup Q}$.

    \item[Case]
    \[\bl
    \solverow(\theta, P, (\cR_1 \rle \cR_2, Q)) = \\
    \quad\bl
      \refb{\mlet\ (\cR_1', \cR_2', \theta') = \unifyrow(\cR_1, \cR_2)} \\
      \refc{\massert\ \rowset{\cR_1'} \subseteq \rowset{\cR_2'}} \\
      \refa{\solverow(\theta'\theta, \theta'P, \theta'Q)}
      \el
    \el\]
    Obviously $\refc{}$ fails when $\cR_1'\rle\cR_2'$ is not satisfiable.
    Our goal follows from the IH on \refa{},
    \Cref{lemma:auxiliary-unification} on \refb{}, and
    $\satsubst{P\cup(\cR_1 \rle \cR_2, Q)} =
    \satsubst{\theta'P\cup\theta'Q}\theta'$.

    \item[Case]
    \[\bl
    \solverow(\theta, P, (\cR_1 \lsep \mu \rle \cR_2 \lsep \mu, Q)) = \\
    \quad\bl
      \refb{\mlet\ (\cR_1', \cR_2', \theta') = \unifyrow(\cR_1, \cR_2)} \\
      \refc{\massert\ \rowset{\cR_1'} \subseteq \rowset{\cR_2'}} \\
      \refa{\solverow(\theta'\theta, \theta' P, \theta' Q)}
      \el
    \el\]
    Obviously $\refc{}$ fails when $\cR_1'\rle\cR_2'$ is not satisfiable.
    Our goal follows from the IH on \refa{},
    \Cref{lemma:auxiliary-unification} on \refb{}, and
    $\satsubst{P\cup(\cR_1\lsep\mu \rle \cR_2\lsep\mu, Q)} =
    \satsubst{\theta'P\cup\theta'Q}\theta'$.

    \item[Case]
    \[\bl
    \solverow(\theta, P, (\cR_1 \lsep \mu \rle \cR_2, Q)) = \\
    \quad\bl
      \refb{\mlet\ (\cR_1', \cR_2', \theta') = \unifyrow(\cR_1, \cR_2)} \\
      \refc{\massert\ \rowset{\cR_1'} \subseteq \rowset{\cR_2'}} \\
      \refa{\solverow(\theta'\theta,
        (\theta'P, \mu \rle (\cR_2'\backslash \cR_1')), \theta' Q)}
      \el
    \el\]
    Obviously $\refc{}$ fails when $\cR_1'\rle\cR_2'$ is not satisfiable.
    Our goal follows from the IH on \refa{},
    \Cref{lemma:auxiliary-unification} on \refb{}, and
    $\satsubst{P\cup(\cR_1\lsep\mu \rle \cR_2, Q)} =
    \satsubst{(\theta'P,\mu\rle(\cR_2'\backslash \cR_1'))\cup\theta'Q}\theta'$.

    \item[Case]
    \[\bl
    \solverow(\theta, P, (\cR_1 \rle \cR_2 \lsep \mu_2, Q)) = \\
    \quad\bl
      {\mlet\ (\cR_1', \cR_2', \theta') = \unifyrow(\cR_1, \cR_2)} \\
      \mfresh\ \mu \\
      \meta{if}\ \rowset{\cR_1'}\subseteq\rowset{\cR_2'}\\
      \meta{then}\
        \refa{\solverow(\theta'\theta, \theta'P, \theta'Q)} \\
      \meta{else}\
      \bl
        \mlet\ \theta'' = [((\cR_1'\backslash \cR_2')\lsep \mu) / \mu_2]\theta' \\
        \refb{\solverow(\theta''\theta, \cdot, \theta''(Q,P))}
        \el
      \el
    \el\]
    For the true branch, our goal follows from the IH on \refb{},
    \Cref{lemma:principal-auxiliary-unification}, and
    $$\satsubst{P\cup(\cR_1 \rle \cR_2 \lsep \mu_2, Q)} =
    \satsubst{\theta'P\cup\theta' Q}\theta'$$

    For the false branch, our goal follows from the IH on \refb{},
    \Cref{lemma:principal-auxiliary-unification}, and
    $$\satsubst{P\cup(\cR_1 \rle \cR_2 \lsep \mu_2, Q)} =
    \satsubst{\theta''(Q,P)}\theta''$$

    Both of the above equations follow from the fact that in order to
    solve $\cR_1 \rle \cR_2 \lsep \mu_2$, it is necessary to unify the
    linearity types of the same labels in $\cR_1$ and $\cR_2$, and
    instantiate $\mu_2$ with at least other labels only in $\cR_1$ (no
    instantiation needed when
    $\rowset{\cR_1'}\subseteq\rowset{\cR_2'}$).

    \item[Case]
    \[\bl
    \solverow(\theta, P, (\cR_1 \lsep \mu_1 \rle \cR_2 \lsep \mu_2, Q)) = \\
    \quad\bl
      {\mlet\ (\cR_1', \cR_2', \theta') = \unifyrow(\cR_1, \cR_2)} \\
      \mfresh\ \mu \\
      \meta{if}\ \rowset{\cR_1'}\subseteq\rowset{\cR_2'}\\
      \meta{then}\
      \bl
        \refa{\solverow(\theta'\theta,(\theta'P,\mu_1\rle(\cR_2'\backslash\cR_1')\lsep\mu_2),\theta'Q)}
        \el\\
      \meta{else}\
      \bl
        \mlet\ \theta'' = [((\cR_1'\backslash \cR_2')\lsep \mu) / \mu_2]\theta' \\
        \refb{\solverow(\theta''\theta, \mu_1 \rle (\cR_2'\backslash \cR_1')\lsep \mu, \theta''(Q,P))}
        \el
      \el
    \el\]
    For the true branch of if, our goal follows from the IH on
    \refa{}, \Cref{lemma:principal-auxiliary-unification}, and
    $$
    \satsubst{P\cup (\cR_1 \lsep \mu_1 \rle \cR_2 \lsep \mu_2, Q)} =
    \satsubst{(\theta'P,\mu_1\rle(\cR_2'\backslash\cR_1')\lsep\mu_2)\cup\theta'Q)}\theta'
    $$

    For the false branch of if, our goal follows from the IH on
    \refa{}, \Cref{lemma:principal-auxiliary-unification}, and
    $$\satsubst{P\cup(\cR_1\lsep\mu_1 \rle \cR_2 \lsep \mu_2, Q)} =
    \satsubst{(\mu_1 \rle (\cR_2'\backslash \cR_1')\lsep \mu)\cup
      \theta''(Q,P)}\theta''$$

    Both of the above two equations follow from the fact that in
    order to solve $\cR_1\lsep\mu_1 \rle \cR_2 \lsep \mu_2$, it is
    necessary to unify the linearity types of the same labels in
    $\cR_1$ and $\cR_2$, and instantiate $\mu_2$ with at least other
    labels only in $\cR_1$ (no instantiation needed when
    $\rowset{\cR_1'}\subseteq\rowset{\cR_2'}$).

    \item[Case]
    \[\bl
    \solverow(\theta, P, (\cR \rlack \mL, Q)) = \\
    \quad\bl
      \refb{\massert\ \dom{\cR} \cap \mL = \emptyset} \\
      \refa{\solverow(\theta, P, Q)}
      \el
      \medskip \\
    \el\]
    Obviously \refb{} fails when $\cR\rlack\mL$ is not satisfiable.
    Our goal follows the IH on \refa{}.

    \item[Case]
    \[\bl
    \solverow(\theta, P, (\cR\lsep\mu \rlack \mL, Q)) = \\
    \quad\bl
      \refb{\massert\ \dom{\cR} \cap \mL = \emptyset} \\
      \refa{\solverow(\theta, (P, \mu \rlack \mL), Q)}
      \el
    \el\]
    Obviously \refb{} fails when $\cR\lsep\mu\rlack\mL$ is not
    satisfiable. Our goal follows the IH on \refa{}.
  \end{description}

  Then, we can conclude that if $\solverow{\iota,\cdot,P}$ returns
  $(\theta,Q)$, then we have $\satsubst{P}=\satsubst{Q}\theta$; if it
  fails, then $P$ is not satisfiable.
  Moreover, in $Q$, row subtyping constraints are all in the forms of
  $\mu\rle\cR$ and $\mu\rle\cR\lsep\mu'$.

  By \Cref{lemma:factor-preserves-solutions}, we have
  $\satsubst{Q}=\satsubst{\factor(Q)}$.
  Moreover, in $\factor(Q)$, linearity constraints are all in atomic
  forms, which means they are only between type variables, row
  variables, and linearity types $Y$.

  Let $Q'' = \factor(Q)$. For $\solvelin(\cdot,Q'') = Q'$, we want to
  show that $\satsubst{Q'} = \satsubst{\factor(Q)}$. Notice that
  $\solvelin(\cdot,Q'')$ essentially computes the transitive closure
  of the linearity constraints in $Q''$. Obviously we have
  $\satsubst{Q''} \subseteq \satsubst{Q'}$.
  For the other direction, we need to show that for any
  $\{\tau_1\linle\tau_2, \tau_2\linle\tau_3\} \subseteq Q'$ and
  $\theta\in\satsubst{\tau_1\linle\tau_2, \tau_2\linle\tau_3}$, we
  have $\cdot\vdash\theta(\tau_1\linle \tau_3)$.
  Notice that the type inference of \CalcS only generates linearity
  constraints of forms $\Gamma \linle \tau$, which means rows only
  appear on the RHS.
  Thus, after factorisation, $\theta\tau_2$ can only be $A$ or $Y$.
  The $\cdot\vdash\theta(\tau_1\linle \tau_3)$ follows from a
  straightforward case analysis on $\theta\tau_2$.

  Finally, if $\Lin\linle\Unl \in Q'$, then $Q'$ is obviously not
  satisfiable. Otherwise, we have a trivial solution by substituting
  all row variables with the empty row $\cdot$, value variables with
  $\TUnit$, and linearity variables with $\Unl$.
  We also have $\satsubst{P} = \satsubst{Q'}\theta$, which further
  implies the trivial solution of $Q'$ also gives a solution of $P$.
  These results also hold for $Q$ since $\satsubst{Q'} =
  \satsubst{\factor(Q)} = \satsubst{Q}$.
\end{proof}

\end{document}